%% file: 18Oct2026Exact_probability_laws_of_random_inner_products.tex
\documentclass[11pt]{article}
\input{qudit_preamble.tex}

\newcommand{\Dir}{\mathrm{Dir}}
\newcommand{\Beta}{\mathrm{Beta}}
\newcommand{\Var}{\operatorname{Var}}
\newcommand{\Cov}{\operatorname{Cov}}

\begin{document}

%====================================================================================================%
\title{\Large \bf Exact probability laws of random inner products:\\ Dirichlet splines, overlaps, uncertainty, and Bargmann invariants}
%====================================================================================================%

\author{\blue{Lin Zhang}\footnote{E-mail: godyalin@163.com}\\
  {\it\small School of Mathematical Sciences, Hangzhou Dianzi University, Hangzhou 310018, PR~China}}
\date{}%October 09, 2026
\maketitle

\begin{abstract}
Inner products of random matrices and random quantum states---such as orbital inner products, expectation values, overlaps, and cyclic Bargmann invariants---are shown to follow from two underlying mechanisms: the diagonal of an induced random state is Dirichlet distributed, and the expectation value of a random pure state in a rank-one or two-by-two operator is a numerical shadow on an elliptical disk. For states induced by an environment of dimension $K$, the expectation value of any observable is a Curry--Schoenberg spline with each eigenvalue repeated $K$ times, and the vector of moments has an explicit density on a simplex; for a qutrit, the variance law is obtained in finite form. The overlap of two random qubit states attains its mode at its mean for every rotation-invariant ensemble, but for Hilbert--Schmidt qutrits the density at the mean has a strictly negative derivative, so mode and mean separate; for large environments, the centred overlap converges to a variance-gamma law rather than a Gaussian. For two noncommuting qubit observables, the joint law of the two standard deviations is derived for all $K$, and the probability of coming within a small distance of the optimal sum of variances scales as a power of that distance, with a different exponent when the Bloch axes are orthogonal. For three Haar-random qubit states, the exact density of the Bargmann invariant is given in terms of a complete elliptic integral, together with its local behaviour and elementary marginals. In every dimension, the probability that its real part is negative---certifying basis-independent coherence---is computed, along with all mixed moments. For cycles of arbitrary length, the invariant is shown to be driven by an exact three-parameter Markov chain, and the logarithm of its modulus obeys a law of large numbers and a central limit theorem.
\end{abstract}

\newpage\setcounter{tocdepth}{2}{\small\tableofcontents}\newpage

%===================================================%
\section{Introduction}\label{sect:intro}
%===================================================%

A random inner product is a scalar of the form $\Inner{\bsX}{\bsY}=\Tr{\bsX\bsY}$ in which one or both
arguments are random. Four instances occur throughout matrix analysis and quantum information:
\begin{eqnarray}\label{eq:four}
\Tr{\bsA\bsU \bsB\bsU^\dagger},\qquad \Tr{\bsA\rho},\qquad \Tr{\rho\tau},\qquad
\Delta_m=\Tr{\psi_1\psi_2\cdots \psi_m},
\end{eqnarray}
where $\bsU$ is Haar distributed on the unitary group $\rU(d)$; $\rho$ and $\tau$ are random density matrices and
$\psi_j\equiv\proj{\psi_j}$ are projections onto independent Haar-random unit vectors on $\bbC^d$. The first is the
orbital inner product whose Laplace transform is the Harish-Chandra--Itzykson--Zuber (HCIZ)
integral \cite{HarishChandra1957,Itzykson1980,McSwiggen2021}. The second is a random expectation
value, studied for random pure states in \cite{CamposVenuti2013,Gallay2012,Dunkl2011}. The third is
the Hilbert--Schmidt overlap of two random states, whose moments enter the analysis of fidelities
and distances between random states \cite{Zyczkowski2005,Puchala2016,Kumar2020,Laha2021,Laha2023}.
The fourth is the cyclic Bargmann invariant \cite{Bargmann1964,Mukunda1993,Simon1993}, now used
as a basis-independent witness of coherence and imaginarity
\cite{Galvao2020,Oszmaniec2024,Wagner2024,Fernandes2024,Li2025,Li2026}.

For each of these quantities the set of attainable values is known or easy to describe: von
Neumann's trace inequality for $\Tr{\bsA\bsU \bsB\bsU^\dagger}$, the numerical range for expectation
values, the uncertainty region for pairs of variances \cite{Zhang2021,Zhang2022}, and the sets of
Bargmann invariants \cite{Zhang2025,Pratapsi2025,Xu2026,Fernandes2024}. A range does not tell
which values are typical, how often a random state comes close to an extreme value, or how either
answer depends on the dimension; for that one needs the probability law. We compute these laws in
closed form, and almost all of them follow from one of two mechanisms (Table~\ref{tab:mechanisms}).

For random density matrices and their moments we rely on
\cite{Hall1998,Zyczkowski2001,Sommers2004,Zyczkowski2011,Collins2006,Collins2016}. Spline formulas for spectral
distributions, obtained from HCIZ-type integrals and the derivative principle, were derived in
\cite{Zhang2019DH,Zhang2020Horn,Zhang2026Marginal}, and densities of other quantities of random states in
\cite{Zhang2017,Zhang2018Spectral,Zhang2019Incompat}. Numerical shadows and restricted numerical ranges were introduced
in \cite{Dunkl2011,Gawron2010,Puchala2012}. Variance-based uncertainty relations go back to
\cite{Robertson1929,Schrodinger1930}, and state-independent bounds and uncertainty regions have been studied in
\cite{Huang2012,Maccone2014,Dammeier2015,Schwonnek2017,Zhang2018,Zhang2021,Zhang2022,Zhang2024}. Reference
\cite{Zhang2026Survey} surveys Bargmann invariants.

\begin{table}[h]
\centering\small
\renewcommand{\arraystretch}{1.25}
\begin{tabular}{p{3.3cm}p{5.0cm}p{6.3cm}}
\hline
Mechanism & Statement & Consequences in this paper\\
\hline
Dirichlet & the diagonal of $\rho\sim\mu_{d,K}$ in any basis is $\op{Dirichlet}(K,\ldots,K)$ (Lemma~\ref{lem:dirichlet}) &
expectation values as splines (Thm~\ref{th:spline}); moment simplex (Thm~\ref{th:simplex}); variance law
(Thm~\ref{th:variance}); overlap of qutrits (Thm~\ref{th:qutritmode})\\
numerical shadow & $\Innerm{\psi}{\bsM}{\psi}$ for Haar $\psi$ and $2\times2$ (or rank-one) $\bsM$ is the
projection of a uniform sphere onto an ellipse (Lemma~\ref{lem:archimedes}) &
qubit overlaps (Thm~\ref{th:qubitmode}); qubit uncertainty laws (Thms~\ref{th:ellipse}--\ref{th:nearopt});
Bargmann invariants (Section~\ref{sect:bargmann})\\
\hline
\end{tabular}
\caption{The two mechanisms behind the exact laws of this paper.}\label{tab:mechanisms}
\end{table}

\subsection{Main results}

Throughout, $\mu_{d,K}$ denotes the law of the reduced state of a Haar-random pure state on
$\bbC^d\ot\bbC^K$; $K=1$ gives pure states and $K=d$ the Hilbert--Schmidt measure
\cite{Zyczkowski2001,Sommers2004}.
\begin{enumerate}
\item \emph{Expectation values and moments} (Section~\ref{sect:expect}). For every Hermitian $\bsA\in\Herm(\bbC^d)$
and every $K$, the density of $\Tr{\bsA\rho}$ is the Curry--Schoenberg B-spline whose knots are the
eigenvalues of $\bsA$, each repeated $K$ times (Theorem~\ref{th:spline}). This covers degenerate spectra
without confluent limits and recovers the pure-state results of
\cite{CamposVenuti2013,Gallay2012,Dunkl2011} at $K=1$. The moment vector
$(\langle \bsA\rangle_\rho,\ldots,\langle \bsA^{d-1}\rangle_\rho)$ has an explicit density on a simplex for
every $K$ (Theorem~\ref{th:simplex}), unifying the pure-state and Hilbert--Schmidt cases of
\cite{Zhang2022,Zhang2021}. For a qutrit we derive the law of the variance $\Var_\rho(\bsA)$ in finite
terms (Theorem~\ref{th:variance}).
\item \emph{Overlaps} (Section~\ref{sect:overlap}). The overlap of two random qubit states has its
mode at its mean $1/2$ for every rotation-invariant ensemble (Theorem~\ref{th:qubitmode}). For
Hilbert--Schmidt qutrits we compute the density and its derivative at the mean $1/3$ exactly; the
derivative is negative, so the mode lies strictly below the mean (Theorem~\ref{th:qutritmode}).
For large environments the centred overlap converges to a variance-gamma law
(Theorem~\ref{th:wishart}).
\item \emph{Uncertainty} (Section~\ref{sect:uncertainty}). For two qubit observables with Bloch
axes at angle $\arccos c$ we give the joint law of $(\Delta_\rho(\bsA),\Delta_\rho(\bsB))$ for all $K$
(Theorem~\ref{th:twostd}), where $\Delta_\rho(*)=\sqrt{\Var_\rho(*)}$ is the standard deviation, and the probability of nearly saturating the optimal bound
$\Var_\rho(\bsA)+\Var_\rho(\bsB)\geqslant 1-|c|$. This probability is $\varepsilon^{K-1/2}$ when $c=0$ and
of exact order $\varepsilon^{K}$ when $c\neq0$ (Theorem~\ref{th:nearopt}): the exponent records whether the
optimal states form a circle or two points.
\item \emph{Bargmann invariants} (Section~\ref{sect:bargmann}). Conditional on all but one state,
$\Delta_m$ is a numerical shadow on an explicit ellipse (Proposition~\ref{prop:conditional}). For
three qubit states, we obtain the joint density of $\Delta_3$ via a complete elliptic integral, its
logarithmic singularity $\frac6\pi\log\frac1{|z|}$ at the origin and its boundary value
(Theorem~\ref{th:b3density}, Proposition~\ref{prop:local}), and the elementary marginals of
the real part $\Re(\Delta_3)$ and the imaginary part $\Im(\Delta_3)$ (Theorem~\ref{th:marginals}). In every dimension $d$, we compute
$$
\Pr\set{\Re(\Delta_3)<0}=\frac12-\frac{\sqrt\pi\Gamma(d)}{4\Gamma(d+\frac12)}
$$
(Theorem~\ref{th:negative}) and all mixed moments of $\Delta_3$ (Theorem~\ref{th:moments}). For
cycles of any length $m$, $\Delta_m$ is the second coordinate of a three-real-parameter Markov chain
(Theorem~\ref{th:markov}), and $\log|\Delta_m|$ satisfies a law of large numbers and a central limit
theorem with explicit constants (Theorem~\ref{th:longcycle}).
\end{enumerate}
The figures compare each closed form with direct Monte Carlo sampling, and Table~\ref{tab:moments}
lists exact and sampled moments of $\Delta_3$.

\subsection{Relation to previous work}

Several results in this paper generalize, unify, or reprove known statements; 
others are new. To avoid overclaiming, we summarize the status of each main 
result here.

The \emph{Dirichlet mechanism} (Lemma~\ref{lem:dirichlet}) is classical: the 
diagonal of an induced random state being Dirichlet distributed is implicit in 
the Ginibre representation of \(\mu_{d,K}\) and is used throughout the 
random-matrix literature 
\cite{Zyczkowski2001,Sommers2004,Zyczkowski2011}. Its aggregation property is 
standard. What is new is not the mechanism itself but its systematic 
exploitation: the observation that it converts \emph{every} expectation value 
into a Curry--Schoenberg spline with repeated knots (Theorem~\ref{th:spline}) 
and \emph{every} moment vector into a Dirichlet law on a simplex 
(Theorem~\ref{th:simplex}). The spline representation for $K=1$ is the 
numerical measure of a Hermitian matrix 
\cite{Gallay2012,Dunkl2011,CamposVenuti2013}; the general-$K$ statement and 
the moment-simplex density are new, but the method is a direct reduction to 
Lemma~\ref{lem:curry}.

The \emph{numerical-shadow mechanism} (Lemma~\ref{lem:archimedes}) is also known: 
the projection of a Haar-random vector onto a $2$D subspace being 
$\Beta(2,d-2)$, and the numerical shadow of a $2\times 2$ 
matrix being the elliptical density \eqref{eq:shadow2}, are due to 
\cite{Dunkl2011} and the references therein. The new content is the 
\emph{assembly}: the qubit uncertainty laws for all $K$ 
(Theorem~\ref{th:twostd}), the exponent dichotomy for near-optimal 
variances (Theorem~\ref{th:nearopt}), the conditional Bargmann laws 
(Proposition~\ref{prop:conditional}), and the exact densities and limit 
theorems for Bargmann invariants 
(Theorems~\ref{th:b3density}--\ref{th:recursion}).

Several results in Sections~\ref{sect:overlap} and~\ref{sect:uncertainty} 
recover or extend the author's earlier work. The qubit overlap law 
(Theorem~\ref{th:qubitmode}) and its corollary \eqref{eq:B1}
are consistent with 
\cite{Zyczkowski2005,Puchala2016,Kumar2020,Laha2021,Laha2023}; the qutrit 
mode separation (Theorem~\ref{th:qutritmode}) and the variance-gamma limit 
(Theorem~\ref{th:wishart}) are new. The joint uncertainty density 
\eqref{eq:U4} recovers the pure-state ($K=1$) and 
Hilbert--Schmidt ($K=2$) cases of \cite{Zhang2021,Zhang2022} and extends 
them to all $K$; the near-optimal probability \eqref{eq:nearopt} is new. 
The Bargmann section is the most novel: Proposition~\ref{prop:conditional} is 
a restatement of known facts about numerical ranges of rank-one matrices 
combined with Lemma~\ref{lem:archimedes} (ii), but 
Theorems~\ref{th:b3density}, \ref{th:marginals}, \ref{th:negative}, 
\ref{th:markov}, and \ref{th:longcycle} are, to our knowledge, new.

Table~\ref{tab:status} records this assessment. "New" means we are not aware 
of a prior statement in this generality; "generalizes" means the result 
extends a known statement to a larger parameter range; "unifies" means the 
result collects two or more known cases under one formula; "reproves" means 
the result is essentially known and the contribution is expository.

\begin{table}[h]
\centering
\scriptsize
\begin{tabular}{p{5.2cm}p{2.2cm}p{3.2cm}p{4.5cm}}
\hline
Result & Status & Prior work & Nature of contribution \\
\hline
Lemma~\ref{lem:dirichlet} (Dirichlet diagonal) & Known & 
\cite{Zyczkowski2001,Sommers2004,Zyczkowski2011} & Standard; aggregation 
property is classical \\
Lemma~\ref{lem:curry} (B-spline) & Known & 
\cite{Curry1966,Micchelli1980,Karlin1986,Carlson1991} & Standard; proof 
included for completeness \\
Lemma~\ref{lem:archimedes} (numerical shadow) & Known & 
\cite{Dunkl2011,Gawron2010,Puchala2012} & Standard; proof included for 
completeness \\
Theorem~\ref{th:orbital} (orbital inner product) & Known & 
\cite{HarishChandra1957,Itzykson1980,McSwiggen2021} & HCIZ inversion; proof 
included for completeness \\
Theorem~\ref{th:spline} (expectation values are splines) & Generalizes & 
\cite{Gallay2012,Dunkl2011,CamposVenuti2013} ($K=1$) & Extends numerical 
measure to arbitrary \(K\); method is reduction to Lemma~\ref{lem:curry} \\
Theorem~\ref{th:simplex} (moment simplex) & Unifies & 
\cite{Zhang2021,Zhang2022} & Collects pure-state and Hilbert--Schmidt cases 
under one Dirichlet density \\
Theorem~\ref{th:variance} (qutrit variance law) & New & --- & Finite-form 
density; no prior statement \\
Theorem~\ref{th:qubitmode} (qubit overlap: mode at mean) & Generalizes & 
\cite{Zyczkowski2005,Puchala2016} & Extends to every rotation-invariant 
ensemble \\
Theorem~\ref{th:qutritmode} (qutrit overlap: mode separates) & New & --- & 
Exact density and derivative at the mean \\
Theorem~\ref{th:wishart} (variance-gamma limit) & New & --- & Non-Gaussian 
limit for large environments \\
Theorem~\ref{th:twostd} (joint uncertainty law) & Generalizes & 
\cite{Zhang2021,Zhang2022} ($K=1,2$) & One formula for all $K$ \\
Theorem~\ref{th:nearopt} (near-optimal exponent) & New & --- & Exponent 
dichotomy $K-\frac12$ vs.\ $K$ \\
Proposition~\ref{prop:conditional} (conditional Bargmann laws) & Reproves & 
\cite{Dunkl2011,Zhang2025} & Numerical-shadow restatement; assembly is new \\
Theorem~\ref{th:b3density} (Bargmann density, $d=2$) & New & --- & Complete 
elliptic integral; local behaviour \\
Theorem~\ref{th:marginals} (marginals) & New & --- & Elementary marginals of 
$\Re(\Delta_3),\Im(\Delta_3)$ \\
Theorem~\ref{th:negative} (negative real part) & New & --- & Exact probability 
in every dimension \\
Theorem~\ref{th:moments} (moments) & New & --- & All mixed moments \\
Theorem~\ref{th:markov} (Markov chain) & New & --- & Three-parameter exact 
sampler \\
Theorem~\ref{th:longcycle} (long cycles) & New & --- & LLN and CLT for 
$\log|\Delta_m|$ \\
Theorem~\ref{th:recursion} (density recursion) & New & --- & Exact integral 
representation for all $n$ \\
\hline
\end{tabular}
\caption{\textbf{Status of the main results.} "New" means we are not aware of a prior 
statement in this generality; "generalizes" means the result extends a known 
statement to a larger parameter range; "unifies" means the result collects 
two or more known cases under one formula; "reproves" means the result is 
essentially known and the contribution is expository.}
\label{tab:status}
\end{table}

The table is deliberately conservative. In several cases 
(e.g., Theorem~\ref{th:spline}, Theorem~\ref{th:twostd}) the 
\emph{statement} is new but the \emph{method} is a direct application of a 
known mechanism; we have marked these as "generalizes" rather than "new" 
to reflect that the conceptual advance is modest. Conversely, in the Bargmann 
section the methods are also new, and those results are marked "new" without 
qualification.

The paper is organized as follows. Section~\ref{sect:prelim} fixes the ensembles and proves the two
mechanisms. Sections~\ref{sect:expect}--\ref{sect:bargmann} treat the four classes of random inner
products in turn. Section~\ref{sect:conclusion} summarizes the results and lists open problems. The
appendices contain computations used in the proofs.

\begin{figure}[t]
\centering
\includegraphics{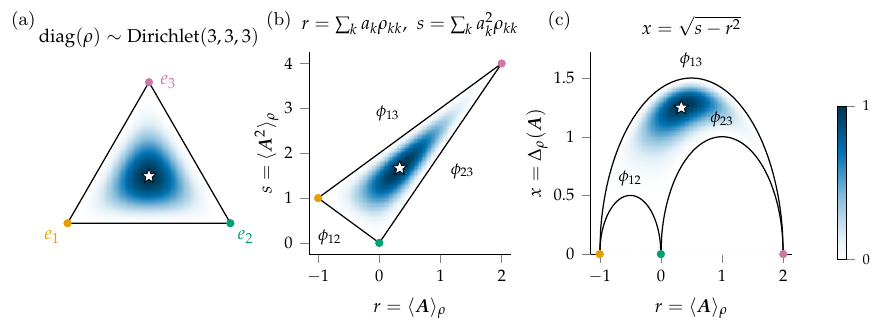}
\caption{The Dirichlet mechanism for $\bsA=\diag(-1,0,2)$ and a Hilbert--Schmidt qutrit
$\rho\sim\mu_{3,3}$. (a) The diagonal of $\rho$ is $\op{Dirichlet}(3,3,3)$ on the simplex. (b) The affine map
$(\rho_{kk})\mapsto(\langle \bsA\rangle_\rho,\langle \bsA^2\rangle_\rho)$ carries it onto the triangle with
vertices $(a_k,a_k^2)$; the density is the same function up to the constant Jacobian
$|V_3(\bsa)|$ (Theorem~\ref{th:simplex}). (c) The substitution $s=r^2+x^2$ turns the triangle into the
curved region of attainable (mean, standard deviation) pairs; each edge $\phi_{ij}$ of (b) becomes the
arc with the same label. Colour: density divided by its maximum. Coloured dots: eigenstates
$\bse_k$; star: the maximally mixed state, tracked through the three panels.}\label{fig:mechanism}
\end{figure}

%===================================================%
\section{Ensembles and the two mechanisms}\label{sect:prelim}
%===================================================%

\subsection{Induced random states}

Let $\Herm(\bbC^d)$ be the real space of $d\times d$ Hermitian matrices with
$\Inner{\bsX}{\bsY}=\Tr{\bsX\bsY}$, and let $\rD(\bbC^d)$ be the set of density matrices. For integers
$d\geqslant1$ and $K\geqslant1$, let $\bsG$ be a $d\times K$ matrix with independent standard complex
Gaussian entries and set
\begin{eqnarray}\label{eq:induced}
\rho=\frac{\bsG\bsG^\dagger}{\Tr{\bsG\bsG^\dagger}} \in\Herm(\bbC^d).
\end{eqnarray}
Its law $\mu_{d,K}$ is the law of $\ptr{\bbC^K}{\proj{\Psi}}$ for a Haar-random unit vector
$\ket{\Psi}\in\bbC^d\ot\bbC^K$ \cite{Zyczkowski2001,Zyczkowski2011}. It is invariant under
$\rho\mapsto \bsV\rho\bsV^\dagger$ for every unitary $\bsV$. For $K\geqslant d$, its eigenvalues
$\lambda_1<\cdots<\lambda_d$ have density proportional to
$[V_d(\bdlambda)]^2\prod^d_{j=1}\lambda_j^{K-d}$ on the ordered simplex, where
$$
V_d(\bsx)=\prod_{1\leqslant i<j\leqslant d}(x_j-x_i).
$$
For $K=d$ (the Hilbert--Schmidt measure), the normalized density on
$\Omega_d=\set{\bdlambda=(\lambda_1,\ldots,\lambda_d):0<\lambda_1<\cdots<\lambda_d,\ \sum^d_{j=1}\lambda_j=1}$ is
\begin{eqnarray}\label{eq:hsorder}
\dif\nu_d(\bdlambda)=\gamma_d[V_d(\bdlambda)]^2\delta\Pa{1-\sum^d_{j=1}\lambda_j}\prod^d_{j=1}\dif\lambda_j,\quad
\gamma_d=\frac{d!\,\Gamma(d^2)}{\prod_{j=1}^{d}\Gamma(j)\Gamma(j+1)},
\end{eqnarray}
so that $\gamma_2=6$ and $\gamma_3=10080$. Integrating over the unordered simplex requires $\frac{\gamma_d}{d!}$ instead.
For $K<d$, the state has rank $K$ almost surely.

We use the following first and second moments. By invariance,
$\bbE[\rho]=\frac{\I_d}d$ and $\bbE[\rho\ot\rho]=\alpha\I_d\ot\I_d+\beta\bsF$ with $\bsF$ the swap; the two constants follow
from $\bbE[\Tr{\rho}^2]=1$ and $\bbE[\Tr{\rho^2}]=\frac{d+K}{dK+1}$ \cite{Zyczkowski2001,Sommers2004}:
\begin{eqnarray}\label{eq:secondmoment}
\bbE[\rho\ot\rho]=\frac{K\I_d\ot\I_d+\bsF}{d(dK+1)}.
\end{eqnarray}

\subsection{The Dirichlet mechanism}

\begin{lem}\label{lem:dirichlet}
Let $\rho\sim\mu_{d,K}$ and let $\set{\ket{\bse_i}:i=1,\ldots,d}$ be any orthonormal basis. Then
$(\rho_{11},\ldots,\rho_{dd})\sim\op{Dirichlet}(K,\ldots,K)$, with density
$$
\frac{\Gamma(dK)}{\Gamma(K)^d}\prod_{i=1}^d x_i^{K-1}
$$
on the simplex $\bbP_{d-1}=\set{\bsx=(x_1,\ldots,x_d)\in\bbR^d_{\geqslant0}:\sum^d_{i=1}x_i=1}$. If $I_1,\ldots,I_\ell$ partition\
$\set{1,\ldots,d}$ with $|I_j|=m_j$, then $\Pa{\sum_{i\in I_1}\rho_{ii},\dots,\sum_{i\in I_\ell}\rho_{ii}}\sim\op{Dirichlet}(Km_1,\ldots,Km_\ell)$.
\end{lem}

\begin{proof}
\begin{enumerate}[(1)]
\item Let $\ket{\Psi}\in\bbC^d\ot\bbC^K$ be Haar-distributed, and define
the reduced state $\rho=\ptr{\bbC^K}{\proj{\Psi}}$. Equivalently,
$\rho=\frac{\bsG\bsG^\dagger}{\Tr{\bsG\bsG^\dagger}}$, where $\bsG$
is an $d\times K$ complex Ginibre matrix. The diagonal entries are
$$
\rho_{ii}=\frac{\sum^K_{j=1}\abs{G_{ij}}^2}{\sum^n_{\ell=1}\sum^k_{j=1}\abs{G_{\ell
j}}^2}.
$$
Define $Y_i=\sum^K_{j=1}\abs{G_{ij}}^2$. The random variables
$Y_1,\ldots,Y_d$ are independent and satisfy $Y_i\sim
\op{Gamma}(K,1)$, whose density function is $\frac1{\Gamma(K)}z^{K-1}e^{-z}\chi_{\set{z>0}}$. The normalized vector
$$
\Pa{\frac{Y_1}{Y_1+\cdots+Y_d},\ldots,\frac{Y_n}{Y_1+\cdots+Y_d}}\sim\op{Dirichlet}(K,\ldots,K).
$$
Thus the density on the probability simplex $\bbP_{K-1}$ is
$$
\frac{\Gamma(dK)}{\Gamma(K)^d}\prod^d_{j=1}x^{K-1}_j.
$$
For the Hilbert-Schmidt ensemble, $K=d$, so
$$
(\rho_{11},\ldots,\rho_{dd})\sim\op{Dirichlet}(K,\ldots,K)=\op{Dirichlet}(d,\ldots,d),
$$
which completes the proof.
\item Now let $I_1,\dots,I_\ell$ be a partition of $\set{1,\dots,d}$ with $|I_j|=m_j$. Define
$Z_j:=\sum_{i\in I_j} Y_i$. Since the $Y_i$ are independent with $Y_i\sim \op{Gamma}(K,1)$, each $Z_j$ is a sum of $m_j$ independent $\op{Gamma}(K,1)$ variables, hence
$Z_j\sim \operatorname{Gamma}(Km_j,1)$. Moreover, $Z_1,\dots,Z_\ell$ are independent because the index sets $I_j$ are disjoint. Also,
$\sum_{j=1}^\ell Z_j=\sum_{i=1}^d Y_i=T$. Therefore
$$
\Pa{\sum_{i\in I_1}\rho_{ii},\dots,\sum_{i\in I_\ell}\rho_{ii}}
=
\Pa{\frac{Z_1}{T},\dots,\frac{Z_\ell}{T}}.
$$
By the same Dirichlet fact, applied to the independent gamma variables $Z_1,\dots,Z_\ell$ with parameters $Km_1,\dots,Km_\ell$,
$$
\Pa{\sum_{i\in I_1}\rho_{ii},\dots,\sum_{i\in I_\ell}\rho_{ii}}
\sim
\op{Dirichlet}(Km_1,\dots,Km_\ell).
$$
\end{enumerate}
This proves the aggregation property. 
\end{proof}

A different Dirichlet structure was recently identified in \cite{Kirby2026}: $\mu_{d,K}$ is the law of a
$\op{Dirichlet}(d,\ldots,d)$-weighted mixture of $K$ independent Haar-random pure states. Lemma~\ref{lem:dirichlet}
concerns the diagonal of $\rho$ rather than its decomposition into pure states, and it is the statement
used below.

\paragraph{Splines.} For $q\geqslant0$, write $x_+^q=x^q\chi_{\set{x>0}}$, where $\chi_S$ is the indicator function of the set $S$ throughout the paper. For a function $f$ and
real numbers $t_1,\ldots,t_N$ (repetitions allowed), $[t_1,\ldots,t_N]f$ denotes the \emph{divided
difference} (see Appendix~\ref{app:divdiff}), with repeated arguments interpreted through derivatives. The following identity is the
probabilistic form of the Curry--Schoenberg B-spline \cite{Curry1966,Micchelli1980,Karlin1986,Carlson1991}.

\begin{lem}\label{lem:curry}
Let $\bsw\sim\op{Dirichlet}(1,\ldots,1)$ on $\bbP_{N-1}$ with $N\geqslant2$, and let $t_1,\ldots,t_N$ be
real numbers that are not all equal. Then $\sum^N_{j=1}w_jt_j$ has the density
$$
M(r|t_1,\ldots,t_N)=(N-1)[t_1,\ldots,t_N](x-r)^{N-2}_+.
$$
It is a piecewise polynomial of degree $N-2$ supported on $[\min_jt_j,\max_jt_j]$, and it is
$C^{N-2-m}$ at a knot of multiplicity $m$ (discontinuous if $m=N-1$).
\end{lem}

\begin{proof}
Let $\bsw=(w_1,\ldots,w_N)\sim \op{Dirichlet}(1,\ldots,1)$ on the simplex
$$
\bbP_{N-1}=\Set{\bsw\in\bbR_{\geqslant0}^N:\sum_{i=1}^N w_i=1}.
$$
Its density with respect to the $(N-1)$-dimensional Lebesgue measure is the constant
$\frac{\Gamma(N)}{\Gamma(1)^N}=(N-1)!$. Define the linear statistic
$S:=\sum_{j=1}^N w_j t_j$. We want the density of $S$.
\begin{enumerate}[(1)]
\item \textbf{Coarea representation.} Let $L(\bsw)=\sum_{j=1}^N w_j t_j$. By the coarea formula, the density of $S$ under the uniform measure on $\bbP_{N-1}$ (with total mass $(N-1)!$) is
$$
f_S(s)=(N-1)!\int_{L^{-1}(s)}\frac{1}{|\nabla L|}\dif H^{N-2},
$$
where $\dif H^{N-2}$ is the $(N-2)$-dimensional Hausdorff measure on the slice $L^{-1}(s)$. A standard formula for the volume of a slice of a simplex gives
$$
\int_{L^{-1}(s)}\frac{1}{|\nabla L|}\dif H^{N-2}=\frac{1}{(N-2)!}\sum_{i=1}^N\frac{(t_i-s)_+^{N-2}}{\prod_{j\neq i}(t_i-t_j)}.
$$
Here $(x)_+=\max(x,0)$. Therefore
$$
f_S(s)=(N-1)\sum_{i=1}^N\frac{(t_i-s)_+^{N-2}}{\prod_{j\ne i}(t_i-t_j)}.
$$
By the definition of the divided difference,
$$
[t_1,\ldots,t_N]g=\sum_{i=1}^N\frac{g(t_i)}{\prod_{j\ne i}(t_i-t_j)},
$$
with $g(x)=(x-s)_+^{N-2}$, we obtain $f_S(s)=(N-1)[t_1,\ldots,t_N](x-s)_+^{N-2}$. This is exactly the stated density $M(s|t_1,\ldots,t_N)=(N-1)[t_1,\ldots,t_N](x-s)_+^{N-2}$.
\item \textbf{Proof of the slice formula.} For completeness, we sketch the induction proof of the slice formula. Let
$$
V_N(s):=\int_{\bbP_{N-1}}
\delta\Pa{s-\sum_{j=1}^N w_jt_j}\delta\Pa{1-\sum^N_{j=1}w_j}[\dif \bsw].
$$
We claim
$$
V_N(s)
=
\frac{1}{(N-2)!}
\sum_{i=1}^N
\frac{(t_i-s)_+^{N-2}}
{\prod_{j\ne i}(t_i-t_j)}.
$$
\begin{itemize}
\item[(a)] \emph{Base $N=2$.}
The simplex is the segment $\Set{\bsw=(w_1,w_2)\in\bbR^2_{\geqslant0}:w_1+w_2=1}$. Then
$S=t_1w_1+t_2(1-w_1)=t_2+(t_1-t_2)w_1$. Thus $S$ is uniform on the interval between $t_1$ and $t_2$, so its density is
$\frac{1}{|t_1-t_2|}$ on $[\min(t_1,t_2),\max(t_1,t_2)]$. The right-hand side for $N=2$ is
$$
\frac{(t_1-s)_+^0}{t_1-t_2} + \frac{(t_2-s)_+^0}{t_2-t_1},
$$
which equals $\frac1{|t_1-t_2|}$ on the interval and $0$ outside. Hence the formula holds.
\item[(b)] \emph{Induction step.}  
Assume the formula holds for $N-1$. Write $w_N=u$ and $w_i=(1-u)v_i$ for $i=1,\ldots,N-1$, where $\bsv=(v_1,\ldots,v_{N-1})\in\bbP_{N-2}$. The Lebesgue measure transforms as
\begin{eqnarray}\label{eq:cvf}
\delta\Pa{1-\sum^N_{j=1} w_j}[\dif \bsw]=(1-u)^{N-2}\dif u\times \delta\Pa{1-\sum^{N-1}_{j=1} v_j}[\dif \bsv].
\end{eqnarray}
Indeed, the Jacobian matrix of the map $(v_1,\ldots,v_{N-1},u)\mapsto (w_1,\ldots,w_N)$ is $J=(1-u)^{N-1}$. Therefore $[\dif \bsw]=\dif w_1\cdots\dif w_N=(1-u)^{N-1}\dif u \dif v_1\cdots\dif v_{N-1}=(1-u)^{N-1}\dif u[\dif\bsv]$ for $\bsv=(v_1,\ldots,v_{N-1})$. Now compute the argument of the delta: 
\begin{eqnarray*}
1-\sum^N_{j=1}w_j &=& 1-\Pa{u + \sum^{N-1}_{j=1}(1-u)v_j} =(1-u)\Pa{1-\sum^{N-1}_{j=1}v_j}\implies\\
\delta\Pa{1-\sum^N_{j=1}w_j} &=& \frac1{1-u}\delta\Pa{1-\sum^{N-1}_{j=1}v_j},
\end{eqnarray*}
hence the desired reuslt \eqref{eq:cvf}. Then
$S=ut_N+(1-u)\sum_{i=1}^{N-1}v_it_i$. Therefore
$$
V_N(s)=\int^1_0\dif u(1-u)^{N-2}\int_{\bbP_{N-2}}\delta\Pa{s-ut_N-(1-u)\sum_{i=1}^{N-1}v_it_i}
\delta\Pa{1-\sum^{N-1}_{j=1} v_j}[\dif\bsv].
$$
Let
$s'=\frac{s-ut_N}{1-u}$. By the induction hypothesis,
$$
\int_{\bbP_{N-2}}
\delta\!\left(s'-\sum_{i=1}^{N-1}v_it_i\right)\delta\Pa{1-\sum^{N-1}_{j=1} v_j}[\dif \bsv]=\frac{1}{(N-3)!}
\sum_{i=1}^{N-1}
\frac{(t_i-s')_+^{N-3}}
{\prod_{j\neq i,\ j\leqslant N-1}(t_i-t_j)}.
$$
Substituting $s'$ gives
$$
t_i-s'=\frac{t_i-s+u(t_N-t_i)}{1-u}.
$$
Hence
$$
V_N(s)=\frac{1}{(N-3)!}\sum_{i=1}^{N-1}\frac{1}{\prod_{j\neq i,\ j\leqslant N-1}(t_i-t_j)}\int_0^1\Pa{t_i-s+u(t_N-t_i)}_+^{N-3}\dif u.
$$
The integral is the standard integral representation of the divided difference of order $N-1$. Evaluating it yields
$$
V_N(s)
=
\frac{1}{(N-2)!}
\sum_{i=1}^N
\frac{(t_i-s)_+^{N-2}}
{\prod_{j\ne i}(t_i-t_j)}.
$$
This completes the induction.
\end{itemize}
\item \textbf{Support and regularity.} Since $S=\sum^N_{j=1} w_jt_j$ is a convex combination of the numbers $t_1,\ldots,t_N$, its support is contained in
$[\min_j t_j,\ \max_j t_j]$. The density is a linear combination of the functions $(t_i-s)_+^{N-2}$, which are polynomials of degree $N-2$ on each interval between consecutive knots $t_i$. Hence it is piecewise polynomial of degree $N-2$.

At a knot of multiplicity $m$, i.e. when $m$ of the $t_i$ coincide, the divided difference involves derivatives up to order $m-1$ of the function $(x-s)_+^{N-2}$ at that knot. Since $(x-s)_+^{N-2}$ is $C^{N-3}$ at $x=s$, the resulting density is $C^{N-2-m}$ at that knot. In particular, if $m=N-1$, then $N-2-m=-1$, so the density is discontinuous. This proves the regularity statement.
\end{enumerate}
Thus the lemma is established. 
\end{proof}

\subsection{The numerical-shadow mechanism}

For a matrix $\bsZ\in M_d(\bbC)$ and Haar-random $\ket{\psi}\in\bbC^d$, the law of
$\Innerm{\psi}{\bsZ}{\psi}$ is the numerical shadow of $\bsM$ \cite{Dunkl2011}; it is supported on the
numerical range $W(\bsZ)$ \cite{Kippenhahn2008,Gustafson1997}. For $d=2$, it reduces to Archimedes'
theorem on the sphere. We state the version used for qubits with an arbitrary environment, together
with the reduction from $d$ to $2$ dimensions.

\begin{lem}\label{lem:archimedes}
\begin{enumerate}[(i)]
\item Let $\rho\sim\mu_{2,K}$ have Bloch vector $\bsr$, i.e.\ $\rho=\frac12(\I_2+\bsr\cdot\bdsigma)$.
Then $\bsr$ is rotation invariant; it is uniform on the unit sphere if $K=1$, and has density
proportional to $(1-|\bsr|^2)^{K-2}$ on the unit ball if $K\geqslant2$. Its orthogonal projection
onto any plane has the density
\begin{eqnarray}\label{eq:gK}
g_K(\xi)=\frac{K-\frac12}{\pi}(1-|\xi|^2)^{K-\frac32},\qquad |\xi|<1 .
\end{eqnarray}
\item Let $\ket{\psi}$ be Haar distributed in $\bbC^d$ for $d\geqslant3$, and let $\Pi$ be the
orthogonal projection onto a two-dimensional subspace $S$. Then $T=\norm{\Pi\psi}^2\sim\Beta(2,d-2)$,
and $\Pi\psi/\sqrt T$ is Haar distributed in $S$ and independent of $T$.
\end{enumerate}
\end{lem}

\begin{proof}
Let us prove the item (i) at first.
\begin{itemize}
\item \textbf{Step 1: Joint eigenvalue density of $\rho\sim \mu_{2,K}$.} For $\rho\sim\mu_{2,K}$, the eigenvalues $\lambda_1\leqslant\lambda_2$ have joint density proportional to $(\lambda_1\lambda_2)^{K-2}(\lambda_2-\lambda_1)^2$ on the simplex $\lambda_1+\lambda_2=1$, where $\lambda_1,\lambda_2\geqslant0$. This is the standard induced-measure result (\cite{Zyczkowski2001,Sommers2004}); it follows from the $d\times K$ Ginibre representation $\rho=\frac{\bsG\bsG^\dagger}{\Tr{\bsG\bsG^\dagger}}$ and the Jacobian of the diagonalization.
\item \textbf{Step 2: Change of variables to the Bloch radius.} Write the eigenvalues as
$\lambda_1=\frac{1-r}2,\lambda_2=\frac{1+r}2$, where $r\in[0,1]$. Therefore the density of $r=\abs{\bsr}$ is propportional to $(\lambda_1\lambda_2)^{K-2}(\lambda_2-\lambda_1)^2=\Pa{\frac{1-r^2}4}^{K-2}r^2$. Collecting powers:
$$
p(r)\propto r^2(1-r^2)^{K-2}.
$$
Now, the Bloch vector $\bsr$ has radial density $p(r)$ and is rotationally invariant (by unitary invariance of $\mu_{2,K}$). The volume element in $\bbR^3$ is $4\pi r^2\dif r$. Therefore the density of $\bsr$ with respect to Lebesgue meansre on $\bbR^3$ is 
$$
f(\bsr)\propto \frac{p(r)}{4\pi r^2}\propto(1-r^2)^{K-2}.
$$
So the Bloch vector has density proportional to $(1-\abs{\bsr}^2)^{K-2}$ on the unit ball. This proves the $K\geqslant2$ statement.
\item \textbf{Step 3: The case $K=1$.} For $K=1,\rho=\proj{\psi}$ is a pure state, and the Bloch vector $\bsr$ lies on the unit sphere. By unitary invariance, $\bsr$ is uniform on the sphere $S^2$ (This can also be seen as the limit of the above density as $K\to1^+$ after normalization, or directly from the fact that a Haar-random pure qubit state corresponds to a uniform point on the Bloch sphere---Archimedes' theorem.)
\item \textbf{Orthogonal projection onto a plane.} Let $\bsr=(r_1,r_2,r_3)$ have density $f(\bsr)\propto (1-\abs{\bsr}^2)^{K-2}$ on the unit ball (for $K\geqslant2$). Project onto the $(r_1,r_2)$-plane: $\xi=(r_1,r_2)$. Then
$$
g_K(\xi) = \int^{\sqrt{1-\abs{\xi}^2}}_{-\sqrt{1-\abs{\xi}^2}}f(\xi,r_3)\dif r_3\propto  \int^{\sqrt{1-\abs{\xi}^2}}_{-\sqrt{1-\abs{\xi}^2}}(1-\abs{\xi}^2-z^2)^{K-2}\dif z.
$$
Let $q=\abs{\xi}^2$ and $R=\sqrt{1-q}$. Substitute $z=R\sin\theta,\dif z=R\cos\theta\dif\theta$:
\begin{eqnarray*}
\int^R_{-R}(R^2-z^2)^{K-2}\dif z=R^{2K-3}\int^{\frac\pi2}_{-\frac\pi2}\cos^{2K-3}\theta\dif \theta=R^{2K-3} \Beta(\tfrac12,K-1)=R^{2K-3}\frac{\Gamma(\tfrac12)\Gamma(K-1)}{\Gamma(K-\tfrac12)},
\end{eqnarray*}
implying that $g_K(\xi)\propto (1-q)^{K-\frac32}=(1-\abs{\xi}^2)^{K-\frac32}$. The normalization constant is fixed by
$$
\int_{\abs{\xi}<1} (1-\abs{\xi}^2)^{K-\frac32}[\dif \xi] = \pi\int^1_0(1-q)^{K-\frac32}\dif q = \frac{\pi}{K-\tfrac12}.
$$
Thus we have obtained Eq.~\eqref{eq:gK}.
\item \textbf{The case $K=1$ for the projection.} For $K=1,\bsr$ is uniform on the unit sphere $S^2$. The projection onto a plane is the classical Archimedes result: the density of the projection of a uniform point on $S^2$ onto a diameter is uniform on $[-1,1]$. More precisely, if $\xi$ is the projection onto a plane, then 
$$
g_1(\xi) = \frac1{2\pi\sqrt{1-\abs{\xi}^2}},\quad \abs{\xi}<1.
$$
So \eqref{eq:gK} holds for all $K\geqslant1$. This completes the proof of Part (i). 
\end{itemize}
Next, we show the part (ii).
\begin{itemize}
\item \textbf{Step 1: Gaussian representation.} Let $\bsz=(z_1,\ldots,z_d)^\t$ be a standard complex Gaussian vector in $\bbC^d$, i.e. $z_j = x_j + iy_j$ with $x_j,y_j\sim\cN(0,\frac12)$ i.i.d. (or equivalently, each $z_j\sim\mathcal{CN}(0,1)$). Then $\ket{\psi}= \frac{\bsz}{\norm{\bsz}}$ is Haar distributed on the unit sphere of $\bbC^d$.
\item \textbf{Step 2: Decomposition into the subspace and its complement.} Let $S$ be a fixed two-dimensional subspace, and let $\Pi$ be the orthogonal projection onto$S$. Write $\bsz = \Pi \bsz + (\I-\Pi)\bsz =: \bsu + \bsv$, where $\bsu\in S\cong\bbC^2$ and $\bsv\in S^\perp\cong\bbC^{d-2}$.
Since $\bsz$ is standard complex Gaussian and $\Pi$ is an orthogonal projection, $\bsu$ and $\bsv$ are independent complex Gaussian vectors in $S$ and $S^\perp$ respectively, with: (1) $\bsu$ having $2$ complex components (i.e. $4$ real components), each standard complex Gaussian;
(2) $\bsv$ having $d-2$ complex components (i.e. $2(d-2)$ real components), each standard complex Gaussian.
\item \textbf{Step 3: Distribution of $\norm{\bsu}^2$ and $\norm{\bsv}^2$.} For a standard complex Gaussian vector with $m$ complex components, the squared norm is a sum of $2m$ independent $\chi^2_1$ variables (each real component squared), which is $\mathrm{Gamma}(m,1)$ (shape $m$, rate $1$). More precisely:
$\norm{\bsu}^2 \sim \mathrm{Gamma}(2,1)$ and $\norm{\bsv}^2\sim \mathrm{Gamma}(d-2,1)$. These are independent because $\bsu$ and $\bsv$ are independent.
\item \textbf{Step 4: Distribution of $T=\norm{\Pi\psi}^2$.} We have
$\psi = \frac{\bsz}{\norm{\bsz}} = \frac{\bsu+\bsv}{\sqrt{\norm{\bsu}^2+\norm{\bsv}^2}}$, so
$T = \norm{\Pi\psi}^2 = \frac{\norm{\bsu}^2}{\norm{\bsu}^2+\norm{\bsv}^2}$. If $X\sim\mathrm{Gamma}(2,1)$ and $Y\sim\mathrm{Gamma}(d-2,1)$ are independent, then $\frac{X}{X+Y}\sim\Beta(2,d-2)$. Therefore $T \sim \Beta(2,d-2)$. This proves the first assertion of (ii).
\item \textbf{Step 5: Conditional distribution of $\Pi\psi/\sqrt{T}$.} Given $\bsz$, the direction of $\bsu=\Pi\bsz$ is independent of $\norm{\bsu}$. Since $\bsu$ is a standard complex Gaussian vector in $S\cong\bbC^2\), its direction$\frac{\bsu}{\norm{\bsu}}$ is Haar distributed on the unit sphere of $S$, and is independent of $\norm{\bsu}$. Moreover, $\bsu$ is independent of $\bsv$, so the direction of $\bsu$ is independent of both $\norm{\bsu}$ and $\norm{\bsv}$. Now,
$$
\frac{\Pi\psi}{\sqrt{T}}
= \frac{\bsu/\norm{\bsz}}{\norm{\bsu}/\norm{\bsz}}
= \frac{\bsu}{\norm{\bsu}}.
$$
Therefore $\Pi\psi/\sqrt{T} = \bsu/\norm{\bsu}$ is Haar distributed on the unit sphere of $S$, and is independent of $T=\norm{\bsu}^2/(\norm{\bsu}^2+\norm{\bsv}^2)$ (since the direction of $\bsu$ is independent of $\norm{\bsu}$ and $\norm{\bsv}$, hence of $T$). This completes the proof of Part (ii). \qedhere
\end{itemize}
\end{proof}

For $K\geqslant2$, (i) applied to a line instead of a plane gives the density of $\frac12(1+t)$ with
$t$ the projection; it is $\Beta(K,K)$, in agreement with Lemma~\ref{lem:dirichlet} for $d=2$. When
$\bsH_1,\bsH_2\in\Herm(\bbC^2)$, the pair $(\Innerm{\psi}{\bsH_1}{\psi},\Innerm{\psi}{\bsH_2}{\psi})$ is an affine image of the
projection of the Bloch vector of $\proj\psi$. Hence its law is supported on an ellipse, the image of
the unit disk, and has the density $\Pa{2\pi ab\sqrt{1-R^2}}^{-1}$, where $a,b$ are the semi-axes and $R$ is
the affine elliptical radius. This is Proposition~2.1 of \cite{Dunkl2011}.

%===================================================%
\section{Expectation values and moment vectors}\label{sect:expect}
%===================================================%

\subsection{Orbital inner products}

Let $\bsA,\bsB\in\Herm(\bbC^d)$ have simple spectra $a_1<\cdots<a_d$ and $b_1<\cdots<b_d$, and let
$\bsU$ be Haar distributed on $\rU(d)$. Write $N_d=\frac{d(d-1)}2$ and
$\Inner{\bsa}{\bsb_\pi}=\sum^d_{k=1}a_kb_{\pi(k)}$ for $\pi\in S_d$, the symmetric group of the set $\set{1,\ldots,d}$. The HCIZ formula reads
\cite{HarishChandra1957,Itzykson1980,McSwiggen2021}
\begin{eqnarray}\label{eq:hciz}
\int_{\rU(d)}e^{z\Tr{\bsA\bsU \bsB\bsU^\dagger}}\dif\mu_{\haar}(\bsU)
=\Pa{\prod_{k=1}^{d-1}k!}\frac{\det\Br{e^{za_ib_j}}_{i,j=1}^d}{z^{N_d}V_d(\bsa)V_d(\bsb)},\qquad z\in\bbC,
\end{eqnarray}
with the removable singularity at $z=0$ filled in by continuity.

\begin{thrm}\label{th:orbital}
For $d\geqslant2$, the random variable $R=\Tr{\bsA\bsU\bsB\bsU^\dagger}$ has the density
\begin{eqnarray}\label{eq:orbital}
p_d(r|\bsa,\bsb)=\frac{(-1)^{N_d}\prod_{k=1}^{d-1}k!}{(N_d-1)!V_d(\bsa)V_d(\bsb)}\sum_{\tau\in S_d}\sign(\tau)
\pa{r-\Inner{\bsa}{\bsb_\tau}}_+^{N_d-1},
\end{eqnarray}
supported on $[\Inner{\bsa}{\bsb_{\mathrm{rev}}},\Inner{\bsa}{\bsb}]$, where $\bsb_{\mathrm{rev}}$ lists $\bsb$
in decreasing order.
\end{thrm}
Put $z=-s$ in \eqref{eq:hciz}. The left-hand side becomes the two-sided Laplace transform of the law of $R$,
which has compact support. For $\Re(s)>0$ and every $\kappa\in\bbR$,
$$
\int_\bbR e^{-sr}(r-\kappa)_+^{N-1}\dif r=(N-1)!e^{-s\kappa}s^{-N}.
$$
Expanding the determinant
gives $\det[e^{-sa_ib_j}]=\sum_{\pi\in S_n}\sign(\pi)e^{-s\Inner{\bsa}{\bsb_\pi}}$, so both sides of \eqref{eq:orbital} have the
same two-sided Laplace transform on $\Re(s)>0$, where both transforms converge absolutely. By uniqueness
they agree almost everywhere. The support is given by the rearrangement inequality.
\begin{proof}
Basically, the density of $R=\Tr{\bsA\bsU\bsB\bsU^\dagger}$ is given by 
\begin{eqnarray}\label{eq:pdf1}
p_d(r|\bsa,\bsb) =
\int_{\rU(d)}\delta\Pa{r-\Tr{\bsA\bsU\bsB\bsU^\dagger}}\dif\mu_{\haar}(\bsU).
\end{eqnarray}
The Laplace transformation of $p_d(r|\bsa,\bsb)$ is given by
\begin{eqnarray*}
\sL(p_d(\cdot|\bsa,\bsb))(s) = \int_{\rU(d)}
e^{-s\Tr{\bsA\bsU\bsB\bsU^\dagger}}\dif\mu(\bsU) =
C_d\cdot\frac{\det_d\Pa{e^{-sa_ib_j}}}{(-s)^{\binom{d}{2}}V_d(\bsa)V_d(\bsb)}\quad(s>0).
\end{eqnarray*}
Thus using inverse Laplace transformation, we get that
\begin{eqnarray*}
p_d(r|\bsa,\bsb)
=C_d\cdot\frac1{V_d(\bsa)V_d(\bsb)}\sL^{-1}\Pa{(-s)^{-\binom{d}{2}}\det_d\Pa{e^{-sa_ib_j}}}.
\end{eqnarray*}

Now
\begin{eqnarray*}
\sL^{-1}\Pa{(-s)^{-\binom{d}{2}}\det_d\Pa{e^{-sa_ib_j}}}
&=&\sL^{-1}\Pa{ (-s)^{-\binom{d}{2}}\sum_{\pi\in
S_d}\sign(\pi)\prod^d_{k=1}e^{-s a_kb_{\pi(k)}}} \\
&=&\sL^{-1}\Pa{ (-s)^{-\binom{d}{2}}\sum_{\pi\in
S_d}\sign(\pi)e^{-s \Inner{\bsa}{\bsb_\pi}}}\\
&=& \sum_{\pi\in
S_d}\sign(\pi)\sL^{-1}\Pa{(-s)^{-\binom{d}{2}}e^{-s
\Inner{\bsa}{\bsb_\pi}}}
\end{eqnarray*}
We have seen that
\begin{eqnarray*}
\sL^{-1}\Pa{(-s)^{-\binom{d}{2}}e^{-s \Inner{\bsa}{\bsb_\pi}}} =
\frac{(-1)^{\binom{d}{2}}}{\Gamma\Pa{\frac{d(d-1)}2}}\Pa{r-\Inner{\bsa}{\bsb_\pi}}^{\frac{(d+1)(d-2)}2}_+
\end{eqnarray*}
Therefore
\begin{eqnarray*}
p_d(r|\bsa,\bsb)
=\frac{(-1)^{\binom{d}{2}}\prod^d_{k=1}\Gamma(k)}{\Gamma\Pa{\frac{d(d-1)}2}V_d(\bsa)V_d(\bsb)}\sum_{\pi\in
S_d}\sign(\pi)\Pa{r-\Inner{\bsa}{\bsb_\pi}}^{\frac{(d+1)(d-2)}2}_+.
\end{eqnarray*}
The support follows from the von Neumann trace inequality, or
equivalently from the rearrangement inequality:
$$
\min_{\bsU\in\rU(d)}\Tr{\bsA\bsU\bsB\bsU^\dagger}
=\Inner{\bsa}{\bsb_{\mathrm{rev}}}\text{ and
}\max_{\bsU\in\rU(d)}\Tr{\bsA\bsU\bsB\bsU^\dagger}
=\inner{\bsa}{\bsb}.
$$
This completes the proof.
\end{proof}
For $d=2$, the density is uniform on $[a_1b_2+a_2b_1,\,a_1b_1+a_2b_2]$ with height
$[(a_2-a_1)(b_2-b_1)]^{-1}$; this also follows from Lemma~\ref{lem:archimedes}. 
\begin{exam}
For $d=3$, $S_3=\set{(1),(12),(13),(23),(123),(132)}$. Then
\begin{eqnarray}
p_3(r|\bsa,\bsb) =-\frac1{V_3(\bsa)V_3(\bsb)}\sum_{\tau\in
S_3}\sign(\tau)\Pa{r-\Inner{\bsa}{\bsb_\tau}}^2_+.
\end{eqnarray}
Let $\bsa=(1,2,4)$ and $\bsb=(3,6,11)$. We get that
$\inner{\bsa}{\bsb_{\mathrm{rev}}}=35$ and
$\inner{\bsa}{\bsb}=59$. Then
\begin{eqnarray*}
&&35=a_1b_3+a_2b_2+a_3b_1<40=a_1b_2+a_2b_3+a_3b_1<41=a_1b_3+a_2b_1+a_3b_2\\
&&<49=a_1b_1+a_2b_3+a_3b_2<56=a_1b_2+a_2b_1+a_3b_3<59=a_1b_1+a_2b_2+a_3b_3
\end{eqnarray*}
and thus
\begin{eqnarray}
p_3(r|\bsa,\bsb) = \frac1{720}
\begin{cases}
(r-35)^2,&\text{if }x\in[35,40];\\
5(2r-75),&\text{if }x\in[40,41];\\
-r^2+92 r-2056,&\text{if }x\in[41,49];\\
3(-2r + 115),&\text{if }x\in[49,56];\\
(r-59)^2,&\text{if }x\in[56,59].
\end{cases}
\end{eqnarray}
See the following Fig.~\ref{fig:splines}(a). For $d\geqslant3$ and
pairwise distinct knots, \eqref{eq:orbital} is a spline of degree $N_d-1$ that is $C^{N_d-2}$.
\end{exam}

\subsection{Expectation values are Dirichlet splines}

The relation between Theorem~\ref{th:orbital} and the following Theorem~\ref{th:spline} deserves emphasis. 
Theorem~\ref{th:orbital} concerns the orbital inner product 
$\Tr{\bsA\bsU\bsB\bsU^\dagger}$ with $\bsU$ Haar-distributed on $\rU(d)$, 
whereas Theorem~\ref{th:spline} concerns the expectation value 
$\Tr{\bsA\rho}$ for an induced random state $\rho\sim\mu_{d,K}$. 
The two are connected as follows. If $\bsB$ is a rank-one projection, say 
$\bsB=\proj{\psi}$ with $\ket{\psi}$ Haar-random in $\bbC^d$, then 
$\Tr{\bsA\bsU\bsB\bsU^\dagger}=\Innerm{\psi}{\bsU^\dagger \bsA \bsU}{\psi}$, 
which is the expectation value of $\bsA$ in a Haar-random pure state. 
Since a Haar-random pure state is exactly the induced state $\rho\sim\mu_{d,1}$ with 
$K=1$, Theorem~\ref{th:spline} with $K=1$ recovers the orbital inner product law 
for rank-one $\bsB$. In this case, \eqref{eq:spline} reduces to the numerical measure 
of a Hermitian matrix, a result established in 
\cite{Gallay2012,Dunkl2011,CamposVenuti2013}. More generally, for $\bsB$ of arbitrary 
rank, Theorem~\ref{th:orbital} and Theorem~\ref{th:spline} describe different 
randomizations---unitary conjugation versus partial trace---and the two laws coincide 
only in the rank-one, $K=1$ case.

\begin{thrm}\label{th:spline}
Let $\bsA\in\Herm(\bbC^d)$ have eigenvalues $a_1,\ldots,a_d$ (with multiplicity) and not be a multiple
of $\I$. Let $(t_1,\ldots,t_{dK})$ list each $a_i$ exactly $K$ times. For $\rho\sim\mu_{d,K}$, the
expectation value $\Tr{\bsA\rho}$ has the density
\begin{eqnarray}\label{eq:spline}
p_{d,K}(r|\bsa)=(dK-1)[t_1,\ldots,t_{dK}](\cdot-r)_+^{dK-2}.
\end{eqnarray}
If $\bsa$ has a distinct eigenvalue of multiplicity $m$, then $p_{d,K}$ is $C^{(d-m)K-2}$ at that
eigenvalue; for a simple spectrum it is $C^{(d-1)K-2}$ everywhere.
\end{thrm}

\begin{proof}
In an eigenbasis of $\bsA$, $\Tr{\bsA\rho}=\sum^d_{i=1}a_i\rho_{ii}$ with $(\rho_{11},\ldots,\rho_{dd})\sim\op{Dirichlet}(K,\ldots,K)$ by
Lemma~\ref{lem:dirichlet}. A $\op{Dirichlet}(K,\ldots,K)$ vector is obtained from $\bsw\sim\op{Dirichlet}(1,\ldots,1)$ on
$\bbP_{dK-1}$ by summing consecutive blocks of $K$ coordinates, so
$$
\Tr{\bsA\rho}=\sum^{dK}_{j=1}w_jt_j,
$$
and Lemma~\ref{lem:curry} applies.
\end{proof}

Two special cases are immediate. If $\bsP$ is a projection of
rank $m$ with $0<m<d$, then $\Tr{\bsP\rho}\sim\Beta(mK,(d-m)K)$.  This follows from Lemma~\ref{lem:dirichlet} and the aggregation property: the trace $\Tr{\bsP\rho}$ is the sum of $m$ diagonal entries of $\rho$ in an eigenbasis of $\bsP$, which is $\Beta(mK,(d-m)K)$ by the aggregation property of the Dirichlet distribution.

\begin{exam}[Qutrit]
For a given qutrit observable $\bsA$ with simple spectrum
$\lambda(\bsA)=(a_1,a_2,a_3) (a_1<a_2<a_3)$, the probability
distribution density of $\Tr{\bsA\rho}$, where $\rho$ is resulted
from partial-tracing a subsystem over a Haar-distributed random pure
state on $\complex^d\ot\complex^K$, where $d=K=3$, is given by the following:
\begin{eqnarray}
p_3(r|\bsa) =p_{3,3}(r|\bsa) =
\frac{4!}{[V_3(\bsa)]^5}\cdot\sum^3_{k=1}c_k(r|\bsa)\chi_{\set{r>a_k}},
\end{eqnarray}
where
\begin{eqnarray}
c_1(r|\bsa)=\Phi(r,a_1,a_2,a_3),\quad
c_2(r|\bsa)=\Phi(r,a_2,a_3,a_1),\quad
c_3(r|\bsa)=\Phi(r,a_3,a_1,a_2).
\end{eqnarray}
Here $\Phi(r,a_1,a_2,a_3)$ is defined as
\begin{eqnarray*}
\Phi(r,a_1,a_2,a_3) &=&
(a_3-a_2)^5\Br{7a^2_1-7(a_2+a_3)a_1+2a^2_2+2a^2_3+3a_2a_3}(r-a_1)^7\notag\\
&&-7(a_3-a_2)^4(a_2+a_3-2a_1)V_3(\bsa)(r-a_1)^6\\
&&+7(a_3-a_2)^3V_3(\bsa)^2(r-a_1)^5.
\end{eqnarray*}
Moreover the support of $p_3(r|\bsa)$ is the closed interval
$[a_1,a_3]$. Note that
$\Phi(r,a_1,a_2,a_3)+\Phi(r,a_2,a_3,a_1)+\Phi(r,a_3,a_1,a_2)=0$.
Thus we see that
\begin{eqnarray}
p_3(r|\bsa) = \frac{24}{[V_3(\bsa)]^5}
\begin{cases}
c_1(r|\bsa),&\text{if
}r\in[a_1,a_2]\\
-c_3(r|\bsa),&\text{if }r\in[a_2,a_3],
\end{cases}
\end{eqnarray}
which can be rewritten as
\begin{eqnarray}\label{eq:qutritlocal}
p_3(r|\bsa) = \frac{24}{[V_3(\bsa)]^5}
\Pa{c_1(r|\bsa)\chi_{[a_1,a_2]}(r) -c_3(r|\bsa)\chi_{[a_2,a_3]}(r)}.
\end{eqnarray}
Unlike the alternating sum \eqref{eq:spline}, formula \eqref{eq:qutritlocal} involves no cancellation
between large terms and is the form used in Section~\ref{sect:overlap}. As the environment grows,
$p_{d,K}$ becomes smoother at the eigenvalues and concentrates around $\Tr{\bsA}/d$;
Figure~\ref{fig:splines}(b) shows this for $\boldsymbol{a}=(0,0.3,1)$.
\end{exam}

\begin{figure}[t]
\centering
\includegraphics{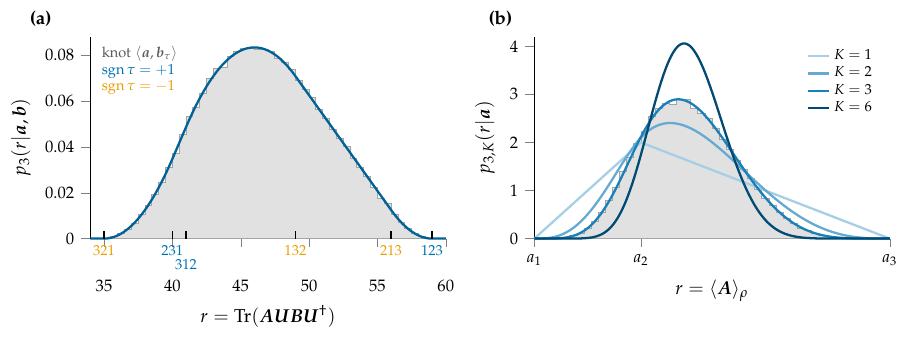}
\caption{Spline structure of random inner products. Lines: exact densities; grey histograms:
$2\times10^5$ samples. (a) Orbital inner product, Theorem~\ref{th:orbital}, for $\bsa=(1,2,4)$ and
$\bsb=(3,6,11)$. The six knots $\Inner{\bsa}{\bsb_\tau}$ are labelled by $\tau$ in one-line notation and
coloured by $\sign\tau$; the density is quadratic between knots and $C^1$ across them. (b) Expectation
values $\Tr{\bsA\rho}$, $\rho\sim\mu_{3,K}$, $\bsa=(0,0.3,1)$, Theorem~\ref{th:spline}; the histogram is for
$K=3$. The kink at $a_2$ for $K=1$ is smoothed to $C^{2K-2}$ as $K$ grows.}\label{fig:splines}
\end{figure}

\subsection{The moment simplex}

Let $\bsA$ have simple spectrum $a_1<\cdots<a_d$ and set $r_k=\langle \bsA^k\rangle_\rho$ for
$k=0,\ldots,d-1$ (so $r_0=1$). For a multi-index
$I=\set{i_1<\cdots<i_k}\subseteq \set{1,\dots,d}$, define the
elementary symmetric polynomial
$$
e_m(a_I)=\sum_{1\leqslant \ell_1<\cdots<\ell_m\leqslant k}
a_{i_{\ell_1}}\cdots a_{i_{\ell_m}},\quad e_0(a_I)=1.
$$
Define $\hat k:=\set{1,\ldots,d}\backslash\set{k}$, where
$k=1,\ldots,d$. Let
$e_j(a_{\hat k})$ denote the $j$-th elementary symmetric polynomial in $\set{a_l:l\neq k}$ for each $k=1,\ldots,d$, and define
\begin{eqnarray}\label{eq:bary}
\phi_k(\bsr)=\sum_{j=0}^{d-1}(-1)^{d-1-j}e_{d-1-j}(a_{\hat k})\,r_j,\qquad
x_k(\bsr)=\frac{\phi_k(\bsr)}{\prod_{l\neq k}(a_k-a_l)}.
\end{eqnarray}
The numbers $x_k(\bsr)$ are the
barycentric coordinates of $\bsr=(r_1,\ldots,r_{d-1})$ with respect to the points
$\bsv_k=(a_k,a_k^2,\ldots,a_k^{d-1})$ (Appendix~\ref{app:bary}).

\begin{thrm}\label{th:simplex}
Let $d\geqslant2$, $K\geqslant1$ and $\rho\sim\mu_{d,K}$. The moment vector
$\bsr=(\langle \bsA\rangle_\rho,\ldots,\langle \bsA^{d-1}\rangle_\rho)$ has the density
\begin{eqnarray}\label{eq:simplex}
f_{d,K}(\bsr)=\frac{\Gamma(dK)}{\Gamma(K)^d\,|V_d(\bsa)|}\prod_{i=1}^dx_i(\bsr)^{K-1}\chi_{S_a}(\bsr)
=\frac{\Gamma(dK)}{\Gamma(K)^d|V_d(\bsa)|^{2K-1}}\prod_{i=1}^d\abs{\phi_i(\bsr)}^{K-1}\chi_{S_a}(\bsr),
\end{eqnarray}
where $S_a=\conv\set{\bsv_1,\ldots,\bsv_d}$ is a $(d-1)$-simplex.
\end{thrm}

\begin{proof}
Work in an eigenbasis of $\bsA$, so $\bsA=\diag(a_1,\dots,a_d)$. Let
$X_i=\rho_{ii}$ be the diagonal entries of $\rho$. For the induced
ensemble $\bbC^d\ot\bbC^K$, the vector $(X_1,\dots,X_d)$ follows the
Dirichlet distribution
$$
(X_1,\dots,X_d)\sim \op{Dirichlet}(\underbrace{K,\dots,K}_{d\text{
times}}),
$$
whose density on the simplex is
$$
g(x_1,\dots,x_d)=\frac{\Gamma(dK)}{\Gamma(K)^d} x_1^{K-1}\cdots
x_d^{K-1} \delta(1-x_1-\cdots-x_d).
$$
The moments are
$$
r_k=\langle\bsA^k\rangle_\rho = \sum_{j=1}^d a_j^k X_j,\quad
k=0,1,\dots,d-1.
$$
The affine map $T: \bbP_{d-1} \to \bbR^{d-1}$ given by
$(x_1,\dots,x_{d-1})\mapsto (r_1,\dots,r_{d-1})$ has Jacobian determinant equal to
the Vandermonde $V_d(\bsa)$ (up to sign). Indeed,
\begin{eqnarray*}
r_k = \sum^{d-1}_{j=1}a^k_jx_j +a^k_d\Pa{1-\sum^{d-1}_{j=1}x_j} = a^k_d + \sum^{d-1}_{j=1}(a^k_j-a^k_d)x_j.
\end{eqnarray*}
Thus $\frac{\partial r_k}{\partial x_j}=a^k_j-a^k_d$, where $j,k=1,\ldots,d-1$. Then the Jacobian is a Vandermonde-like determinant. A standard computation shows that its absolute value is $\abs{V_d(\bsa)}$. Solving for the
barycentric coordinates gives
$$
x_i = \frac{\phi_i(\bsr)}{\prod_{j\neq i}(a_i-a_j)}.
$$
Thus
$$
\prod_{i=1}^d x^{K-1}_i(\bsr) = \frac{\prod^d_{i=1}
\abs{\phi_i(\bsr)}^{K-1}}{\abs{V_d(\bsa)}^{2(K-1)}}.
$$
The density of $\bsr=(r_1,\ldots,r_{d-1})$ is obtained from the density of $\bsx=(x_1,\ldots,x_d)$ via the change-of-variable formula. Since $T$ maps the simplex $\bbP_{d-1}$ to $S_a=\op{conv}\Set{\bsv_1,\ldots,v_d}$, where $\bsv_k=(a_k,a^2_k,\ldots,a^{d-1}_k)$, we have
\begin{eqnarray*}
f_{d,K}(\bsr) &=& g(\bsx(\bsr))\cdot\Abs{\det\Pa{\frac{\partial(x_1,\ldots,x_{d-1})}{\partial(r_1,\ldots,r_{d-1})}}} =\Pa{\frac{\Gamma(dK)}{\Gamma(K)^d}\prod^d_{k=1}x^{K-1}_k(\bsr) }\frac1{\abs{V_d(\bsa)}}\\
&=& \frac{\Gamma(dK)}{\Gamma(K)^d} \frac{\prod^d_{i=1}
\abs{\phi_i(\bsr)}^{K-1}}{\abs{V_d(\bsa)}^{2(K-1)}} \frac1{\abs{V_d(\bsa)}}=\frac{\Gamma(dK)}{\Gamma(K)^d|V_d(\bsa)|^{2K-1}}\prod_{i=1}^d\abs{\phi_i(\bsr)}^{K-1}\chi_{S_a}(\bsr).
\end{eqnarray*}
The support of $g$ is the simplex $\bbP_{d-1}$. Under the affine map $T$, the image of $\bbP_{d-1}$ is the convex hull of the points
$$
\bsv_k=T(\bse_k)=(a_k,a^2_k,\ldots,a^{d-1}_k),
$$
where $\bse_k$ are the standard basis vectors. Therefore, $S_a=\op{conv}\set{\bsv_1,\ldots,\bsv_d}$ is a $(d-1)$-simplex in $\bbR^{d-1}$. Ths is the support of $f_{d,K}$.
\end{proof}

At $K=1$, the density is uniform on $S_a$, which is the pure-state result of \cite{Zhang2022}; at $K=d$ it
is the Hilbert--Schmidt law. Only the first expression in \eqref{eq:simplex} is valid without absolute
values: $\prod_k\prod_{l\neq k}(a_k-a_l)=(-1)^{\frac{d(d-1)}2}V_d(\bsa)^2$ is negative when $d\equiv2,3\pmod4$.
The first two moments follow from \eqref{eq:secondmoment} alone.

\begin{prop}\label{prop:moments}
Let $\rho\sim\mu_{d,K}$ and $\bsA,\bsB\in\Herm(\bbC^d)$. Then $\bbE[\Tr{\bsA\rho}]=\frac1d\Tr{\bsA}$,
$$
\Cov\pa{\Tr{\bsA\rho},\Tr{\bsB\rho}}=\frac{\Tr{\bsA\bsB}-\tfrac1d\Tr{\bsA}\Tr{\bsB}}{d(dK+1)},
\qquad
\bbE[\Var_\rho(\bsA)]=\frac{dK}{dK+1}\sigma^2(\bsA),
$$
where $\sigma^2(\bsA):=\frac1d\Tr{\bsA^2}-(\frac1d\Tr{\bsA})^2$ is the variance of $\bsA$ in the maximally mixed state.
\end{prop}

\begin{proof}
The covariance is
$$
\Tr{(\bsA\ot \bsB)\bbE[\rho\ot\rho]}-\tfrac1{d^2} \Tr{\bsA}\Tr{\bsB};
$$
insert \eqref{eq:secondmoment} and use $\Tr{(\bsA\ot \bsB)\bsF}=\Tr{\bsA\bsB}$. For the last identity,
$$
\bbE[\Var_\rho(\bsA)]=\frac{\Tr{\bsA^2}}d-\Var[\Tr{\bsA\rho}]-\Pa{\frac{\Tr{\bsA}}d}^2=\sigma^2(\bsA)-\frac{\sigma^2(\bsA)}{dK+1}.\qedhere
$$
\end{proof}

For $d\geqslant3$ and a simple spectrum, the pair $(\langle \bsA\rangle_\rho,\langle \bsA^2\rangle_\rho)$ has a density whose support is the polygon $\conv\set{(a_k,a_k^2)}$; it is a triangle for $d=3$. For
$d=2$, the pair lies on the line $s=(a_1+a_2)r-a_1a_2$ and has no two-dimensional density.

\begin{cor}\label{cor:qutrit}
Let $d=3$ and write $\phi_{ij}(r)=(a_i+a_j)r-a_ia_j$. Then $(R,S)=(\langle \bsA\rangle_\rho,\langle \bsA^2\rangle_\rho)$
has the density
\begin{eqnarray}\label{eq:qutritjoint}
f_{R,S}(r,s)=\frac{\Gamma(3K)}{\Gamma(K)^3|V_3(\bsa)|^{2K-1}}\prod_{1\leqslant i<j\leqslant3}\abs{s-\phi_{ij}(r)}^{K-1}
\end{eqnarray}
on the triangle bounded by $s=\phi_{13}(r)$ from above and by $s=\phi_{12}(r)$ ($a_1\leqslant r\leqslant a_2$) and
$s=\phi_{23}(r)$ ($a_2\leqslant r\leqslant a_3$) from below. The pair $(\langle \bsA\rangle_\rho,\Delta_\rho \bsA)$ has the
density $2x\,f_{R,S}(r,r^2+x^2)$ on the region
$$
\sqrt{(a_{k+1}-r)(r-a_k)}\leqslant x\leqslant\sqrt{(a_3-r)(r-a_1)},\qquad a_k\leqslant r\leqslant a_{k+1},\ k=1,2.
$$
\end{cor}

\begin{proof}
For $d=3$, $\phi_k(r,s)=s-\phi_{ij}(r)$ with $\set{i,j,k}=\set{1,2,3}$. The second statement follows
from $s=r^2+x^2$, $\frac{\partial s}{\partial x}=2x$, and $\phi_{ij}(r)-r^2=(a_j-r)(r-a_i)$.
\end{proof}

Figure~\ref{fig:mechanism} follows the Hilbert--Schmidt law through these three coordinate systems.

\subsection{A complete variance law}

For a concrete observable, the variance $V=\Var_\rho(\bsA)=S-R^2$ can be integrated out of
Corollary~\ref{cor:qutrit} in finite terms. Take $\bsA=\diag(-1,0,1)$. Then $|V_3(\bsa)|=2$,
$x_1=\frac{s-r}2$, $x_2=1-s$, $x_3=\frac{s+r}2$, and \eqref{eq:qutritjoint} becomes
\begin{eqnarray}\label{eq:m101}
f_{R,S}(r,s)=c_K\big((1-s)(s^2-r^2)\big)^{K-1},\qquad |r|\leqslant s\leqslant1,\qquad
c_K=\frac{\Gamma(3K)}{2^{2K-1}\Gamma(K)^3}.
\end{eqnarray}

\begin{thrm}\label{th:variance}
Let $\bsA=\diag(-1,0,1)$ and $\rho\sim\mu_{3,K}$. Denote
\begin{eqnarray*}
\begin{cases}
c(v)=\sqrt{1-v},&v\in(0,1);\\
\alpha_\pm(v)=\frac12(1\pm\sqrt{1-4v}),&v<\frac14.
\end{cases}
\end{eqnarray*}
With
$h_v(t)=\big((1-v-t^2)((t^2+v)^2-t^2)\big)^{K-1}$, the variance $V=\Var_\rho(\bsA)$ has the density
\begin{eqnarray}\label{eq:variance}
f_V(v)=2c_K\times\begin{cases}
\int_0^{\alpha_-(v)}h_v(t)\dif t+\int_{\alpha_+(v)}^{c(v)}h_v(t)\dif t, & 0<v<\frac14,\\[4pt]
\int_0^{c(v)}h_v(t)\dif t, & \frac14\leqslant v<1,
\end{cases}
\end{eqnarray}
and $f_V=0$ outside $[0,1]$. Its mean is $\frac{2K}{3K+1}$. For $K=3$, $f_V(v)=315\,[H_v(c)+H_v(\alpha)-H_v(\beta)]$ for
$v<\frac14$ and $f_V(v)=315\,H_v(c)$ for $v\geqslant\frac14$, where
\begin{eqnarray*}
H_v(t)&=&\frac{t^{13}}{13}+\frac{6v-4}{11}t^{11}+\frac{5v^2-6v+2}{3}t^9+\frac{20v^3-32v^2+18v-4}{7}t^7\\
&&+\frac{15v^4-28v^3+19v^2-6v+1}{5}t^5+\frac{6v^5-12v^4+8v^3-2v^2}{3}t^3+(v^6-2v^5+v^4)t.
\end{eqnarray*}
The standard deviation has density $2x f_V(x^2)$.
\end{thrm}

\begin{proof}
Substitute $s=r^2+v$ in \eqref{eq:m101}; the Jacobian of $(r,s)\mapsto(r,v)$ is $1$. The support conditions
become $r^2+v\leqslant1$ and $r^2-|r|+v\geqslant0$. With $t=|r|$, the second condition holds for all $t$ when
$v\geqslant\frac14$ and means $t\leqslant\alpha_-(v)$ or $t\geqslant\alpha_+(v)$ when $v<\frac14$. The integrand is even in $r$,
which gives the factor $2$. For $K=3$, $2c_3=315$ and $H_v$ is the primitive of $h_v$ with $H_v(0)=0$. The mean
follows from Proposition~\ref{prop:moments} with $\sigma^2(A)=\frac23$.
\end{proof}

The threshold $v=\frac14$ is the largest variance on the lower boundary $x=\sqrt{(a_{k+1}-r)(r-a_k)}$ of
Corollary~\ref{cor:qutrit}: below it, the set of means compatible with a given variance splits into three
intervals (Figure~\ref{fig:variance}(a)).

\begin{figure}[t]
\centering
\includegraphics{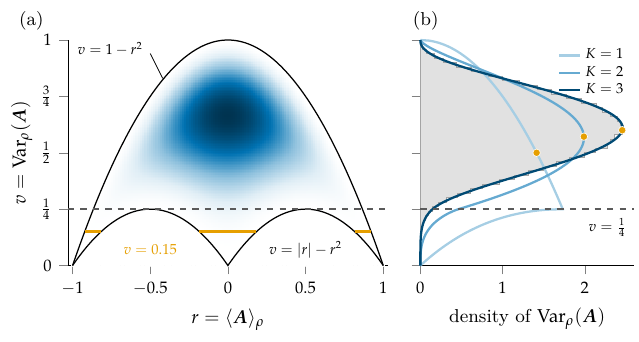}
\caption{The qutrit variance law of Theorem~\ref{th:variance}, $\bsA=\diag(-1,0,1)$. (a) Joint density of
$(\langle \bsA\rangle_\rho,\Var_\rho(\bsA))$ for $\rho\sim\mu_{3,3}$. The horizontal section at $v=0.15$ (orange)
meets the support in three intervals, while sections with $v\geqslant\frac14$ (dashed) meet it in one.
(b) Density of $\Var_\rho(\bsA)$ for $K=1,2,3$, drawn on the same vertical axis as (a); the histogram shows
$2\times10^5$ samples for $K=3$, and the orange dot on each curve marks its mean $\frac{2K}{3K+1}$. The density changes form at
$v=\frac14$.}\label{fig:variance}
\end{figure}

%===================================================%
\section{The overlap of two random states}\label{sect:overlap}
%===================================================%

Let $\rho\sim\mu_{d,K}$ and $\tau\sim\mu_{d,L}$ be independent and $X=\Tr{\rho\tau}$.

\begin{prop}\label{prop:overlapmoments}
$\bbE [X]=\frac1d$ and $\Var(X)=\dfrac{d^2-1}{d^2(dK+1)(dL+1)}$.
\end{prop}

\begin{proof}
Conditionally on $\tau$, Proposition~\ref{prop:moments} with $\bsA=\tau$ gives 
$$
\bbE[X|\tau]=\frac1d,\quad \Var[X|\tau]=\frac{\Tr{\tau^2}-\tfrac1d}{d(dK+1)}. 
$$
Average over $\tau$ with $\bbE[\Tr{\tau^2}]=\frac{d+L}{dL+1}$. The desired result follows.
\end{proof}

\subsection{Qubits: mode at the mean}

\begin{thrm}\label{th:qubitmode}
Let $\rho,\tau$ be independent random qubit states whose Bloch vectors have rotation-invariant laws, with
radii $A,B\in[0,1]$. Conditionally on $AB>0$, $X=\Tr{\rho\tau}$ is uniform on $[\frac{1-AB}2,\frac{1+AB}2]$. Consequently,
apart from an atom at $\frac12$ of mass $\Pr\Set{AB=0}$, $X$ has the density
\begin{eqnarray}\label{eq:qubitmix}
p_2(x)=\bbE\Br{\tfrac{1}{AB}\chi_{\set{|2x-1|<AB}}},
\end{eqnarray}
which is symmetric about $\frac12$ and non-increasing in $|2x-1|$. In particular $\frac12$ is both the mean
and a mode, and $p_2(\frac12)=\bbE[A^{-1}]\bbE[B^{-1}]$.
\end{thrm}

\begin{proof}
It is easily seen that 
$$
X=\Tr{\rho\tau}=\frac{1+AB\cos\Theta}2,
$$
where $\Theta$ is the angle between the Bloch vectors. By rotation invariance and
Archimedes' theorem, $\cos\Theta$ is uniform on $[-1,1]$ and independent of $(A,B)$. Averaging the conditional
uniform densities gives \eqref{eq:qubitmix}, and each conditional density is symmetric and non-increasing in
$|2x-1|$. Specifically, Conditionally on $A=a$ and $B=b$, $X$ is uniform on
$\Br{\frac{1-AB}2,\frac{1+AB}2}$ with density
$$
f_X(x|a,b)=\frac1{ab}\chi_{\abs{2x-1}\leqslant ab}.
$$
Averaging over $a$ and $b$ gives
$$
p_2(x)=\bbE[\tfrac1{AB}\chi_{\abs{2x-1}\leqslant AB}]=9\int^1_0\int^1_0 ab \chi_{\abs{2x-1}\leqslant ab}\dif a\dif
b.
$$
Let $r=\abs{2x-1}$. Then
\begin{eqnarray*}
p_2(x) &=& 9 \int^1_r \dif a a\int^1_{\frac ra}\dif b b =
\frac94(1-r^2)+\frac92 r^2\ln r.
\end{eqnarray*}
Since $1-r^2=4x(1-x)$, we recover
\begin{eqnarray}
p_2(x) = 9x(1-x)+\frac92(1-2x)^2\ln\abs{1-2x}.
\end{eqnarray}
At $x=\frac12$, the logarithmic term vanishes in the limiting sense
and $p_2(\frac12)=\frac94=\bbE[A^{-1}]\bbE[B^{-1}]$, where 
$$
\bbE[A^{-1}]=\int^1_0 \frac1a(3a^2)\dif a=\frac32=\bbE[B^{-1}]. \qedhere
$$
\end{proof}

\begin{cor}\label{cor:qubitoverlap}
For $\rho\sim\mu_{2,K}$ and $\tau\sim\mu_{2,L}$ with $K,L\geqslant2$, put $r:=|2x-1|$ and
$C_K:=2[\mathrm{B}(\frac32,K-1)]^{-1}$. Then
\begin{eqnarray}\label{eq:B1}
p_{K,L}(x)=\frac{C_KC_L}{2(L-1)}\int^1_r\dif a\Br{a(1-a^2)^{K-2}\Pa{1-\frac{r^2}{a^2}}^{L-1}},
\end{eqnarray}
where
\begin{eqnarray}
p_{K,L}(\tfrac12)=\frac4\pi\,\frac{\Gamma(K+\frac12)\Gamma(L+\frac12)}{\Gamma(K)\Gamma(L)}.
\end{eqnarray}
The value at $\frac12$ holds for all $K,L\geqslant1$. For $K=L=2$,
$$
p_2(x)\equiv p_{2,2}(x)=9x(1-x)+\frac92(1-2x)^2\log|1-2x|,\quad p_2(\tfrac12)=\frac94.
$$
\end{cor}

\begin{proof}
By Lemma~\ref{lem:archimedes}(i), Bloch length $A$ has density $C_Ka^2(1-a^2)^{K-2}$ on $[0,1]$. Integrating
\eqref{eq:qubitmix} first over $b\in(\frac ra,1]$ gives \eqref{eq:B1}. Moreover
$$
\bbE[A^{-1}]=\frac{C_K}{2(K-1)}=\frac{\Gamma(K+\frac12)}{\Gamma(\frac32)\Gamma(K)},
$$
and $A\equiv1$ when $K=1$.
\end{proof}

\subsection{Qutrits: the mode separates from the mean}

The symmetry used in Theorem~\ref{th:qubitmode} is special to $d=2$: the Bloch inversion
$\bsr\mapsto-\bsr$ maps the state space onto itself and sends $X=\Tr{\rho(\bsr)\tau}$ to $1-X=\Tr{\rho(-\bsr)\tau}$, where $\rho(\bsr)=\frac12(\I+\bsr\cdot\boldsymbol{\sigma})$. For $d=3$, there is no such
symmetry, and the mean is not a mode. We take the Hilbert--Schmidt case $K=L=3$, write $p_3$ for the density
of $X$, and use the representation
\begin{eqnarray}\label{eq:p3rep}
p_3(r)=\int_{\Omega_3}p_{3,3}(r|\bdlambda)\dif\nu_3(\bdlambda)=\int_{\Omega_3}p_3(r|\bdlambda)\dif\nu_3(\bdlambda),
\end{eqnarray}
which follows by conditioning on the spectrum $\bdlambda$ of $\tau$ and applying
Theorem~\ref{th:spline} with $\bsA=\tau$; here $\nu_3$ is \eqref{eq:hsorder} with $N_3=10080$, and
$p_{3,3}(r|\bdlambda)$ is given by \eqref{eq:qutritlocal}. We see that $0<a_1<\frac13$ and $a_1<a_2<\frac{1-a_1}2<\frac12$ and $a_3=1-a_1-a_2$.\\
%Note that $N_3=10080$.
(i) $r\in[\tfrac12,1)$. We have that the above integral is not
vanished if and only if $r\in(a_2,a_3)$, the open interval. Then the
domain of integral is given by the set of binary coordinates
\begin{eqnarray}
D(r)=D^{(1)}_{23}(r) =
\Set{(a_1,a_2):0<a_1<\tfrac{1-r}2,a_1<a_2<1-a_1-r}.
\end{eqnarray}
Thus
\begin{eqnarray}
p_3(r) = (-241920)\cdot\int^{\frac{1-r}2}_0\dif a_1
\int^{1-a_1-r}_{a_1}\dif a_2 \frac{c_3(r|\bsa)}{V_3(\bsa)^3}.
\end{eqnarray}
(ii) $r\in[\tfrac13,\tfrac12]$. Then the domain of integral is given
by
\begin{eqnarray}
D(r) = D^{(2)}_{23}(r)\cup D^{(3)}_{23}(r)\cup D^{(1)}_{12}(r).
\end{eqnarray}
where
\begin{eqnarray}
D^{(2)}_{23}(r) &=& \Set{(a_1,a_2):0<a_1<1-2r, a_1<a_2<r},\\
D^{(3)}_{23}(r) &=& \Set{(a_1,a_2):1-2r<a_1<\tfrac{1-r}2, a_1<a_2<1-a_1-r},\\
D^{(1)}_{12}(r) &=& \Set{(a_1,a_2):0<a_1<1-2r,r<a_2<\tfrac{1-a_1}2}.
\end{eqnarray}
Note that on $D^{(2)}_{23}(r)\cup D^{(3)}_{23}(r)$, we see
$r\in(a_2,a_3)$; on $D^{(1)}_{12}(r)$, we have $r\in(a_1,a_2)$. Thus
we have
\begin{eqnarray*}
p_3(r)&=& (-241920)\cdot\int^{1-2r}_0\dif a_1 \int^{r}_{a_1}\dif a_2
\frac{c_3(r|\bsa)}{V_3(\bsa)^3}+(-241920)\cdot\int^{\frac{1-r}2}_{1-2r}\dif
a_1
\int^{1-a_1-r}_{a_1}\dif a_2 \frac{c_3(r|\bsa)}{V_3(\bsa)^3}\\
&&+241920\cdot\int^{1-2r}_0\dif a_1 \int^{\frac{1-a_1}2}_r\dif a_2
\frac{c_1(r|\bsa)}{V_3(\bsa)^3}.
\end{eqnarray*}
(iii) $r\in(0,\tfrac13]$. Then the domain of integral is given by
\begin{eqnarray}
D(r) = D^{(4)}_{23}(r)\cup D^{(2)}_{12}(r),
\end{eqnarray}
where
\begin{eqnarray}
D^{(4)}_{23}(r) &=& \Set{(a_1,a_2):0<a_1<r, a_1<a_2<r},\\
D^{(2)}_{12}(r) &=& \Set{(a_1,a_2):0<a_1<r, r<a_2<\tfrac{1-a_1}2}.
\end{eqnarray}
Note that on $D^{(4)}_{23}(r)$, we see $r\in(a_2,a_3)$; on
$D^{(2)}_{12}(r)$, we have $r\in(a_1,a_2)$. Thus we have
\begin{eqnarray*}
p_3(r)&=& (-241920)\cdot\int^r_0\dif a_1 \int^{r}_{a_1}\dif a_2
\frac{c_3(r|\bsa)}{V_3(\bsa)^3}+241920\cdot\int^r_0\dif a_1
\int^{\frac{1-a_1}2}_r\dif a_2 \frac{c_1(r|\bsa)}{V_3(\bsa)^3}.
\end{eqnarray*}

\begin{thrm}\label{th:qutritmode}
The density $p_3$ is continuously differentiable near the mean $r=\frac13$, and
\begin{eqnarray}
p_3(\tfrac13)&=&-\frac{143560}{243}+\frac{43870}{81}\log3=4.23194778\ldots,\label{eq:p3center}\\
p_3'(\tfrac13)&=&-\frac{910420}{81}+\frac{276115}{27}\log3=-4.81486102\ldots<0.\label{eq:p3deriv}
\end{eqnarray}
Hence $\frac13$ is not a local maximum of $p_3$, and the mode of the Hilbert--Schmidt qutrit overlap
differs from its mean.
\end{thrm}

\begin{proof}
We can use the above expression of $p_3(r)$, by \textsc{Mathematica}, we can get that
\begin{eqnarray*}
p_3(\tfrac13)&=& 480\int^1_0 t^2(1+t)^2 q_t(\tfrac{2+t}3)\Pa{(2+t)^{-7}+(1+2t)^{-7}}\dif t\\
&=&-\frac{143560}{243}+\frac{43870}{81}\log3,
\end{eqnarray*}
where $q_t$ is from \eqref{eq:qt}. Parametrize $\Omega_3=\Set{\bdlambda=(\lambda_1,\lambda_2,\lambda_3):\lambda_3>\lambda_2>\lambda_1,\sum^3_{k=1}\lambda_k=1}$ by the spectral gaps $g=\lambda_2-\lambda_1$ and $h=\lambda_3-\lambda_2$:
\begin{eqnarray}\label{eq:gaps}
\lambda_1=\tfrac13(1-2g-h),\quad\lambda_2=\tfrac13(1+g-h),\quad\lambda_3=\tfrac13(1+g+2h),
\end{eqnarray}
with $g,h>0$, $2g+h<1$. The Lebesgue measure on $\Omega_3$ becomes $\dif\lambda_1\dif\lambda_2=\frac13\dif g\dif h$, and the Vandermonde is $V_3(\bdlambda)=(\lambda_2-\lambda_1)(\lambda_3-\lambda_1)(\lambda_3-\lambda_2)=gh(g+h)$. Hence 
\begin{eqnarray}\label{eq:nu3-gaps}
\dif\nu_3(\bdlambda) = 3360g^2 h^2 (g+h)^2\dif g\dif h,
\end{eqnarray}
where $g\in(0,\tfrac12)$ and $0<h<1-2g$.
\paragraph{Reduction to a one-parameter family.} The key observation is that the conditional density $p_{3,3}(r|\bdlambda)$ 
depends on $\bdlambda$ only through the gap ratio $t=h/g$ (or its 
reciprocal), up to an affine rescaling. To see this, fix $g>0$ and write 
$h=tg$ with $t>0$. Then
$$
\lambda_1=\tfrac{1}{3}(1-(2+t)g),\qquad
\lambda_2=\lambda_1+g,\qquad
\lambda_3=\lambda_1+(1+t)g.
$$
The spectrum of $\bsA=\tau$ is therefore an affine image of the reference 
spectrum $(0,1,1+t)$:
\begin{eqnarray} \label{eq:affine}
\lambda_k=\lambda_1+g\,a_k^{(t)},\qquad
(a_1^{(t)},a_2^{(t)},a_3^{(t)})=(0,1,1+t).
\end{eqnarray}
Since the map $\bsA\mapsto \alpha \bsA+\beta\I$ sends 
$\Tr{\bsA\rho}$ to $\alpha\Tr{\bsA\rho}+\beta$, the 
conditional density transforms as
\begin{eqnarray*}
p_3(r|\bdlambda)=g^{-1}q_t\!\left(\frac{r-\lambda_1}{g}\right), 
\label{eq:rescaling}
\end{eqnarray*}
where $q_t$ is the density \eqref{eq:qt} of $\Tr{\bsA_t\rho}$ for the 
reference observable
$\bsA_t=\diag(0,1,1+t)$, where $0<t\leqslant 1$,
and $\rho\sim\mu_{3,3}$. The prefactor $g^{-1}$ is the Jacobian of the 
affine rescaling; the shift by $\lambda_1$ is absorbed into the argument. 
For $t>1$, we use the reciprocal $t\mapsto t^{-1}$ with the spectrum 
reversed; see below. 
\paragraph{The reference density and its derivative at the mean.}
For the reference observable $\bsA_t=\diag(0,1,1+t)$, the mean of 
$\Tr{\bsA_t\rho}$ is $\frac{\Tr{\bsA_t}}3=\frac{2+t}3$. The centre
$r_0=\frac{2+t}3$ lies in $(\frac23,1]\subset[a_1,a_2]$, so \eqref{eq:qutritlocal} and \eqref{eq:p3rep} lead to
\begin{eqnarray}\label{eq:qt}
q_t(r)=\frac{24r^5}{(1+t)^3}
\Br{7-7r\Pa{1+\frac{1}{1+t}}
+r^2\Pa{(2+\frac{3}{1+t}+\frac{2}{(1+t)^2}}},
\quad r\in[0,1]. 
\end{eqnarray}
Note that $q_t$ is supported on $[0,1+t]$, but the expression above is the 
branch valid on $[0,1]$; the mean $\frac{2+t}3\in(\frac23,1]$ lies in this interval 
for $t\in(0,1]$, so \eqref{eq:qt} is the relevant branch. Differentiating 
\eqref{eq:qt} with respect to $r$ and evaluating at $r_0=\frac{2+t}3$ yields, 
after simplification,
\begin{equation}\label{eq:Dt}
\vartheta(t):=q'_t\Pa{\frac{2+t}{3}}
=\frac{56(t-1)(t+2)^4(2t+1)(t^2-t-1)}{243(t+1)^5}. 
\end{equation}
The derivation of \eqref{eq:Dt} is a straightforward but tedious 
differentiation; the key point is that the sign of $\vartheta(t)$ on $(0,1)$ is 
determined by the factor $t^2-t-1$, which is negative for $t\in(0,1)$ 
(since $t^2-t-1=0$ at $t=\frac{1+\sqrt{5}}2\approx 1.618$). All other factors 
are positive on $(0,1)$. Hence
\begin{eqnarray}\label{eq:Dt-sign}
\vartheta(t)>0\qquad\text{for all }t\in(0,1). 
\end{eqnarray}
This is the source of the negative derivative of $p_3$ at the mean.
\paragraph{From the conditional derivative to $p'_3(\frac13)$.}
We now assemble the integral representation of $p'_3(\frac13)$. Write 
$p'_3(\frac13)$ as the sum of two contributions, corresponding to the regions 
$h\leqslant g$ (i.e., $t\leqslant 1$) and $h\geqslant g$ (i.e., $t\geqslant 1$).

\emph{Region I: $h\leqslant g$.} Write $h=tg$ with $0<t\leqslant 1$. The constraint 
$2g+h<1$ becomes $0<g<\frac1{2+t}$. Thus $\dif g\dif h = g\dif g\dif h$ and the weight \eqref{eq:nu3-gaps} becomes
\begin{eqnarray}\label{eq:weight-I}
\dif \nu_3=3360t^2(1+t)^2g^7\dif g\dif t, \quad 0<g<\tfrac1{2+t}.
\end{eqnarray}
The conditional rescaling is $p_3(r|\bdlambda)=g^{-1}q_t\Pa{\frac{r-\lambda_1}g}$. The conditional derivative at the mean is, by the 
chain rule,
\begin{eqnarray} \label{eq:cond-deriv-I}
\partial_{r=\frac13} p_{3,3}(r|\bdlambda)=g^{-2}\vartheta(t),
\end{eqnarray}
because the argument $\frac{r-\lambda_1}g$ equals $\frac{2+t}3$ at $r=\frac13$ (by 
\eqref{eq:gaps} and \eqref{eq:affine}), and the derivative brings down a 
factor $g^{-1}$ from the rescaling, giving 
$g^{-1}\cdot g^{-1}\vartheta(t)=g^{-2}\vartheta(t)$. Therefore the contribution of Region I 
to $p'_3(\frac13)$ is
\begin{eqnarray}\label{eq:I1}
\cI_1=3360\int_0^1\dif t t^2(1+t)^2 \vartheta(t)\int_0^{\frac1{2+t}}\dif g g^5 = 560 \int^1_0 \dif t\frac{t^2(1+t)^2\vartheta(t)}{(2+t)^6}. 
\end{eqnarray}

\emph{Region II: $g\leqslant h$.} Write $g=th$ with $0<t\leqslant 1$. The 
constraint $2g+h<1$ becomes $h<\frac1{1+2t}$. The spectrum in this region is 
the reverse of the reference spectrum: if we set $h=g't$ with $g'=h$, then 
the gap ratio is $g'/h=1/t$, and the roles of $\lambda_1$ and 
$\lambda_3$ are interchanged. Equivalently, the conditional density satisfies
\begin{eqnarray}\label{eq:rescaling-II}
p_3(r|\bdlambda)=h^{-1}q_t\Pa{\frac{\lambda_3-r}{h}}, 
\end{eqnarray}
where $q_t$ is the same reference density \eqref{eq:qt} and the argument is 
$\frac{\lambda_3-r}h$. At $r=\frac13$, this argument again equals $\frac{2+t}3$ (by 
symmetry of the spectrum under reversal). The derivative brings down a factor 
$-h^{-1}$, so
\begin{eqnarray}\label{eq:cond-deriv-II}
\partial_{r=\frac13} p_{3,3}(r|\bdlambda)=-h^{-2}\vartheta(t). 
\end{eqnarray}
The weight in this region is the same as \eqref{eq:weight-I} with 
$g\leftrightarrow h$: 
$\dif\nu_3=3360t^2(1+t)^2h^7\dif h\dif t$. 
Therefore the contribution of Region II is
\begin{eqnarray}\label{eq:I2}
\cI_2=-560\int_0^1\frac{t^2(1+t)^2\vartheta(t)}{(1+2t)^6}\dif t. 
\end{eqnarray}

Combining \eqref{eq:I1} and \eqref{eq:I2}, substituting \eqref{eq:Dt} into \eqref{eq:total} and simplifying the rational 
integrand, we get 
\begin{eqnarray} \label{eq:total}
p'_3\Pa{\tfrac{1}{3}}
&=&560\int^1_0 t^2(1+t)^2 \vartheta(t)\Pa{\frac{1}{(2+t)^6}-\frac{1}{(1+2t)^6}}\dif t\\
&=&\int_0^1\frac{31360t^2(t-1)^2(t^2-t-1)(t^2+t+1)(7t^2+13t+7)}
{27(t+1)^2(t+2)^2(2t+1)^5}\dif t\\
&=&\frac{245}{81} (3381 \log 3-3716)<0
\end{eqnarray}
This proves that $\frac13$ is not a local maximum of $p_3$, so the mode 
differs from the mean.
\end{proof}

Numerically, $p_3$ attains its maximum $4.25626\ldots$ at $r\approx0.323275$
(Figure~\ref{fig:overlap}(b)). Since $\bbE[X]=\frac13$ and the variance is $\frac{8}{9\cdot100}$ by
Proposition~\ref{prop:overlapmoments}, the separation $\frac13-0.3233\approx0.010$ is about $0.11$ standard
deviations. 

\subsection{Large environments}

As $K,L\to\infty$ both states concentrate at $\frac{\I}d$, so $X\to\frac1d$; by Proposition~\ref{prop:overlapmoments} the
fluctuations are of order $(KL)^{-\frac12}$, and their limit law is not Gaussian.

\begin{thrm}\label{th:wishart}
Fix $d\geqslant2$ and let $K,L\to\infty$. Then
$$
d^2\sqrt{KL}\Pa{\Tr{\rho\sigma}-\frac1d}\ \Longrightarrow\ \sum_{j=1}^{d^2-1}\xi_j\eta_j,
$$
where $\xi_j,\eta_j$ are independent standard normal variables. The limit is a variance-gamma law with
characteristic function $(1+t^2)^{-\frac{d^2-1}2}$ and density
$$
\frac{|x|^\nu K_\nu(|x|)}{\sqrt\pi\,2^\nu\,\Gamma(\nu+\frac12)},\qquad \nu=\frac{d^2-2}2,
$$
where $K_\nu$ is the modified Bessel function of the second kind, see Appendix~\ref{app:2ndBessel}.
\end{thrm}

\begin{proof}
The detailed proof can be described as follows:
\begin{itemize}
\item \textbf{Step 1: Gaussian fluctuation of the Wishart matrix.} Let $\bsG$ be a $d \times K$ complex Ginibre matrix with independent standard complex Gaussian entries, so that
$\bsW = \bsG\bsG^\dagger = \sum_{k=1}^K \bsg_k \bsg_k^\dagger$, where $\bsg_1, \ldots, \bsg_K \in \bbC^d$ are independent standard complex Gaussian vectors. Each $\bsg_k$ satisfies
$$
\bbE[\bsg_k \bsg_k^\dagger] = \I_d, \qquad \op{Var}(\Innerm{\bsg_k} {\bsX}{\bsg_k}) = \Tr{\bsX^2}
$$
for every Hermitian matrix $\bsX \in \Herm(\bbC^d)$. By the multivariate central limit theorem applied to the i.i.d. rank-one matrices $\bsg_k \bsg_k^\dagger$, we have
$$
\sqrt{K}\Pa{\frac{\bsW}{K} - \I_d} \Longrightarrow \bsZ,
$$
where $\bsZ$ is a Gaussian Hermitian matrix with
$$
\bbE[\Tr{\bsX\bsZ}] = 0, \qquad \bbE[( \Tr{\bsX\bsZ})^2] =  \Tr{\bsX^2}
$$
for all Hermitian $\bsX$. Equivalently, in an orthonormal basis $\set{\bsE_j}_{j=1}^{d^2}$ of $\Herm(\bbC^d)$ with respect to the Hilbert--Schmidt inner product, the coordinates $\Tr{\bsZ\bsE_j}$ are independent standard normal variables (with the trace coordinate having variance $d$).
\item \textbf{Step 2: The delta method for the normalized state.} The induced random state is
$\rho = \frac{\bsW}{\Tr{\bsW}}$. Define the map $\Phi: \Herm(\bbC^d) \to \Herm(\bbC^d)$ by
$\Phi(\bsW) = \frac{\bsW}{\Tr{\bsW}}$. Its differential at $\bsW_0 = K\I_d$ is
$$
D\Phi_{K\I_d}(\bsH) = \frac{\bsH}{Kd} - \frac{\Tr{\bsH}}{Kd^2}\I_d.
$$
Indeed, for $\bsW = K\I_d + \sqrt{K}\bsH$ with $\bsH$ fixed,
\begin{eqnarray*}
\Phi(\bsW) &=& \frac{K\I_d + \sqrt{K}\bsH}{Kd + \sqrt{K}\Tr{\bsH}}
= \frac{1}{d}\Pa{\I_d + \frac{\bsH}{\sqrt{K}}}\Pa{1 + \tfrac{\Tr{\bsH}}{d\sqrt{K}}}^{-1}\\
&=& \frac{1}{d}\I_d + \frac{1}{d\sqrt{K}}\Pa{\bsH - \frac{\Tr{\bsH}}{d}\mathbb{1}_d} + O(K^{-1}).
\end{eqnarray*}
Thus, by the delta method, $\sqrt{K}(\rho - \frac{1}{d}\I_d) \Longrightarrow \frac{1}{d}\bsZ_0$, where
$\bsZ_0 = \bsZ - \frac{\Tr{\bsZ}}{d}\I_d$ is the traceless part of $\bsZ$.
\item \textbf{Step 3: Independence and the overlap expansion.} Let $\rho \sim \mu_{d,K}$ and $\sigma \sim \mu_{d,L}$ be independent, constructed from independent Ginibre matrices $\bsG$ and $\bsG'$. Write
$$
\rho = \frac{1}{d}\I_d + \frac{1}{d\sqrt{K}}\bsZ_0^{(K)} + o(K^{-1/2}),
\qquad
\sigma = \frac{1}{d}\I_d + \frac{1}{d\sqrt{L}}\bsZ_0^{(L)} + o(L^{-1/2}),
$$
where $\bsZ_0^{(K)} \Rightarrow \bsZ_0$ and $\bsZ_0^{(L)} \Rightarrow \bsZ_0'$ are independent copies of the traceless Gaussian matrix.
Since $\rho - \frac{1}{d}\I_d$ and $\sigma - \frac{1}{d}\I_d$ are traceless,
$$
\Tr{\rho\sigma} - \tfrac{1}{d}
= \Tr{((\rho - \tfrac{1}{d}\I_d)(\sigma - \tfrac{1}{d}\I_d)}.
$$
Substituting the expansions,
$$
d^2\sqrt{KL}\Pa{\Tr{\rho\sigma} - \tfrac{1}{d}}
= \Tr{\bsZ_0^{(K)}\bsZ_0^{(L)}} + o(1).
$$
By independence and the continuous mapping theorem,
$\Tr{\bsZ_0^{(K)}\bsZ_0^{(L)}} \Longrightarrow \Tr{\bsZ_0 \bsZ_0'}$, where $\bsZ_0, \bsZ_0'$ are independent copies of the traceless Gaussian matrix.
\item \textbf{Step 4: Representation as a sum of products of normals.} Let $\Set{\bsE_j}_{j=1}^{d^2-1}$ be an orthonormal basis of the traceless Hermitian matrices with respect to the Hilbert--Schmidt inner product $\Inner{\bsX}{\bsY}=\Tr{\bsX\bsY}$. Then
$$
\bsZ_0 = \sum_{j=1}^{d^2-1} \xi_j \bsE_j, \qquad \bsZ_0' = \sum_{j=1}^{d^2-1} \eta_j \bsE_j,
$$
where $\xi_j, \eta_j$ are independent standard normal variables. Indeed, the coordinates $\Tr{\bsZ_0 \bsE_j}$ are independent standard normals because $\bsZ_0$ is a Gaussian matrix with covariance $\bbE[\Tr{\bsX\bsZ_0}\Tr{\bsY\bsZ_0}] = \Tr{\bsX\bsY}$ for traceless $\bsX, \bsY$. Therefore,
$\Tr{\bsZ_0 \bsZ_0'} = \sum_{j=1}^{d^2-1} \xi_j \eta_j$. This establishes the first assertion of the theorem.
\item \textbf{Step 5: Identification of the variance-gamma law.} Let $S = \sum_{j=1}^{d^2-1} \xi_j \eta_j$ with $\xi_j, \eta_j \sim \cN(0,1)$ independent. Conditionally on $\eta = (\eta_1, \ldots, \eta_{d^2-1})$,
$$
S \mid \eta \sim \cN\Pa{0, \sum_{j=1}^{d^2-1} \eta_j^2}.
$$
Since $\sum_{j=1}^{d^2-1} \eta_j^2 \sim \chi^2_{d^2-1}$, the characteristic function of $S$ is
$$
\phi_S(t) = \bbE\Pa{e^{-\frac{t^2}{2}\sum_{j=1}^{d^2-1} \eta_j^2}}= (1 + t^2)^{-\frac{d^2-1}{2}}.
$$
This is exactly the characteristic function of the variance-gamma law with parameters $\nu = \frac{d^2-2}{2}$ and $\lambda = 1$. Its density is
$$
f_S(x) = \frac{|x|^\nu K_\nu(|x|)}{\sqrt{\pi}2^\nu \Gamma(\nu + \frac{1}{2})}, \qquad \nu = \frac{d^2-2}{2},
$$
where $K_\nu$ is the modified Bessel function of the second kind.
\item \textbf{Step 6: Consistency of the limiting variance.} From Proposition~\ref{prop:overlapmoments}, the variance of $X = \Tr{\rho\sigma}$ is
$$
\op{Var}(X) = \frac{d^2-1}{d^2(dK+1)(dL+1)}.
$$
Thus
$$
\op{Var}\left(d^2\sqrt{KL}\left(X - \frac{1}{d}\right)\right)
= d^4 KL \cdot \frac{d^2-1}{d^2(dK+1)(dL+1)}
\to d^2-1.
$$
The limiting variance of $S = \sum_{j=1}^{d^2-1} \xi_j \eta_j$ is
$$
\op{Var}(S) = \sum_{j=1}^{d^2-1} \op{Var}(\xi_j \eta_j) = d^2-1,
$$
which agrees with the limit above.
\item \textbf{Step 7: Non-Gaussianity and tail behavior.} The characteristic function $(1+t^2)^{-\frac{d^2-1}{2}}$ is not of the form $e^{-ct^2}$ for any $c > 0$, so the limit is not Gaussian for any fixed $d$. The density has exponential tails because $K_\nu(|x|) \sim \sqrt{\frac{\pi}{2|x|}} e^{-|x|}$ as $|x| \to \infty$, so
$$
f_S(x) \sim \frac{|x|^{\nu-\frac12} e^{-|x|}}{2^{\nu+\frac12}\Gamma(\nu + \tfrac{1}{2})}
$$
for large $|x|$, confirming exponential decay rather than Gaussian decay. For $d = 2,  \nu = 1$, and the density is
$f_S(x) = \frac{|x| K_1(|x|)}{\pi}$, which is $C^1$ but not $C^2$ at the origin, since $K_1(|x|)$ has a logarithmic singularity in its derivative at $x = 0$.
\end{itemize}
Therefore the limit is a variance-gamma law with characteristic function $(1+t^2)^{-\frac{d^2-1}{2}}$ and density
$$
\frac{|x|^\nu K_\nu(|x|)}{\sqrt{\pi}2^\nu \Gamma(\nu + \frac{1}{2})}, \qquad \nu = \frac{d^2-2}{2}.
$$
This completes the proof. 
\end{proof}

The limiting variance $d^2-1$ agrees with Proposition~\ref{prop:overlapmoments}. For $d=2$, the limiting density is $C^1$, but not $C^2$, at the origin, and the limiting density has exponential tails for every $d$; it is not Gaussian for any fixed $d$
(Figure~\ref{fig:overlap}(c)). A Gaussian limit for expectation values with prescribed eigenvalues in
growing dimension is proved in \cite{Meckes2020}. Central limit theorems for linear spectral statistics of
large sample covariance matrices are given in \cite{Bai2004}.

\begin{figure}[t]
\centering
\includegraphics{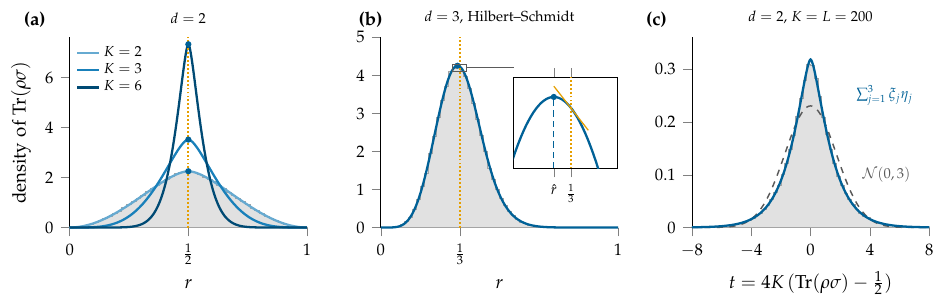}
\caption{Mode and mean of the overlap $\Tr{\rho\sigma}$. Dotted lines: means; dots: modes; histograms:
sampling. (a) Qubits, $\rho,\sigma\sim\mu_{2,K}$, Corollary~\ref{cor:qubitoverlap}: for every $K$ the mode
is at the mean $\frac12$, with height $\frac4\pi\Gamma(K+\frac12)^2/\Gamma(K)^2$. (b) Hilbert--Schmidt
qutrits, Theorem~\ref{th:qutritmode}: the inset magnifies the top of the density. The mode $\hat r\approx0.3233$
lies below the mean $\frac13$, where the tangent (orange) has slope $p_3'(\frac13)<0$. (c) $d=2$,
$K=L=200$: the standardized overlap follows the variance-gamma limit of Theorem~\ref{th:wishart} and not the
Gaussian with the same variance (dashed).}\label{fig:overlap}
\end{figure}

%======================================================================================%
\section{Uncertainty: expectation-value bodies and their laws}\label{sect:uncertainty}
%======================================================================================%

For $\bsA,\bsB\in\Herm(\bbC^d)$ the uncertainty region is
$\cU_{\bsA,\bsB}=\set{(\Delta_\rho(\bsA),\Delta_\rho(\bsB)):\rho\in\rD(\bbC^d)}$ \cite{Zhang2021,Zhang2022}. It is
the image of the compact convex set
$$
\cC_{\bsA,\bsB}=\Set{\pa{\langle \bsA\rangle_\rho,\langle \bsA^2\rangle_\rho,\langle \bsB\rangle_\rho,\langle \bsB^2\rangle_\rho}:\rho\in\rD(\bbC^d)}
$$
under $(r_1,s_1,r_2,s_2)\mapsto\Pa{\sqrt{s_1-r_1^2},\sqrt{s_2-r_2^2}}$. The support function of 
$\cC_{\bsA,\bsB}$ is
$$
h(\bst)=\lambda_{\max}(t_1\bsA+t_2\bsA^2+t_3\bsB+t_4\bsB^2),
$$
because $\max_\rho\Tr{\rho \bsH}=\lambda_{\max}(\bsH)$. The image
$\cU_{\bsA,\bsB}$ need not be convex.

\subsection{The optimal sum of variances}

\begin{prop}[\cite{Schwonnek2017}]\label{prop:sdw}
For $\bsA,\bsB\in\Herm(\bbC^d)$,
\begin{eqnarray}\label{eq:sdw}
\min_{\rho}\Br{\Var_\rho (\bsA)+\Var_\rho (\bsB)}=\min_{u\in I_{\bsA},\ v\in I_{\bsB}}\lambda_{\min}\pa{(\bsA-u\I)^2+(\bsB-v\I)^2},
\end{eqnarray}
where $I_{\bsA}=[\lambda_{\min}(\bsA),\lambda_{\max}(\bsA)]$ and $I_{\bsB}$ is defined likewise. If $c(u,v)$ denotes the
function minimized on the right and $\set{(u_i,v_j)}$ is a grid with mesh $\delta_u,\delta_v$ covering
$I_{\bsA}\times I_{\bsB}$, then $\min_{i,j}c(u_i,v_j)-w_{\bsA}\delta_u-w_{\bsB}\delta_v$ is a lower bound for \eqref{eq:sdw}, where
$w_{\bsA}=\lambda_{\max}(\bsA)-\lambda_{\min}(\bsA)$.
\end{prop}

\begin{proof}
For every state, $\Tr{\rho(\bsA-u\I)^2}=\Var_\rho (\bsA)+(u-\langle \bsA\rangle_\rho)^2$, and similarly for $\bsB$; minimizing over
$u,v$ first and then over $\rho$ gives \eqref{eq:sdw}, with optimal $u=\langle \bsA\rangle_\rho\in I_{\bsA}$. For
$u,u'\in I_{\bsA}$, $\norm{(\bsA-u\I)^2-(\bsA-u'\I)^2}=|u-u'|\,\norm{2\bsA-u\I-u'\I}\leqslant2w_{\bsA}|u-u'|$, so $c$ is Lipschitz with
constants $2w_{\bsA}$ and $2w_{\bsB}$; every point of $I_{\bsA}\times I_{\bsB}$ is within $\delta_u/2,\delta_v/2$ of a grid point.
\end{proof}

The right-hand side of \eqref{eq:sdw} can vanish for noncommuting $\bsA,\bsB$ when they share an eigenvector.
For qubits, it is $1-|c|$ in the setting below.

\subsection{Two qubit observables}

Let $\bsA=\bsa\cdot\bdsigma$ and $\bsB=\bsb\cdot\bdsigma$ with unit vectors $\bsa,\bsb$ and $c=\Inner{\bsa}{\bsb}$, $|c|<1$.
For a qubit, $\Var_\rho(\bsA)=1-\langle \bsA\rangle_\rho^2$, so everything is determined by $(U,V)=(\langle \bsA\rangle_\rho,\langle \bsB\rangle_\rho)$.

\begin{thrm}\label{th:ellipse}
For $\rho\sim\mu_{2,K}$ for $K\geqslant1$, the pair $(U,V)$ has the joint density
\begin{eqnarray}\label{eq:U3}
g_{K,c}(u,v)=\frac{K-\frac12}{\pi\sqrt{1-c^2}}\Pa{1-\frac{u^2-2cuv+v^2}{1-c^2}}_+^{K-\frac32}.
\end{eqnarray}
\end{thrm}

\begin{proof}
Any qubit density matrix $\rho$ can be written in the Bloch representation as $\rho=\frac12(\I_2+\bsr\cdot\bdsigma)$. For $\rho\sim \mu_{2,K}$, Lemma~\ref{lem:archimedes} (i) states that
\begin{itemize}
\item $\bsr$ is rotation invariant;
\item For $K=1$, $\bsr$ is uniform on the unit sphere $S^2$;
\item For $K\geqslant2$, $\bsr$ has density proportional to $(1-\abs{\bsr}^2)^{K-2}$ on the unit ball $B^3$.
\end{itemize}
For $\bsA=\bsa\cdot\bdsigma$ with $\abs{\bsa}=1$, the expectation value in state $\rho$ is $\langle\bsA\rangle_\rho=\Inner{\bsa}{\bsr}$.  Similarly, $\langle\bsB\rangle_\rho=\Inner{\bsb}{\bsr}$. Therefore $(U,V)=(\Inner{\bsa}{\bsr},\Inner{\bsb}{\bsr})$. The pair $(U,V)$ depends on $\bsr$ only through its projection onto the two-dimensional subspace $S=\op{span}\set{\bsa,\bsb}$. Write $\bsr=\zeta+\xi$, where $\zeta\in S$ is the orthogonal projection of $\bsr$ onto $S$, and $\xi\in S^\perp$ is the orthogonal complement. Since $\bsa,\bsb\in S$, we have $\Inner{\bsa}{\bsr}=\Inner{\bsa}{\zeta}$ and $\Inner{\bsb}{\bsr}=\Inner{\bsb}{\zeta}$. Thus $(U,V)$ depends only on $\zeta$.

In what follows, we derive the distribution of the projection $\zeta$. By Lemma~\ref{lem:archimedes} (i), the Bloch vector $\bsr$ has a rotation-invariant distribution. Consequently, its orthogonal projection onto any two-dimensional subspace $S$ has the same distribution as the projection onto the $(r_1,r_2)$-plane.

From Lemma~\ref{lem:archimedes}(i), the density of the projection $\zeta=(\zeta_1,\zeta_2)$ onto any plane is 
$$
g_K(\zeta_1,\zeta_2) = \frac{K-\frac12}\pi \Pa{1-(\zeta^2_1+\zeta^2_2)}^{K-\frac32}, \quad \zeta^2_1+\zeta^2_2<1.
$$
This is the density of $\zeta$ with respect to the Lebesgure measure on $S\cong\bbR^2$. Next perform change of variables to $(U,V)$. Let $\set{\bse_1,\bse_2}$ be an orthonormal basis of $S=\op{span}\set{\bsa,\bsb}$. Write $\zeta=\xi_1\bse_1+\xi_2\bse_2$, so that $\zeta$ has coordinates $(\xi_1,\xi_2)$ in this basis. The density of $(\xi_1,\xi_2)$ is 
$$
g_K(\xi_1,\xi_2) = \frac{K-\frac12}\pi \Pa{1-(\xi^2_1+\xi^2_2)}^{K-\frac32}, \quad \xi^2_1+\xi^2_2<1.
$$
Now express $\bsa$ and $\bsb$ in the basis $\set{\bse_1,\bse_2}$. Since $\abs{\bsa}=\abs{\bsb}=1$ and $\Inner{\bsa}{\bsb}=c$, we can choose the basis such that $\bsa=\bse_1$ and $\bsb=c\bse_1+\sqrt{1-c^2}\bse_2$. Then $U=\Inner{\bsa}{\zeta}=\xi_1$ and $V=\Inner{\bsb}{\xi}=c\xi_1+\sqrt{1-c^2}\xi_2$. Solving for $(\xi_1,\xi_2)$ in terms of $(U,V)$:
$$
\xi_1=U,\quad \xi_2=\frac{V-cU}{\sqrt{1-c^2}}.
$$
The Jacobian of the transformation $(\xi_1,\xi_2)\mapsto (U,V)$ is $\frac1{\sqrt{1-c^2}}$. The density of $(U,V)$ is obtained from the density of $(\xi_1,\xi_2)$ via the change-of-variable formula:
$$
g_{K,c}(u,v) = g_K(\xi_1(u,v),\xi_2(u,v))\Abs{\frac{\partial(\xi_1,\xi_2)}{\partial(u,v)}}.
$$
Substituting the expressions for $\xi_1$ and $\xi_2$: $\xi^2_1+\xi^2_2=u^2+\frac{(v-cu)^2}{1-c^2}=\frac{u^2-2cuv+v^2}{1-c^2}$. Thus
we get the desired \eqref{eq:U3}.

The support condition is $\xi^2_1+\xi^2_2<1$, which becomes $u^2-2cuv+v^2<1-c^2$, the interior of the ellipse $u^2-2cuv+v^2\leqslant1-c^2$. The quatratic form 
$$
u^2-2cuv+v^2 = (u,v)\bsM\Pa{u,v}^\t,\quad \bsM=\begin{pmatrix}1&-c\\-c&1\end{pmatrix},
$$
with eigenvalues $1\pm c$. Since $|c|<1$, both eigenvalues are positive, so the support is a non-degenerate ellipse centered at the origin. The veryfication of normalization is trivially.
\end{proof}

\begin{thrm}\label{th:twostd}
For $\rho\sim\mu_{2,K}$, the standard deviations $(X,Y)=(\Delta_\rho(\bsA),\Delta_\rho(\bsB))$ have the joint density
\begin{eqnarray}\label{eq:U4}
f_{X,Y}(x,y)=\frac{xy}{\sqrt{(1-x^2)(1-y^2)}}\sum_{\epsilon_1,\epsilon_2=\pm1}g_{K,c}\Pa{\epsilon_1\sqrt{1-x^2},\epsilon_2\sqrt{1-y^2}}
\end{eqnarray}
for all $x,y\in(0,1)$, supported on 
$$
\Set{(x,y):x^2+y^2+2|c|\sqrt{(1-x^2)(1-y^2)}\geqslant1+c^2}.
$$
For $c=0$,
$$
f_{X,Y}(x,y)=\frac{4K-2}{\pi}\,\frac{xy}{\sqrt{(1-x^2)(1-y^2)}}(x^2+y^2-1)^{K-\frac32},\qquad x^2+y^2>1.
$$
\end{thrm}

\begin{proof}
For a qubit state, $\op{Var}_\rho(\bsA)=1-\langle\bsA\rangle^2_\rho$. This follows from
$$
\op{Var}_\rho(\bsA)=\Tr{\bsA^2\rho} - [\Tr{\bsA\rho}]^2=\Tr{\rho} - \langle\bsA\rangle^2_\rho
$$
since $\bsA^2=(\bsa\cdot\bdsigma)^2=\I_2$ for a unit vector $\bsa$. Therefore
$$
X=\Delta_\rho(\bsA)=\sqrt{1-U^2},\quad Y=\Delta_\rho(\bsB)=\sqrt{1-V^2}.
$$
The map $(u,v)\mapsto (x,y)=(\sqrt{1-u^2},\sqrt{1-v^2})$ is four-to-one on $(-1,1)^2\backslash \set{uv=0}$, with the four branches corresponding to the sign choices $u=\varepsilon_1\sqrt{1-x^2}$ and $v=\varepsilon_2\sqrt{1-y^2}$, where $\varepsilon_1,\varepsilon_2\in\set{\pm1}$. 

We need the Jacobian of the transformation $(u,v)\mapsto (x,y)$. Compute:
$$
x=\sqrt{1-u^2}\implies \frac{\partial x}{\partial u} = \frac{-u}{\sqrt{1-u^2}}=-\frac ux.
$$
Similarly, $\frac{\partial y}{\partial u} =-\frac vy$. The Jacobian matrix is diagonal:
$$
\bsJ=\diag\Pa{\frac{\partial x}{\partial u},\frac{\partial y}{\partial u}}=\diag\Pa{-\frac ux,-\frac vy}.
$$
Therefore $\abs{\det(\bsJ)}=\frac{\abs{uv}}{xy}$. The inverse Jacobian (from $(x,y)$ to $(u,v)$) has absolute value $\Abs{\frac{\partial(u,v)}{\partial(x,y)}}=\frac{xy}{\abs{uv}}$. Since $u=\varepsilon_1\sqrt{1-x^2}$ and $v=\varepsilon_2\sqrt{1-y^2}$, we have 
$$
\Abs{\frac{\partial(u,v)}{\partial(x,y)}}= \frac{xy}{\sqrt{(1-x^2)(1-y^2)}}.
$$
This factor is the same for all four sign branches. The density of $(X,Y)$ is obtained by summing the density of $(U,V)$ over all preimages
$$
f_{X,Y}(x,y) =\Abs{\frac{\partial(u,v)}{\partial(x,y)}} \sum_{\varepsilon_1,\varepsilon_2=\pm1}g_{K,c}\Pa{\varepsilon_1\sqrt{1-x^2},\varepsilon_2\sqrt{1-y^2}},
$$
which is exactly \eqref{eq:U4}.
\end{proof}

For pure states ($K=1$) and for Hilbert--Schmidt states ($K=2$), the joint law of $(\Delta_\rho (\bsA),\Delta_\rho (\bsB))$ was
obtained in \cite{Zhang2022} and \cite{Zhang2021}, respectively; Theorem~\ref{th:twostd} covers every $K$
with one formula, the new content here, but the derivation is a straightforward change of variables from Lemma~\ref{lem:archimedes} (i). The support is the uncertainty region of the pair, and its point closest to the origin
gives the optimal bound $\Var_\rho (\bsA)+\Var_\rho (\bsB)\geqslant1-|c|$: maximizing $u^2+v^2$ over the ellipse
$u^2-2cuv+v^2\leqslant1-c^2$ gives $1+|c|$. The law also shows how much probability lies near this bound.

\begin{thrm}\label{th:nearopt}
Let $\rho\sim\mu_{2,K}$ for $K\geqslant1$, and $S=\Var_\rho(\bsA)+\Var_\rho(\bsB)$ and $S_{\min}=1-|c|$.
\begin{enumerate}[(i)]
\item If $c=0$, then $S=2-Q$ with $Q\sim\Beta(1,K-\frac12)$; in particular $\Pr\Set{S\leqslant1+\varepsilon}=\varepsilon^{K-\frac12}$ for
$0\leqslant\varepsilon\leqslant1$.
\item If $0<|c|<1$, then as $\varepsilon\downarrow0$
\begin{eqnarray}\label{eq:nearopt}
\Pr\Set{S\leqslant S_{\min}+\varepsilon}=\sqrt{\frac{2}{|c|}}\,\frac{\Gamma(K+\frac12)}{2\sqrt\pi\,\Gamma(K+1)}\,
(1+|c|)^{\frac12-K}\,\varepsilon^K\,\pa{1+O(\varepsilon)}.
\end{eqnarray}
\end{enumerate}
\end{thrm}

\begin{proof}
It is easily seen that $S=2-\Inner{\bsa}{\bsr}^2 - \Inner{\bsb}{\bsr}^2=2-U^2-V^2=2-T$, where $T=U^2+V^2$. So the event $\set{S\leqslant S_{\min}+\varepsilon}$ is equivalent to $\set{T\geqslant 2-S_{\min}-\varepsilon}$. Since $S_{\min}=1-\abs{c}$, we have
$$
2-S_{\min} = 2-(1-\abs{c})=1+\abs{c}.
$$
Thus $\set{S\leqslant S_{\min}+\varepsilon}=\Set{T\geqslant 1+\abs{c}-\varepsilon}$. Recall from the proof of Theorem~\ref{th:ellipse} that
$$
(U,V)^\t=\bsM\xi,\quad \bsM=\begin{pmatrix} 1&0\\ c&\sqrt{1-c^2}\end{pmatrix},
$$
where $\boldsymbol{\xi}=(\xi_1,\xi_2)$ is the projection of the Bloch vector onto $\op{span}\set{\bsa,\bsb}$, with rotation-invariant density
$$
g_K(\xi_1,\xi_2) = \frac{K-\frac12}\pi\Pa{1-\xi^2_1-\xi^2_2}^{K-\frac32},\quad\xi^2_1+\xi^2_2<1.
$$
Now $T=U^2+V^2=\boldsymbol{\xi}^\t\bsM^\t\bsM\boldsymbol{\xi}$, the eigenvalues of $\bsM^\t\bsM$ are $\lambda_\pm =1\pm\abs{c}$, with corresponding orthonormal eigenvectors. Let $\bsO$ be the orthogonal matrix diagonalizing $\bsM^\t\bsM$: $\bsO^\t\bsM^\t\bsM\bsO = \diag(1+\abs{c},1-\abs{c})$. Set $\boldsymbol{\eta}=\bsO^\t\boldsymbol{\xi}$. Since $\bsO$ is orthogonal and $g_K$ is rotation invariant, $\boldsymbol{\eta}$ has the same density $g_K$:
$$
g_K(\eta_1,\eta_2) = \frac{K-\frac12}\pi\Pa{1-\eta^2_1-\eta^2_2}^{K-\frac32},\quad\eta^2_1+\eta^2_2<1.
$$
Then $T=(1+\abs{c})\eta^2_1+(1-\abs{c})\eta^2_2$.
\begin{itemize}
\item When $c=0$, $1+\abs{c}=1-\abs{c}=1$, so $T=\eta^2_1+\eta^2_2=\abs{\boldsymbol{\eta}}^2=:Q$. Since $\boldsymbol{\eta}$ has density $g_K$, the radial density of $Q$ is obtained by integrating over the circle $\abs{\boldsymbol{\eta}}^2=q$:
\begin{eqnarray*}
f_Q(q) = \frac{K-\frac12}\pi (1-q)^{K-\frac32}\cdot 2\pi\sqrt{q}\cdot\frac1{2\sqrt{q}}=\Pa{K-\frac12}(1-q)^{K-\frac32},
\end{eqnarray*}
for $q\in(0,1)$. Therefore $Q=\abs{\boldsymbol{\eta}}^2\sim\Beta(1,K-\frac12)$. Since $S=2-T=2-Q$, where $Q\sim\Beta(1,K-\frac12)$. Now compute the tail probability. For $\varepsilon\in[0,1]$, 
\begin{eqnarray*}
\Pr\Set{S\leqslant 1+\varepsilon} =\Pr\Set{Q\geqslant 1-\varepsilon}= \int^1_{1-\varepsilon} (K-\tfrac12)(1-q)^{K-\frac32}\dif q = \varepsilon^{K-\frac12}.
\end{eqnarray*}
\item When $0<\abs{c}<1$, $T=(1+\abs{c})\eta^2_1+(1-\abs{c})\eta^2_2$, and we need the asymptotic behavior of $\Pr\Set{T\geqslant 1+\abs{c}-\varepsilon}$ as $\varepsilon\downarrow0$. Set $w=1-\eta^2_1-\eta^2_2=1-\abs{\boldsymbol{\eta}}^2$. We see that $w\in(0,1)$. Then $T=(1-w)(1+\abs{c})-2\abs{c}\eta^2_2$, and thus $1+\abs{c}-T=(1+\abs{c})w+2\abs{c}\eta^2_2$. So the event $\Set{S\leqslant S_{\min}+\varepsilon}=\Set{T\geqslant1+\abs{c}-\varepsilon}$ becomes $(1+\abs{c})w+2\abs{c}\eta^2_2\leqslant\varepsilon$.

Consider the change of variables from $(\eta_1,\eta_2)$ to $(w,\eta_2)$. We have $\eta_1=\pm\sqrt{1-w-\eta^2_2}$. The map is two-to-one (for $\eta_1>0$ and $\eta_1<0$). The Jacobian:
$$
\Abs{\frac{\partial \eta_1}{\partial w}}=\frac1{2\sqrt{1-w-\eta^2_2}}=\frac1{2\abs{\eta_1}}.
$$
So $\dif \eta_1=\frac{\dif w}{2\abs{\eta_1}}$. The density of $\boldsymbol{\eta}$ is $g_K(\eta_1,\eta_2)=\frac{K-\frac12}\pi w^{K-\frac32}$. Thus the joint density of $(w,\eta_2)$ (summing over the two signs of $\eta_1$) is
$$
f_{w,\eta_2}(w,\eta_2) = 2\cdot \frac{K-\frac12}\pi w^{K-\frac32}\cdot\frac1{2\abs{\eta_1}} = \frac{K-\frac12}\pi \frac{w^{K-\frac32}}{\abs{\eta_1}}.
$$
But $\abs{\eta_1}=\sqrt{1-w-\eta^2_2}$. On the event $(1+\abs{c})w+2\abs{c}\eta^2_2\leqslant\varepsilon$, we have $w\leqslant \frac{\varepsilon}{1+\abs{c}}$ and $\abs{\eta_2}\leqslant \sqrt{\frac{\varepsilon}{2\abs{c}}}$. Therefore
$$
1-w-\eta^2_2=1+O(\varepsilon)\quad \text{as }\varepsilon\downarrow 0.
$$
More precisely, $\frac1{\abs{\eta_1}}=\frac1{\sqrt{1-w-\eta^2_2}}=1+O(\varepsilon)$. So
$$
f_{w,\eta_2}(w,\eta_2) = \frac{K-\frac12}\pi w^{K-\frac32}(1+O(\varepsilon)).
$$
Now
\begin{eqnarray*}
\Pr\Set{S\leqslant S_{\min}+\varepsilon}&=&\iint_{(1+\abs{c})w+2\abs{c}\eta^2_2\leqslant\varepsilon}f_{w,\eta_2}(w,\eta_2)\dif w\dif \eta_2\\
&=& (1+O(\varepsilon))\frac{K-\frac12}\pi \int_{(1+\abs{c})w+2\abs{c}\eta^2_2\leqslant\varepsilon} w^{K-\frac32}\dif w\dif\eta_2.
\end{eqnarray*}
The domain of integration can be reduced as:
$$
0\leqslant w\leqslant \frac{\varepsilon-2\abs{c}\eta^2_2}{1+\abs{c}},\qquad \abs{\eta_2}\leqslant\sqrt{\frac{\varepsilon}{2\abs{c}}}.
$$
Thus
\begin{eqnarray*}
\Pr\Set{S\leqslant S_{\min}+\varepsilon}&=& (1+O(\varepsilon))\frac{K-\frac12}\pi  \int^{\sqrt{\frac{\varepsilon}{2\abs{c}}}}_{-\sqrt{\frac{\varepsilon}{2\abs{c}}}} \dif\eta_2\int^{\frac{\varepsilon-2\abs{c}\eta^2_2}{1+\abs{c}}}_0 w^{K-\frac32}\dif w\\
&=&  (1+O(\varepsilon))\frac{K-\frac12}\pi \frac1{K-\frac12} \frac1{(1+\abs{c})^{K-\frac12}}\frac{\varepsilon^K}{2\abs{c}}\frac{\sqrt{\pi}\Gamma(K+\tfrac12)}{\Gamma(K+1)},
\end{eqnarray*}
implying the desired result \eqref{eq:nearopt} after simplyfication.\qedhere
\end{itemize}

\end{proof}

The exponents $K-\frac12$ and $K$ reflect the dimension of the set of optimal states. 
\begin{itemize}
\item For $c=0$, the states attaining
$S_{\min}=1$ are the pure states whose
Bloch vectors lie on the great circle lies in $\op{span}\set{\bsa,\bsb}$. This is a one-dimensional set (a circle). The  probability of being within $\varepsilon$ of the optimum scales as $\varepsilon^{K-\frac12}$, where the exponent reflects the codimension of the optimal set in the Bloch ball (dimension 3), shifted by the $K$-dependent concentration.
\item for $c\neq0$, there are only two optimal
pure states, with Bloch vectors $\bsr=\pm\frac{\bsa\pm\bsb}{|\bsa\pm\bsb|}$ for the appropriate sign. This is a zero-dimensional set. The probability scales as $\varepsilon^K$, with a larger exponent because the optimal set is smaller (codimension $2$ instead of $1$).
\end{itemize}
Enlarging the optimal set by
one dimension lowers the exponent by $\frac12$. Increasing $K$ suppresses near-optimal states because they must
be nearly pure (Figure~\ref{fig:uncertainty}).

\begin{figure}[t]
\centering
\includegraphics[width=\textwidth]{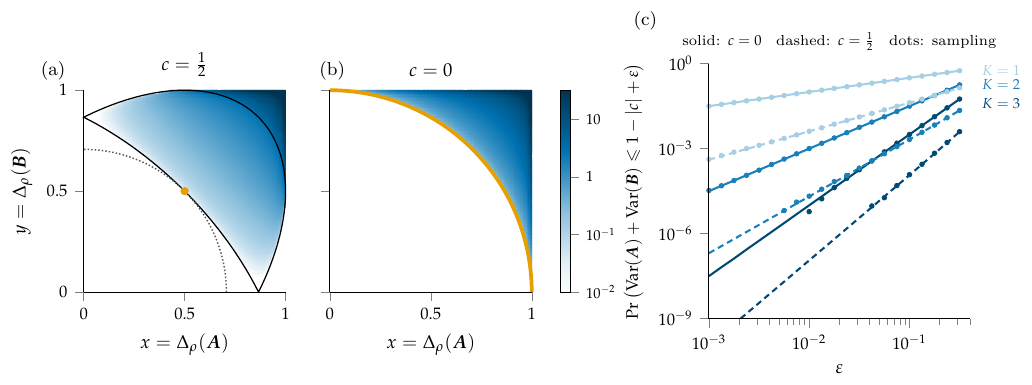}
\caption{Probability laws on the uncertainty region of $\bsA=\bsa\cdot\bdsigma$, $\bsB=\bsb\cdot\bdsigma$, $c=\bsa\cdot\bsb$.
(a), (b) Joint density of $(\Delta_\rho(\bsA),\Delta_\rho (\bsB))$ for $\rho\sim\mu_{2,2}$, Theorem~\ref{th:twostd}, on a
logarithmic colour scale; the support is bounded by the black curve. Orange marks the states attaining the
optimal sum of variances $1-|c|$: a single point for $c=\frac12$, where the dotted circle
$x^2+y^2=\frac12$ touches the region, and a whole arc for $c=0$. (c) Probability that
$\Var_\rho(\bsA)+\Var_\rho(\bsB)$ is within $\varepsilon$ of the optimum. Solid lines ($c=0$) are the exact power laws
$\varepsilon^{K-\frac12}$; dashed lines ($c=\frac12$) are exact probabilities computed from \eqref{eq:U3}, which approach slope $K$
as in \eqref{eq:nearopt}; dots are estimates from $4\times10^6$ samples.}\label{fig:uncertainty}
\end{figure}

%===================================================%
\section{Random Bargmann invariants}\label{sect:bargmann}
%===================================================%

This section is the most novel in the paper, Proposition~\ref{prop:conditional} (conditional laws are numerical shadows) is essentially a restatement of known facts about numerical ranges of rank-one matrices, combined with the decomposition of a Haar vector into a 2D subspace (Lemma~\ref{lem:archimedes} (ii)). The real new contributions are the exact density for three qubit states (Theorem~\ref{th:b3density}), the marginals (Theorem~\ref{th:marginals}), the negative-real-part probability (Theorem~\ref{th:negative}), the Markov chain (Theorem~\ref{th:markov}), and the long-cycle limit (Theorem~\ref{th:longcycle}).

Let $\ket{\psi_1},\ldots,\ket{\psi_n}$ be independent Haar-random unit vectors in $\bbC^d$ for $d\geqslant2$, and
$\psi_j=\proj{\psi_j}$. Inner products are conjugate-linear in the first argument. The cyclic Bargmann invariant of the first $n$ vectors is
$$
\Delta_n=\Inner{\psi_1}{\psi_2}\Inner{\psi_2}{\psi_3}\cdots\Inner{\psi_{n-1}}{\psi_n}\Inner{\psi_n}{\psi_1}=\Tr{\psi_1\psi_2\cdots \psi_n}.
$$
It is invariant under phase changes of the $\ket{\psi_j}$ and under a common unitary. Its law is invariant under
complex conjugation (reverse the order of the states), so $\Im(\Delta_n)$ is symmetric. Since
$\bbE[\psi_j]=\I_d/d$, $\bbE[\Delta_n]=d^{1-n}$; the law is not rotation invariant in the complex plane. For $n\geqslant3$, we seek the planar density
\begin{eqnarray*}
{\Pr}_d\Set{\Delta_n\in[\dif z]} = f^{(d)}_n(z)[\dif z],\quad [\dif z]=\dif \Re(z)\dif\Im(z).
\end{eqnarray*}

\subsection{Conditional laws are numerical shadows}

Let $C_k=\prod_{j=1}^{k-1}\Inner{\psi_j}{\psi_{j+1}}$ be the open-chain product, with $C_1=1$. Then
$\psi_1\cdots \psi_k=C_k\out{\psi_1}{\psi_k}$ and
\begin{eqnarray}\label{eq:shadowrep}
\Delta_{k+1}=\Innerm{\psi_{k+1}}{\bsM_k}{\psi_{k+1}},\qquad \bsM_k=\psi_1\cdots \psi_k=C_k\out{\psi_1}{\psi_k}.
\end{eqnarray}

\begin{prop}\label{prop:conditional}
Conditionally on $\psi_1,\ldots,\psi_k$ with $\psi_1,\psi_k$ linearly independent and $C_k\neq0$, the law of
$\Delta_{k+1}$ is the numerical shadow of the rank-one matrix $\bsM_k$. It is supported on the closed elliptical disk
$W(\bsM_k)$ with foci $0$ and $\Tr{\bsM_k}=\Delta_k$, major axis $|C_k|$ and minor axis $\sqrt{|C_k|^2-|\Delta_k|^2}$
\cite{Zhang2025}. If $\Delta_k\neq0$, let $e^{\mathrm{i}\theta}=\frac{\Delta_k}{\abs{\Delta_k}}$, and rotate the centred coordinate $\zeta=e^{-\mathrm{i}\theta}(z-\frac{\Delta_k}2)=X+\mathrm{i}Y$. Define the affine elliptical radius $R_(z)$ as follows. With $a=\frac{|C_k|}2, b=\frac12\sqrt{|C_k|^2-|\Delta_k|^2}$,
\begin{eqnarray}\label{eq:affineer}
[R(z)]^2 = \frac{X^2}{a^2} +  \frac{Y^2}{b^2} = \frac{4\Pa{\Re(e^{-\mathrm{i}\theta}(z-\frac{\Delta_k}2))}^2}{\abs{C_k}^2} + \frac{4\Pa{\Im(e^{-\mathrm{i}\theta}(z-\frac{\Delta_k}2))}^2}{\abs{C_k}^2-\abs{\Delta_k}^2}.
\end{eqnarray}
If $\Delta_k=0$, then $R(z)=\frac{2\abs{z}}{\abs{C_k}}$. The support condition is exactly $z\in W(\bsM_k)$ if and only if $R(z)\leqslant1$.
\begin{enumerate}[(i)]
\item For $d=2$, its density of $\Delta_{k+1}$ is 
\begin{eqnarray}\label{eq:shadow2}
\cF_2(z|\bsM_k)=\frac{1}{2\pi ab(1-R(z)^2)^{\frac12}}\chi_{W(\bsM_k)}(z).
\end{eqnarray}
\item For $d\geqslant3$, the density of $\Delta_{k+1}$ is
\begin{eqnarray}
\cF_d(z | \bsM_k)=(d-1)(d-2)\int_0^1\tfrac{(1-t)^{d-3}}{t} \cF_2\Pa{\tfrac zt \big|\bsM_k}\dif t,
\end{eqnarray}
where $\cF_2(\cdot|\bsM_k)$ is the $d=2$ density above for the restriction of $\bsM_k$ to $\spn\set{\psi_1,\psi_k}$; the support is still the closed disk $W(\bsM_k)$, with foci $0$ and $\Delta_k$, major axis $\abs{C_k}$, and minor axis $\sqrt{\abs{C_k}^2-\abs{\Delta_k}^2}$.
\end{enumerate}
\end{prop}

\begin{proof}
\begin{itemize}
\item \textbf{Conditional law of $\Delta_{k+1}$.} Let $\psi_1,\ldots,\psi_k$ be fixed unit vectors in $\bbC^d$, and let
$$
C_k=\prod_{j=1}^{k-1}\Inner{\psi_j}{\psi_{j+1}},\qquad
\bsM_k=\psi_1\cdots\psi_k=C_k\out{\psi_1}{\psi_k}.
$$
Assume $\psi_1,\psi_k$ are linearly independent and $C_k\neq0$. Let $\psi_{k+1}$ be Haar-random on the unit sphere of $\bbC^d$, independent of $\psi_1,\ldots,\psi_k$. Then $\Delta_{k+1}=\Innerm{\psi_{k+1}}{\bsM_k}{\psi_{k+1}}$. This is exactly the numerical shadow of the rank-one matrix $\bsM_k$. Let
$$
S=\operatorname{span}\set{\psi_1,\psi_k},\qquad \dim_{\bbC}S=2.
$$
Since $\bsM_k$ vanishes on $S^\perp$ and maps $S$ into itself, we may decompose
$$
\ket{\psi_{k+1}}=\bsu+\bsv,\qquad \bsu\in S,\quad \bsv\in S^\perp.
$$
Then $\Delta_{k+1}=\Innerm{\bsu}{\bsM_k}{\bsu}$. For a Haar-random unit vector in $\bbC^d$, the squared norm of its projection onto a fixed $2$-dimensional subspace is
$$
T:=\norm{\bsu}^2\sim \op{Beta}(2,d-2),
$$
and, conditionally on $T>0$, the normalized vector
$\phi=\frac{\bsu}{\norm{\bsu}}$ is Haar-random on the unit sphere of $S\cong\bbC^2$. Hence
$$
\Delta_{k+1}=T\,\zeta,\qquad \zeta:=\Innerm{\phi}{\bsM_k}{\phi},
$$
where $T\sim\op{Beta}(2,d-2)$ is independent of $\zeta$. For $d=2$, we have $T=1$, so $\Delta_{k+1}=\zeta$. For $d\geqslant3$, the density of $\Delta_{k+1}$ is the convolution of the density of $T$ and the density of $\zeta$.
\item \textbf{The numerical range of $\bsM_k$.} Write
$\bsM_k=C_k\out{\psi_1}{\psi_k}$. Its trace is
$\Tr{\bsM_k}=C_k\Inner{\psi_k}{\psi_1}=\Delta_k$. Its unique nonzero singular value is
$|C_k|$. For a rank-one matrix $\bsM=|C|\out{\bsa}{\bsb}$ with $\norm{\bsa}=\norm{\bsb}=1$), the numerical range $W(\bsM)$ is the closed elliptical disk with foci $0$ and $\Tr{\bsM}$, major axis $|C|$, and minor axis
$\sqrt{|C|^2-|\Tr{\bsM}|^2}$.
Thus $W(\bsM_k)$ has foci $0$ and $\Delta_k$, major axis $|C_k|$, and minor axis $\sqrt{|C_k|^2-|\Delta_k|^2}$. Set
$$
a=\frac{|C_k|}{2},\qquad
b=\frac12\sqrt{|C_k|^2-|\Delta_k|^2}.
$$
Then $a$ and $b$ are the semi-major and semi-minor axes of $W(\bsM_k)$.
\item \textbf{The case $d=2$.} For $d=2$, the density of $\zeta=\Innerm{\phi}{\bsM_k}{\phi}$ for Haar-random $\phi\in\bbC^2$ is the numerical shadow of the $2\times2$ matrix $\bsM_k$. By a unitary conjugation and, if necessary, a phase rotation $\bsM_k\mapsto e^{\mathrm{i}\varphi}\bsM_k$, we may assume that
$$
\bsM_k=
\begin{pmatrix}
t & s\\
0 & 0
\end{pmatrix},
\qquad t\geqslant0,\quad s\geqslant0.
$$
The phase rotation only rotates the numerical shadow in the complex plane and leaves its density unchanged.
Let $\ket{\phi}=(x_1,x_2)$ with $|x_1|^2+|x_2|^2=1$. Write
$u=|x_1|^2$ and $\varphi=\arg(x_2)-\arg(x_1)$. Then $u\sim\mathrm{Uniform}(0,1)$ and $\varphi\sim\mathrm{Uniform}(0,2\pi)$ are independent, and
$\zeta=t u+s\sqrt{u(1-u)} e^{i\varphi}$. Equivalently, with $u=\frac{1+\cos\theta}2, \theta\in[0,\pi]$,
$$
\zeta=\frac t2+\frac t2\cos\theta+\frac s2\sin\theta\,e^{\mathrm{i}\varphi}.
$$
The joint density of $(\theta,\varphi)$ is $\frac{1}{4\pi}\sin\theta\dif \theta\dif\varphi$. Changing variables to $z=\zeta$ and summing over the two preimages of a point $z$ inside the ellipse gives the density
\begin{eqnarray}\label{eq:densityk1}
\cF_2(z|\bsM_k)=\frac{1}{2\pi ab\sqrt{1-R(z)^2}}\chi_{W(\bsM_k)}(z),
\end{eqnarray}
where $R(z)$ is the affine elliptical radius\footnote{Here the affine elliptical radius $R(z)$ defined in Eq.~\eqref{eq:affineer} is different from the quantity $R_c(z)=\frac{\abs{z}+\abs{z-\Delta_k}}{\abs{C_k}}$, called the confocal elliptical radius.}. This is the standard density of the numerical shadow of a $2\times2$ matrix; the calculation is a straightforward change of variables and simplification using the quadratic equation defining the preimages. The support is exactly the closed elliptical disk $W(\bsM_k)$.
\item \textbf{The case $d\geqslant3$.} For $d\geqslant3$, we have $\Delta_{k+1}=T\zeta$, where $T\sim\op{Beta}(2,d-2)$ is independent of $\zeta$, and $\zeta$ has density $\cF_2(\cdot|\bsM_k)$ on the $2$-dimensional subspace $S$ by
Lemma~\ref{lem:archimedes}(ii). The density of $T$ is
$$
f_T(t)=\frac{1}{\rB(2,d-2)}t(1-t)^{d-3}=(d-1)(d-2)t(1-t)^{d-3},\qquad 0<t<1.
$$
If $Z=T\zeta$, then conditionally on $T=t$, the density of $Z$ is $t^{-2}\cF_2(\frac zt|\bsM_k)$, because scaling a $2$-dimensional random variable by $t$ multiplies its density by $t^{-2}$. Therefore the unconditional density is
$$
\cF_d(z|\bsM_k)=\int^1_0 f_T(t)\tfrac{1}{t^2}\cF_2(\tfrac zt |\bsM_k)\dif t.
$$
Substituting $f_T(t)$,
$$
\cF_d(z|\bsM_k)
=
(d-1)(d-2)
\int_0^1
\tfrac{(1-t)^{d-3}}{t}
\cF_2(\tfrac zt|\bsM_k)\dif t.
$$
The support remains the closed elliptical disk $W(\bsM_k)$, because each ellipse $tW(\bsM_k), 0<t\leqslant1$, is contained in the convex set $W(\bsM_k)$, and their union over $t\in(0,1]$ is exactly $W(\bsM_k)\setminus\set{0}$, with $0$ included as a limit point.\qedhere
\end{itemize}
\end{proof}

The ellipses $W(\bsM_k)$ are the building blocks of the sets of Bargmann invariants studied in
\cite{Zhang2025,Pratapsi2025,Xu2026,Fernandes2024}; Proposition~\ref{prop:conditional} adds the probability law on
each of them.

\subsection{Three states: support and exact density}

Let $0<\alpha<\beta<\gamma$ and 
$$
I=\int_\alpha^\beta\frac{\dif u}{\sqrt{u(u-\alpha)(\beta-u)(\gamma-u)}}.
$$
Define
\begin{eqnarray}
\sE_1(x) = \int_0^{\frac\pi2}\frac{\dif\theta}{\sqrt{1-x^2\sin^2\theta}}.
\end{eqnarray}
This is the definition of the \emph{complete elliptic integral of the first kind}. Use the substitution $\sin^2\theta =
\frac{\beta(u-\alpha)}{u(\beta-\alpha)}$. Solving for $u$ gives
$$
u = \frac{\alpha\beta}{\beta-(\beta-\alpha)\sin^2\theta}.
$$
As $u$ runs from $\alpha$ to $\beta$, $\theta$ runs from $0$ to
$\frac\pi2$. Direct simplification yields
$$
I=\frac{2}{\sqrt{\beta(\gamma-\alpha)}}\int_0^{\frac\pi2}\frac{\dif\theta}{\sqrt{1-\kappa^2\sin^2\theta}}=\frac{2}{\sqrt{\beta(\gamma-\alpha)}} \sE_1(\kappa),
$$
where $\kappa^2=\frac{\gamma(\beta-\alpha)}{\beta(\gamma-\alpha)}$. In numerical libraries, $\sE_1$ is often parametrized by $\kappa^2$ rather than by the modulus $\kappa$ \cite{DLMF}.

For $m=3$, $\bsM_2=\Inner{\psi_1}{\psi_2}\out{\psi_1}{\psi_2}$ has trace $U=|\Inner{\psi_1}{\psi_2}|^2$ and
$|C_2|=\sqrt U$, where $U\sim\Beta(1,d-1)$. The conditional support is the ellipse
\begin{eqnarray}\label{eq:Eu}
\cE_u:=\Set{x+\mathrm{i}y:\frac{(x-u/2)^2}{u/4}+\frac{y^2}{u(1-u)/4}\leqslant1},
\end{eqnarray}
with foci $0$ and $u$. Equivalently $y^2\leqslant F_x(u)$ with $F_x(u):=\frac14(1-u)\pa{u(1-u+4x)-4x^2}$.

\begin{prop}\label{prop:support}
For every $d\geqslant2$, the support of $\Delta_3$ is
\begin{eqnarray*}
\cB_3&=&\bigcup_{0\leqslant u\leqslant1}E_u=\Set{z\in\bbC:-\frac18\leqslant \Re(z)\leqslant 1, 27|z|^2\leqslant(1+2\Re(z))^3}\\
&=&\Set{re^{\mathrm{i}\varphi}:0\leqslant r\leqslant\tfrac18\sec^3(\tfrac{\pi-|\varphi|}3), |\varphi|\leqslant\pi}.
\end{eqnarray*}
\end{prop}

\begin{proof}
For fixed $x$, $F_x$ is maximized at $u^*=(1+2x)/3$, with $F_x(u^*)=(1-x)^2(1+8x)/27$, and
$27(x^2+y^2)\leqslant(1+2x)^3$ is the same inequality as $y^2\leqslant(1-x)^2(1+8x)/27$. The support of $U$ is $[0,1]$ for every
$d$, and by Proposition~\ref{prop:conditional} the conditional support is all of $E_u$.
\end{proof}

The set $\cB_3$ is known from \cite{Fernandes2024,Zhang2025}; the curve $27|z|^2=(1+2\Re(z))^3$ has a cusp at $z=1$
and passes through $z=-\frac18$ (Figure~\ref{fig:bargmann}(a)). For qubits, the conditional density is
explicit, and averaging over $U\sim\mathrm{Uniform}[0,1]$ produces an elliptic integral. Put $t=1+2x$ and
$\varrho=|z|$, and consider the cubic
$$
C_z(u)=u(u-t)^2-4\varrho^2=4\pa{F_x(u)-y^2}.
$$
The equation $C_z(u)=0$, i.e., 
$u^3 - 2t u^2 + t^2 u-4\varrho^2=0$,  has discriminant $16\varrho^2(t^3-27\varrho^2)>0$ for $z\neq0$ in the interior of $\cB_3$. Let $u=v+\frac{2t}3$. We get that
$$
v^3 -\frac{t^2}3v - 2\Pa{2\varrho^2-\frac{t^3}{27}}=0.
$$
The depressed cubic has the form $v^3-3pv-2q=0$, where $p=\frac{t^2}9$ and $q=2\varrho^2-\frac{t^3}{27}$. For three real roots, we use the trignometric substitution $v=2\sqrt{p}\cos\theta=\frac{2t}3\cos\theta$. Substituting $v$ into the depressed cubic $v^3-3pv-2q=0$:
\begin{eqnarray*}
0=\Pa{\frac{2t}3\cos\theta}^3 - \frac{t^2}3\Pa{\frac{2t}3\cos\theta}-2\Pa{2\varrho^2 - \frac{t^3}{27}}= \frac{8t^3}{27}\cos^3\theta - \frac{2t^3}9\cos\theta - 4\varrho^2+\frac{2t^3}{27}.
\end{eqnarray*}
Multiply by $\frac{27}{2t^3}$: $4\cos^3\theta - 3\theta = \frac{54\varrho^2}{t^3}-1$. Using the triple angle identity $\cos(3\theta)=4\cos^3\theta-3\cos\theta$: $\cos(3\theta)=\frac{54\varrho^2}{t^3}-1$. Let  
$$
\theta_0 :=\frac13\arccos\Pa{\frac{54\varrho^2}{t^3}-1} \in[0,\tfrac\pi3].
$$
The $3\theta=3\theta_0+2\pi k$ for integer $k$. Denote $\theta_k=\theta_0+\frac{2\pi k}3$ for $k=0,1,2$. Therefore $C_z(u)=0$ has three real roots $0<\alpha<\frac t3<\beta<t<\gamma$, and
\begin{eqnarray}\label{eq:roots}
\set{\alpha,\beta,\gamma}=\Set{\tfrac{2t}3\pa{1+\cos(\theta_0+\tfrac{2\pi k}3)}:k=0,1,2}.
\end{eqnarray}
with $\gamma$ corresponding to $k=0$.

\begin{thrm}\label{th:b3density}
For $d=2$ and $z=x+\mathrm{i}y\neq0$ in the interior of $\cB_3$, $\Delta_3$ has the density
\begin{eqnarray}\label{eq:b3density}
f^{(d)}_3(z)=\frac1\pi\int_{\set{F_x(u)>y^2}}\frac{\dif u}{\sqrt u\sqrt{F_x(u)-y^2}}
=\frac{4}{\pi\sqrt{\beta(\gamma-\alpha)}}\sE_1(\kappa),\qquad
\kappa^2=\frac{\gamma(\beta-\alpha)}{\beta(\gamma-\alpha)},
\end{eqnarray}
where $\sE_1(\kappa)=\int_0^{\pi/2}(1-\kappa^2\sin^2\theta)^{-1/2}\dif\theta$.
\end{thrm}

\begin{proof}
Write $\psi_j\equiv\proj{\psi_j}$ and let $\psi_1,\psi_2,\psi_3$ be independent Haar-random qubit pure states. By Proposition~\ref{prop:conditional}, conditional on $\psi_1,\psi_2$ with $U:=\abs{\Inner{\psi_1}{\psi_2}}^2\in(0,1)$ and $C_2=\Inner{\psi_1}{\psi_2}\neq0$, the rank-one matrix $\bsM_2=C_2\out{\psi_1}{\psi_2}$ has numerical range equal to the closed elliptical disk $\cE_u$ with foci $0$ and $U=\abs{C_2}^2$, major semi-axis $a=\sqrt{u}$ and minor semi-axis $b=\frac12\sqrt{u(1-u)}$. Since $d=2$, the projection factor $T\equiv1$ (Lemma~\ref{lem:archimedes}(ii) gives $T\sim \Beta(2,0)\equiv1$), so $\Delta_3=\Innerm{\psi_3}{\bsM_2}{\psi_3}$ is distributed exactly as the numerical shadow of $\bsM_2$ on $\cE_u$, namely
$$
\cF_2(z|\bsM_2) = \frac1{2\pi ab\sqrt{1-R(z)^2}}\chi_{W(\bsM_2)}(z) = \frac1{\pi\sqrt u\sqrt{F_x(u)-y^2}}\chi_{\set{F_x(u)>y^2}},
$$
where $R(z)$ is the affine elliptical radius of $\cE_u$, because $1-R^2(z)=\frac{4(F_x(u)-y^2)}{u(1-u)}$, and $(x,y)\in\cE_u$ if and only if $y^2\leqslant F_x(u)$. The random variable $U=\abs{\Inner{\psi_1}{\psi_2}}^2$ for independent Haar qubit states has the law 
$\Beta(1,1)=\op{Uniform}(0,1)$ when $d=2$. Therefore
$$
f^{(2)}_3(z) = \int^1_0\cF_2(z|\bsM_2)\dif u = \frac1\pi\int _{u\in(0,1):F_x(u)>y^2}\frac{\dif u}{\sqrt{u}\sqrt{F_x(u)-y^2}}.
$$
This is the first equality in \eqref{eq:b3density}. The support condition $F_x(u)>y^2$ for some $u\in(0,1)$ is precisely $z\in\cB_3$. Set $t=1+2x$ and $\varrho=\abs{z}=\sqrt{x^2+y^2}$. A direct computation gives 
$$
f^{(2)}_3(z) = \frac2\pi \int_{u:C_z(u)>0}\frac{\dif u}{\sqrt{C_z(u)}}.
$$
The cubic $C_z(u)=u(u-t)^2-4\varrho^2$ satisfies $C_z(0)=C_z(t)=-4\varrho^2<0$, and $C_z(\frac t3)=4\Pa{\frac{t^3}{27}-\varrho^2}$. Inside $\cB_3$, one has $27\varrho^2<(1+2x)^3=t^3$, so $C_z(\frac t3)>0$. Therefore $C_z$ has three real roots 
$$
0<\alpha<\frac t3<\beta<t<\gamma,
$$
and $C_z(u)>0$ exactly on $(\alpha,\beta)\cup(\gamma,\infty)$; since $C_z(1)=-4y^2\leqslant0$ (with $y\neq0$ in the interior), the relevant interval is $(\alpha,\beta)$ with $\beta\leqslant1$. Hence 
$$
f^{(2)}_3(z) = \frac2\pi \int^\beta_\alpha \frac{\dif u}{\sqrt{u(u-\alpha)(\beta-u)(\gamma-u)}}.
$$
After simplification, we get the desired result.
\end{proof}

\begin{prop}\label{prop:local}
For $d=2$, $f^{(d)}_3(z)=\frac6\pi\log\frac1{|z|}+o(1)$ as $z\to0$. As $z$ tends to a boundary point of $\cB_3$ other than
$-\frac18$ and $1$, $f^{(d)}_3(z)\to\frac{2\sqrt3}{1+2\Re(z)}$.
\end{prop}

\begin{proof}
Note that $z\to0$ if and only if $\varrho=\abs{z}\to0$. At the same time, $t=1+2x=1+2\Re(z)\to1$. More precisely $t=1+O(\varrho)\to1$.
The cubic is $C_z(u)=u(u-t)^2-4\varrho^2$. Its roots are the solutions of $u(u-t)^2=4\varrho^2$.
\begin{itemize}
\item Small root $\alpha$: near $u=0$, we have $u(u-t)^2\approx ut^2$, so $ut^2\approx 4\varrho^2$, giving $\alpha=\frac{4\varrho^2}{t^2}+O(\varrho^4)$.
\item Roots near $t$: write $u=t+w$. Then $u(u-t)^2=tw^2+w^3$. Setting this equal to $4\varrho^2$ gives $tw^2+w^3=4\varrho^2$. For small $w$, $w^2\approx \frac{4\varrho^2}{t}$, so $w\approx \pm \frac{2\varrho}{\sqrt{t}}$. A first-order correction: writing $w=\pm\frac{2\varrho}{\sqrt{t}}+\delta$, one finds 
$$
tw^2+w^3=4\varrho^2\pm \frac{12\varrho^2}{\sqrt{t}}\delta+O(\varrho^3)
$$
hence $\delta=\mp \frac{w^3}{6\varrho\sqrt{t}}\cdot \sqrt{t}\cdot \frac{\sqrt{t}}{2\varrho}$. Moreover,
$\beta,\gamma=t\mp \frac{2\varrho}{\sqrt{t}} + O(\varrho^2)$. The two roots straddle $t$; since $\alpha<\frac t3<\beta<t<\gamma$, we assign 
$$
\beta=t-\frac{2\varrho}{\sqrt{t}}+O(\varrho^2),\qquad \gamma=t+\frac{2\varrho}{\sqrt{t}}+O(\varrho^2).
$$
In particular, $\beta\to t,\gamma\to t,\gamma-\alpha\to t,\beta(\gamma-\alpha)\to t^2$, and $\gamma-\beta=\frac{4\varrho}{\sqrt{t}}+O(\varrho^2)$.
\end{itemize}
Using the identity $1-\kappa^2=\frac{\alpha(\gamma-\beta)}{\beta(\gamma-\alpha)}$ and the asymptotics above:
$$
1-\kappa^2= \frac{\Pa{\frac{4\varrho^2}{t^2}+O(\varrho^4)}\Pa{\frac{4\varrho}{\sqrt{t}}+O(\varrho^2)}}{t^2+O(\rho)} = \frac{16\varrho^3}{t^{\frac92}}(1+O(\varrho)).
$$
Hence 
$$
1-\kappa^2=\frac{16\varrho^3}{t^{\frac92}}(1+O(\varrho)),\quad \kappa\to 1^-.
$$
We need the standard asymptotic expansion of $\cE_1(\kappa)$ as $\kappa\to 1^-$. Write $\kappa'=\sqrt{1-\kappa^2}$. The complete elliptical integral of the first kind has the expansion
$$
\cE_1(\kappa) = \log\frac4{\kappa'} + O({\kappa'}^2\log\kappa').
$$
Therefore $\cE_1(\kappa)=\log\frac4{\sqrt{1-\kappa^2}}+o(1)$. Furthermore, $\sqrt{1-\kappa^2}=\frac{4\varrho^\frac32}{t^{\frac94}},(1+O(\varrho))$, so 
$$
\log\frac4{\sqrt{1-\kappa^2}} = \frac94\log t-\frac32\log\varrho+o(1). 
$$
Since $t\to1,\frac94\log t\to0$, and we obtain $\cE_1(\kappa)=-\frac32\log\varrho+o(1)$. Now the prefactor $\frac4{\pi\sqrt{\beta(\gamma-\alpha)}}\to \frac4{\pi t}\to\frac4\pi$. Therefore
$$
f^{(2)}_3(z)=\frac4{\pi\sqrt{\beta(\gamma-\alpha)}}\cE_1(\kappa) = \frac4{\pi t}\Pa{-\frac32\log\varrho}+o(1) = -\frac6{\pi t}\log\varrho+o(1).
$$
Since $t\to1$, $f^{(2)}_3(z)=\frac6\pi\log\frac1{\abs{z}}+o(1)$. Next, we consider the boundary behaviour. The boundary $\partial\cB_3$ consists of 
\begin{itemize}
\item the curved arc $\Gamma$: $27\abs{z}^2=(1+2\Re(z))^3$, on which $54\varrho^2=2t^3$, with $-\frac18<\Re(z)<1$;
\item the cusp points $z=-\frac18$ and $z=1$. 
\end{itemize}
We treat the curved arc first, then indicate why the two cusps are excluded. 

On $\Gamma$, we have $\varrho^2=\frac{t^3}{27}$. Then $C_z(u)=u(u-t)^2-4\frac{t^3}{27}$. We claim that $C_z$ has a double root at $u=\frac t3$. Indeed, evaluate $C_z(\frac t3)=0$, and the derivative $C'_z(u)=(u-t)(3u-t)$, so $C'_z(\frac t3)=0$. Thus $u=\frac t3$ is a double root. Since $C_z$ is cubic with leading coefficient $1$, we may write $C_z(u)=\Pa{u-\frac t3}^2(u-\gamma_b)$ for some third root $\gamma_b$. Comparing the coefficient of $u^2$ (which is $-2t$ in $u(u-t)^2=u^3-2t u^2+t^2u$), we get $-2\frac t3-\gamma_b=-2t$, implying that $\gamma_b=2t-\frac{2t}3=\frac{4t}3$. So on the boundary, $\alpha=\beta=\frac t3$, and $\gamma=\frac{4t}3$. This is the limiting case $\alpha\to\beta$: the interval $(\alpha,\beta)$ on which $C_z>0$ shrinks to a point as $z$ approaches the boundary from inside.

With $\alpha=\beta=\frac t3$ and $\gamma=\frac{4t}3$: $\kappa^2=0$. Hence $\cE_1(\kappa)\to\cE_1(0)=\frac\pi2$. We have $\beta=\frac t3$ and $\gamma-\alpha=t$, so $\sqrt{\beta(\gamma-\alpha)}=\frac{t}{\sqrt{3}}$. Therefore $\frac4{\pi\sqrt{\beta(\gamma-\alpha)}}\to\frac4\pi\cdot\frac{\sqrt{3}}t=\frac{4\sqrt{3}}{\pi t}$. Combining the above steps, we get that
$$
f^{(2)}_3(z)\to \frac{4\sqrt{3}}{\pi t}\cdot\frac\pi2 = \frac{2\sqrt{3}}t=\frac{2\sqrt{3}}{1+2\Re(z)}.
$$
This proves the desired result for boundary points on the curved arc $\Gamma$ with $-\frac18<\Re(z)<1$.
\end{proof}

The convergence is uniform for $z$ approaching a fixed boundary point $z_b\in\Gamma\backslash\set{-\frac18,1}$ from the interior, because the roots $\alpha,\beta,\gamma$ depend continuously on $z$ and $\alpha-\beta\to0$ continuously; the elliptic integral $\cE_1(\kappa)$ is continuous at $\kappa=0$.

The density is therefore bounded away from zero on the boundary, unlike each conditional density, which
blows up at the boundary of its ellipse; averaging over $u$ removes the inverse-square-root singularity. The
only singularity is the integrable logarithm at the origin, which carries no atom. For $d\geqslant3$,
Proposition~\ref{prop:conditional} and $U\sim\Beta(1,d-1)$ give the two-dimensional representation
$$
f_3^{(d)}(z)=(d-1)^2(d-2)\int_0^1\dif u(1-u)^{d-2}\int_0^1\tfrac{(1-t)^{d-3}}{t} \cF_2\Pa{\tfrac zt \big|u}\dif t,
$$
with $\cF_2(\cdot|u)$ the conditional qubit density in the proof of Theorem~\ref{th:b3density}.

\subsection{Elementary marginals}

The marginals of $\Re(\Delta_3)$ and $\Im(\Delta_3)$ do not require the elliptic integral: it is simpler to take
marginals of the conditional laws first.

\begin{thrm}\label{th:marginals}
For $d=2$, $(X,Y)=(\Re(\Delta_3),\Im(\Delta_3))$ have the marginal densities
$$
f_X(x)=\begin{cases}2\sqrt{1+8x}, & -\frac18\leqslant x\leqslant0,\\ 3-\sqrt{1+8x}, & 0\leqslant x\leqslant1,\end{cases}
\qquad\qquad
f_Y(y)=2\arccos(4|y|),\quad |y|\leqslant\tfrac14.
$$
In particular $\bbE[X]=\frac14$, $\bbE[X^2]=\frac18$, $\bbE[Y^2]=\frac1{72}$, $\bbE[|Y|]=\frac\pi{32}$ and
$\Pr\Set{X<0}=\frac16$.
\end{thrm}

\begin{proof}
Conditionally on $U=u$, $X$ is the expectation value of the Hermitian part $\frac12(\bsM_2+\bsM_2^\dagger)$ in a Haar
qubit. Its eigenvalues are $\lambda_\pm(u)=\frac12(u\pm\sqrt u)$, so by Archimedes' theorem $X$ is uniform on
$[\lambda_-(u),\lambda_+(u)]$, an interval of length $\sqrt u$. Likewise $Y$ is uniform on
$[-\frac12\sqrt{u(1-u)},\frac12\sqrt{u(1-u)}]$. With $U$ uniform and $s=\sqrt u$,
$f_X(x)=2\,\abs{\set{s\in[0,1]:s^2-s\leqslant2x\leqslant s^2+s}}$, which evaluates to the stated piecewise formula.
With $u=\sin^2\theta$, $f_Y(y)=2\,\abs{\set{\theta\in[0,\frac\pi2]:\sin2\theta>4|y|}}=2\arccos(4|y|)$. The
moments are elementary integrals.
\end{proof}

Both marginals peak at $0$, but $X$ has mean $\frac14$ and variance $\frac1{16}$ while $Y$ has mean $0$ and variance
$\frac1{72}$: the law of $\Delta_3$ is concentrated near the real axis and skewed to the right
(Figure~\ref{fig:bargmann}).

\begin{figure}[t]
\centering
\includegraphics{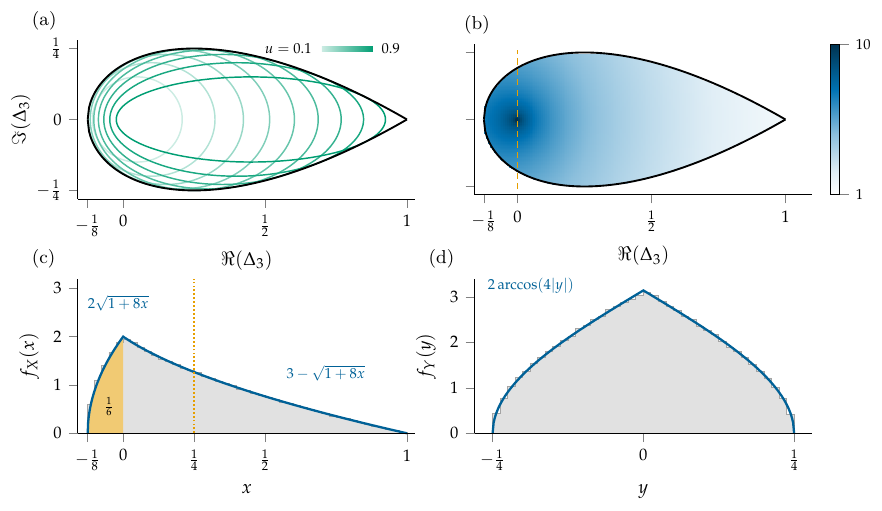}
\caption{The Bargmann invariant $\Delta_3$ of three Haar-random qubit states. (a) Conditional supports $E_u$ for
$u=|\Inner{\psi_1}{\psi_2}|^2=0.1,0.2,\ldots,0.9$ (light to dark green); their union is the support $\cB_3$ (black),
Proposition~\ref{prop:support}. (b) The exact density of Theorem~\ref{th:b3density} on a logarithmic colour
scale, with the same axes as (a). The density is largest near the origin, where it diverges like
$\frac6\pi\log\frac1{|z|}$, and stays positive up to the boundary. The dashed line separates
$\Re(\Delta_3)<0$. (c), (d) Marginals of Theorem~\ref{th:marginals} against $4\times10^5$ samples; the orange area in
(c) is $\Pr\Set{\Re(\Delta_3)<0}=\frac16$ and the dotted line is the mean $\frac14$.}\label{fig:bargmann}
\end{figure}

\subsection{Negative real part in every dimension}

If $\rho_1,\rho_2,\rho_3$ are diagonal in a common basis, $\Tr{\rho_1\rho_2\rho_3}$ is the probability that
independent measurements of the three states in that basis agree, hence a number in $[0,1]$. A Bargmann invariant
with negative real part therefore certifies basis-independent coherence of the triple
\cite{Galvao2020,Oszmaniec2024,Wagner2024}, and it does so from the real part alone. For Haar-random pure states
the probability of this event has a closed form in every dimension.

\begin{thrm}\label{th:negative}
For every $d\geqslant2$,
\begin{eqnarray}\label{eq:negative}
{\Pr}_{d} \Set{\Re(\Delta_3)<0}=\frac12-\frac{\sqrt\pi\,\Gamma(d)}{4\,\Gamma(d+\frac12)}
=\frac12-\frac{\sqrt\pi}{4\sqrt d}+O(d^{-\frac32}),
\end{eqnarray}
whose limit is $\frac12$ when $d$ approches infinity. For $d=2,3,4$ the values are $\frac16,\frac7{30},\frac{19}{70}$.
\end{thrm}

\begin{proof}
Condition on $\psi_1,\psi_2$ with $U=u\in(0,1)$, where $U=\abs{\Inner{\psi_1}{\psi_2}}^2$. By Proposition~\ref{prop:conditional}, $\Delta_3=T\Innerm{\phi}{\bsM_2}{\phi}$
with $T>0$ almost surely ($T\equiv1$ for $d=2$), so $\Re(\Delta_3)<0$ if and only if $\Re(\Innerm{\phi}{\bsM_2}{\phi})<0$. As in the
proof of Theorem~\ref{th:marginals}, $\Re(\Innerm{\phi}{\bsM_2}{\phi})$ is uniform on $[\lambda_-(u),\lambda_+(u)]=\Br{\frac{u-\sqrt{u}}2,\frac{u+\sqrt{u}}2}$, so the conditional
probability is 
\begin{eqnarray*}
{\Pr}_d\Set{\Re(\zeta)<0|U=u} = \frac{(0-\lambda_-(u))_+}{\sqrt{u}} =  \frac12(1-\sqrt u). 
\end{eqnarray*}
Finally
\begin{eqnarray*}
{\Pr}_d\Set{\Re(\zeta)<0} =\bbE[\sqrt U]=\int^1_0 \frac12(1-\sqrt u) (d-1)(1-u)^{d-2}\dif u=\frac12-\frac{\sqrt\pi\,\Gamma(d)}{4\,\Gamma(d+\frac12)}.
\end{eqnarray*}
for $U\sim\Beta(1,d-1)$. This is the exact formula \eqref{eq:negative}. The final asymptotic can be obtained by the standard asymptotic: $\frac{\Gamma(d)}{\Gamma(d+\frac12)}=d^{-\frac12}(1+O(d^{-1}))$.
\end{proof}

The probability increases with $d$ towards $\frac12$ (Figure~\ref{fig:cycles}(a)). Since $\Im(\Delta_3)\neq0$ almost
surely, the triple is coherent with probability one; \eqref{eq:negative} quantifies how often a measurement of the
real part alone detects this.

\subsection{Moments}

\begin{thrm}\label{th:moments}
For integers $p,q\geqslant0$ with $s=p+q$, and $(a)_j=a(a+1)\cdots(a+j-1)$,
\begin{eqnarray}\label{eq:mixed}
\bbE\Br{\Delta_3^{\,p}\,\overline{\Delta_3}^{\,q}}=\frac1{(d)_s^2}\sum_{j=0}^{\min(p,q)}\binom pj\binom qj\,((s-j)!)^2\,j!\,(d-1)_j.
\end{eqnarray}
\end{thrm}

\begin{proof}
Fix $\psi_1=\bse_1$ and $\psi_2=\sqrt U \bse_1+\sqrt{1-U}\bse_2$ (phases are irrelevant), and write $\psi_3=(z_1,\ldots,z_d)$.
Then $\Delta_3=U|z_1|^2+\sqrt{U(1-U)}\,\bar z_1z_2$, that is,
$$
\Delta_3=UP+\sqrt{U(1-U)}\sqrt{PQ}\,e^{\mathrm{i}\Theta},
$$
where $U\sim\Beta(1,d-1)$, $(P,Q,1-P-Q)\sim\Dir(1,1,d-2)$ (for $d=2$: $P$ uniform and $Q=1-P$), and $\Theta$ uniform
on $[0,2\pi)$, all independent. Expanding $\Delta_3^p\overline{\Delta_3}^q$ binomially, only terms with the same number $j$
of phase factors from each side survive averaging over $\Theta$, and
$$
\bbE\Br{U^{s-j}(1-U)^j}=\frac{(s-j)!\,(d-1)_j}{(d)_s},\qquad\bbE\Br{P^{s-j}Q^j}=\frac{(s-j)!\,j!}{(d)_s}.\qedhere
$$
\end{proof}

Since $\Delta_3$ is bounded, these moments determine its law. In particular
\begin{gather*}
\bbE[\Delta_3^2]=\frac4{d^2(d+1)^2},\qquad \bbE[|\Delta_3|^2]=\frac{d+3}{d^2(d+1)^2},\\
\bbE[(\Re(\Delta_3))^2]=\frac{d+7}{2d^2(d+1)^2},\qquad\bbE[(\Im(\Delta_3))^2]=\frac{d-1}{2d^2(d+1)^2}.
\end{gather*}
For cycles of any length two families of moments have closed forms.

\begin{prop}\label{prop:anym}
For $m\geqslant2$ and $k\geqslant1$, with $D_k\equiv\binom{d+k-1}k$,
$$
\bbE[\Delta_m^{\,k}]=D_k^{\,1-m},\qquad\bbE[|\Delta_m|^2]=\frac1{d^m}+\frac{d^2-1}{d^m(d+1)^m}.
$$
\end{prop}

\begin{proof}
Note that
$$
\Delta_m^k=\Tr{\psi_1^{\ot k}\cdots \psi_m^{\ot k}}\quad \text{and}\quad \bbE[\psi^{\ot k}]=\tfrac{\Pi_{\mathrm{sym}}}{D_k},
$$
where $\Pi_{\mathrm{sym}}$ projects onto
the symmetric subspace. For the second moment, condition on the first $m-1$ states and use
$$
\bbE_\psi[\abs{\Innerm{\psi}{\bsM}{\psi}}^2]=\frac{\Tr{\bsM\bsM^\dagger}+|\Tr{\bsM}|^2}{d(d+1)}
$$ 
with $\bsM=\bsM_{m-1}$: since
$\Tr{\bsM_{m-1}\bsM_{m-1}^\dagger}=|C_{m-1}|^2$ has mean $d^{2-m}$, $s_m=\bbE[|\Delta_m|^2]$ satisfies
$s_m=\frac{d^{2-m}+s_{m-1}}{d(d+1)}$ with $s_1=1$.
\end{proof}

\begin{table}[t]
\centering\small
\renewcommand{\arraystretch}{1.3}
\begin{tabular}{c|cc|cc|cc|cc}
\hline
 & \multicolumn{2}{c|}{$\bbE[\Re(\Delta_3)]$} & \multicolumn{2}{c|}{$\bbE[|\Delta_3|^2]$} & \multicolumn{2}{c|}{$\bbE[\Im(\Delta_3)^2]$}
 & \multicolumn{2}{c}{$\Pr\Set{\Re(\Delta_3)<0}$}\\
$d$ & exact & sampled & exact & sampled & exact & sampled & exact & sampled\\
\hline
2 & $\frac{1}{4}$ & \begin{tabular}{@{}c@{}}0.24995(56)\\ 0.25036(56)\end{tabular} & $\frac{5}{36}$ & \begin{tabular}{@{}c@{}}0.13890(40)\\ 0.13917(40)\end{tabular} & $\frac{1}{72}$ & \begin{tabular}{@{}c@{}}0.013892(33)\\ 0.013848(33)\end{tabular} & $\frac16$ & \begin{tabular}{@{}c@{}}0.16708(83)\\ 0.16675(83)\end{tabular} \\
3 & $\frac{1}{9}$ & \begin{tabular}{@{}c@{}}0.11110(33)\\ 0.11153(33)\end{tabular} & $\frac{1}{24}$ & \begin{tabular}{@{}c@{}}0.04169(16)\\ 0.04179(16)\end{tabular} & $\frac{1}{144}$ & \begin{tabular}{@{}c@{}}0.006931(21)\\ 0.006959(21)\end{tabular} & $\frac7{30}$ & \begin{tabular}{@{}c@{}}0.23380(95)\\ 0.23246(94)\end{tabular} \\
5 & $\frac{1}{25}$ & \begin{tabular}{@{}c@{}}0.04004(16)\\ 0.04004(16)\end{tabular} & $\frac{2}{225}$ & \begin{tabular}{@{}c@{}}0.008930(44)\\ 0.008917(44)\end{tabular} & $\frac{1}{450}$ & \begin{tabular}{@{}c@{}}0.0022288(90)\\ 0.0022147(90)\end{tabular} & $\frac{187}{630}$ & \begin{tabular}{@{}c@{}}0.2977(10)\\ 0.2969(10)\end{tabular} \\
\hline
\end{tabular}
\caption{Exact values from Theorems~\ref{th:negative} and~\ref{th:moments} and estimates from $2\times10^5$ samples
(standard errors in parentheses). Upper entry of each sampled column: direct sampling of three Haar vectors;
lower entry: the Markov chain of Theorem~\ref{th:markov}.}\label{tab:moments}
\end{table}

\subsection{A three-parameter Markov chain}

The conditional representation \eqref{eq:shadowrep} involves the matrix $\bsM_k$, but its law given the past depends
on only three real numbers.

\begin{thrm}\label{th:markov}
Let $S_k=|C_k|^2$ and $Z_k=\Delta_k$, with $S_1=Z_1=1$. Then $|Z_k|^2\leqslant S_k$, and $(S_k,Z_k)_{k\geqslant1}$ is a Markov
chain with transitions
\begin{eqnarray}\label{eq:markov}
S_{k+1}=S_kU_k,\qquad Z_{k+1}=Z_kU_k+\sqrt{S_k-|Z_k|^2}\,\sqrt{U_kV_k}\,e^{\mathrm{i}\Theta_k},
\end{eqnarray}
where $(U_k,V_k,1-U_k-V_k)\sim\op{Dirichlet}(1,1,d-2)$ for $d\geqslant 3$, while for $d=2$, $U_k$ uniform on $(0,1)$ and $V_k=1-U_k$. The innovation pairs $(U_k,V_k)$ are independent across $k$ and independent of the past; $\Theta_k$ is independent of the pair and of the past.
\end{thrm}

\begin{proof}
Choose an orthonormal basis with $\bse_1=\psi_k$ and $C_k\psi_1=Z_k\bse_1+\sqrt{S_k-|Z_k|^2}\bse_2$; this is possible because
$\Inner{\psi_k}{C_k\psi_1}=Z_k$ and $\norm{C_k\psi_1}^2=S_k$. Writing $\psi_{k+1}=\sum_iw_ie_i$, we have
$C_{k+1}=C_k\Inner{\psi_k}{\psi_{k+1}}=C_kw_1$ and
$$
Z_{k+1}=C_kw_1\Inner{\psi_{k+1}}{\psi_1}=w_1\Inner{\psi_{k+1}}{C_k\psi_1}=|w_1|^2Z_k+\sqrt{S_k-|Z_k|^2}\,w_1\bar w_2 .
$$
The vector $\psi_{k+1}$ is Haar and independent of $\psi_1,\ldots,\psi_k$, so $(|w_1|^2,|w_2|^2)$ and
$\arg w_1-\arg w_2$ have the stated law independently of the past.
\end{proof}
The chain gives an exact sampler for $\Delta_m$ with $O(m)$ scalar operations, independent of $d$; Table~\ref{tab:moments}
compares it with direct sampling. It also closes the moment hierarchy. Let
$T_k(a,p,q)=\bbE[S_k^aZ_k^p\overline{Z_k}^{\,q}]$. Expanding \eqref{eq:markov}, averaging the phase and expanding
$(S_k-|Z_k|^2)^j$ gives
\begin{multline}\label{eq:M2}
T_{k+1}(a,p,q)=\frac1{(d)_{a+p+q}}\sum_{j=0}^{\min(p,q)}\sum_{l=0}^j\binom pj\binom qj\binom jl(-1)^{j-l}\\
\times(a+p+q-j)!\,j!\;T_k(a+l,p-l,q-l),
\end{multline}
with $T_1\equiv1$. The recursion preserves the weight $2a+p+q$, so every moment of fixed order is computed from finitely
many rational numbers. It reproduces Theorem~\ref{th:moments} at $k=3$ and Proposition~\ref{prop:anym} for all $k$.

\subsection{Long cycles}

\begin{thrm}\label{th:longcycle}
Let $H_{n}=\sum_{k=1}^{n}\frac1k$ and $H^{(2)}_{n}=\sum_{k=1}^{n}\frac1{k^2}$. For $\Delta_m$ built from a single infinite
sequence of independent Haar-random vectors in $\bbC^d$,
$$
\frac1m\log|\Delta_m|\longrightarrow-\frac12H_{d-1}\quad\text{almost surely},\qquad
\frac{\log|\Delta_m|+\frac m2H_{d-1}}{\sqrt m}\Longrightarrow\mathcal N\Pa{0,\tfrac14H^{(2)}_{d-1}}.
$$
\end{thrm}

\begin{proof}
Let $U_j=|\Inner{\psi_j}{\psi_{j+1}}|^2$ and $W_m=|\Inner{\psi_m}{\psi_1}|^2$, so that
$$
\log|\Delta_m|=\frac12\sum_{j<m}\log U_j+\frac12\log W_m.
$$
Conditionally on $\psi_1,\ldots,\psi_j$, $U_j\sim\Beta(1,d-1)$;
hence $U_1,U_2,\ldots$ are independent and identically distributed. From
$$
\bbE[U^s]=\frac{\Gamma(d)\Gamma(1+s)}{\Gamma(d+s)},
$$
we get $\bbE[\log U]=-H_{d-1}$ and $\Var[\log U]=H^{(2)}_{d-1}$. Each $W_m$ is
$\Beta(1,d-1)$, so 
$$
\Pr\Set{W_m<e^{-t}}\leqslant(d-1)e^{-t}
$$
by the Borel--Cantelli lemma, $\frac1m\log W_m\to0$ almost surely,
and $\frac{\log W_m}{\sqrt m}\to0$ in probability. The strong law of large numbers, the central limit theorem and Slutsky's
lemma complete the proof.
\end{proof}

Theorem~\ref{th:longcycle} and Proposition~\ref{prop:anym} give three different decay rates. The typical modulus
is $|\Delta_m|\approx e^{-\frac12mH_{d-1}}$, the mean $\bbE[\Delta_m]=d^{1-m}$ decays at rate $\log d$, and the root mean square decays
at rate $\frac12\log d$ by Proposition~\ref{prop:anym}. Since $\frac12\log d<\frac12H_{d-1}<\log d$ for every $d\geqslant2$, the root
mean square is dominated by atypically large values, while the mean is reduced further by cancellation of phases. As
$d\to\infty$, 
$$
\frac{H_{d-1}-\log d}2\to\frac{\gamma_{\rE}}2,
$$
where $\gamma_{\rE}$ is Euler's constant.

\begin{figure}[t]
\centering
\includegraphics[width=\textwidth]{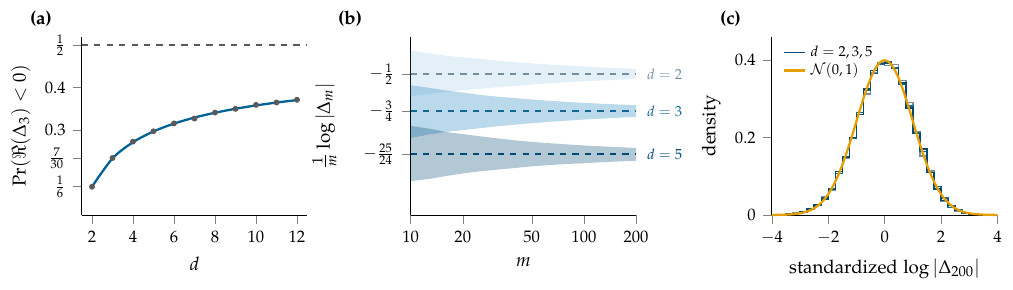}
\caption{Dependence on dimension and cycle length. (a) $\Pr\Set{\Re(\Delta_3)<0}$: exact formula \eqref{eq:negative}
(line) and Markov-chain estimates with two standard errors (dots); the dashed line is the limit $\frac12$.
(b) $\frac1m\log|\Delta_m|$ for $d=2,3,5$: mean $\pm$ one standard deviation over $4000$ sequences (bands), and
the almost-sure limits $-\frac12H_{d-1}$ (dashed). The bands shrink like $m^{-\frac12}$. (c) The standardized
$\log|\Delta_{200}|$ of Theorem~\ref{th:longcycle}, $2\times10^4$ samples for each $d$, against
$\mathcal N(0,1)$.}\label{fig:cycles}
\end{figure}

\subsection{The conditional density is elementary in every dimension}

Fix 
$$
S_k=\abs{C_k}^2=s,\quad Z_k=\Delta_k=z,\quad D=S_k-\abs{Z_k}^2=s-\abs{z}^2>0.
$$
Then 
$W:=Z_{k+1}=uz+\sqrt{(s-\abs{z}^2)uv}e^{\mathrm{i}\theta}$. Write its conditional density as
\begin{eqnarray}
\Pr\Set{W\in[\dif w]| (S_k,Z_k)=(s,z)]}=\cH_d(w;s,z)[\dif w]
\end{eqnarray}

Define $B=s-\abs{z}^2+2\Re(w\bar z)$ and $Q(u)=(s-\abs{z}^2)u(1-u)-\abs{w-uz}^2$. Expanding gives 
$$
Q(u)=-su^2+Bu-\abs{w}^2.
$$
Its discriminant is $\cD=B^2-4s\abs{w}^2$, and, when $\cD>0$, its roots are $u_\pm=\frac{B\pm\sqrt{\cD}}{2s}$. For almost every $w$ in the interior of the conditional ellipse. $0<u_-<u+<1$, and $Q(u)>0$ if and only if $u_-<u<u_+$. The support is precisely
$$
E(s,z):=\Set{w\in\bbC: \abs{w}+\abs{w-z}\leqslant \sqrt{s}}.
$$
\begin{itemize}
\item \textbf{Qubits: $d=2$.} Conditional on $U=u$, $W$ is uniform on the circle $\abs{w-uz}^2=(s-\abs{z}^2)u(1-u)$. Its planar distribution can therefore be written $\frac1\pi \delta\Pa{\abs{w-uz}^2-(s-\abs{z}^2)u(1-u)}$. Averaging over $u$ gives 
$$
\cH_2(w;s,z) = \frac1\pi \int^1_0 \delta(Q(u))\dif u.
$$
At the two roots, $\abs{Q'(u_\pm)}=\sqrt{\cD}$. Consequently, 
\begin{eqnarray}\label{eq:H2wsz}
\cH_2(w;s,z)=\frac2{\pi\sqrt{\cD}}\chi_{E^\circ(s,z)}(w).
\end{eqnarray}
Equivalently, $\cD=s(s-\abs{z}^2)(1-R^2(w))$, where $R(w)$ is the affine elliptical radius of $E(s,z)$, so this is exactly the corrected elliptical numerical-shadow density.
\item \textbf{Higher dimensions: $d\geqslant3$.} The joint density of $(U,V)$ is $(d-1)(d-2)(1-u-v)^{d-3}$, where $(u,v)\in\bbR^2_{>0}$ and $u+v<1$. For fixed $u$, use $w=uz+\sqrt{(s-\abs{z}^2)uv}e^{\mathrm{i}\theta}$. The area Jacobian is $\dif w=\frac{(s-\abs{z}^2)u}2\dif v\dif\theta$. Thus the joint density of $(U,W)$ is 
$$
\frac{(d-1)(d-2)}{\pi (s-\abs{z}^2)u}\Pa{1-u-\frac{\abs{w-uz}^2}{(s-\abs{z}^2)u}}^{d-3},
$$
on the region where the parenthesis is positive. Integrating over $u$, we obtain
\begin{eqnarray}
\cH_d(w;s,z) = \frac{(d-1)(d-2)}{\pi (s-\abs{z}^2)^{d-2}} \int^{u_+}_{u_-} \frac{Q(u)^{d-3}}{u^{d-2}}\dif u,\quad d\geqslant3,
\end{eqnarray}
inside $E_{s,z}$, and zero outside. For $d\in\bbN$, this integral is elementary: its integrand is a finite Laurent polynomial in $u$. For example, for $d=3$,
\begin{eqnarray}\label{eq:H3wsz}
\cH_3(w;s,z)=\frac2{\pi (s-\abs{z}^2)}\log \frac{u_+}{u_-} = \frac2{\pi (s-\abs{z}^2)}\log\frac{B+\sqrt{\cD}}{B-\sqrt{\cD}}.
\end{eqnarray}
For $d=4$, 
\begin{eqnarray}\label{eq:H4wsz}
\cH_4(w;s,z)=\frac6{\pi (s-\abs{z}^2)^2}\Br{-s(u_+-u_-) +B\log\frac{u_+}{u_-} + \abs{w}^2\Pa{\frac1{u_+}-\frac1{u_-}}}.
\end{eqnarray}
All these formulas hold almost everywhere; values at isolated singular points can be assigned by limits. 

Put $m=d-3$ and define 
$$
L_\ell (a,b)=\begin{cases}
\frac{b^\ell-a^\ell}{\ell},&\ell \neq0,\\
\log\frac ba,&\ell=0.
\end{cases}
$$
Expanding $Q(u)^m$ gives
\begin{eqnarray}\label{eq:Hdwsz}
\cH_d(w;s,z) = \frac{(d-1)(d-2)}{\pi (s-\abs{z}^2)^{m+1}}\sum_{a,b,c\geqslant0: a+b+c=m} \binom{m}{a,b,c} (-s)^aB^b(-\abs{w}^2)^c L_{2a+b-m}(u_-,u_+).
\end{eqnarray}
This eliminates the additional numerical-shadow integration in the paper's higher-dimensional representation.
\end{itemize}

\subsection{Exact density recursion for arbitrary $n$}

The scalar density $f^{(d)}_k$ alone does not close under iteration: the next conditional law depends on both $(S_k,Z_k)$. We therefore retain the joint law $\mu_k(\dif s,\dif z)=\cL(S_k,Z_k)$. The initial nontrivial law is particularly simple: $S_2=Z_2=U_1$ for $U_1\sim \Beta(1,d-1)$. Hence, for every bounded test function $g$, 
\begin{eqnarray}\label{eq:mun1}
\int g(s,z)\mu_2(\dif s,\dif z)=(d-1)\int^1_0 (1-u)^{d-2}g(u,u)\dif u.
\end{eqnarray}
Notice that $\mu_2$ is singular: it is supported on $s=z=u$.

Let the next state be $(s',w)$, and set $u=\frac{s'}{s}$. From the preceding Jacobian calculation,
\begin{eqnarray}\label{eq:Kds'wsz}
K_d(s',w|s,z)=\frac{(d-1)(d-2)}{\pi s (s-\abs{z}^2)u}\Pa{1-u-\frac{\abs{w-uz}^2}{(s-\abs{z}^2)u}}^{d-3},
\end{eqnarray}
provided $u\in(0,1),\abs{w-uz}^2<(s-\abs{z}^2)u(1-u)$, and $K_d=0$ otherwise. Thus
\begin{eqnarray}\label{eq:muk+1}
\mu_{k+1}(\dif s',\dif w) =\Pa{\int K_d(s',w|s,z)\mu_k(\dif s,\dif z)} \dif s'\dif w.
\end{eqnarray}

For $d=2$, the corresponding distributional kernel is
\begin{eqnarray}\label{eq:K2s'wsz}
K_2(s',w|s,z) = \frac1{\pi s}\delta\Pa{\abs{w-uz}^2 - (s-\abs{z}^2)u(1-u)},\quad u=\tfrac{s'}s\in(0,1).
\end{eqnarray}
It is a normalized transition measure, although not an ordinary density for fixed $(s,z)$.

We can now state the general result.
\begin{thrm}[Exact density representation]\label{th:recursion}
For $d\geqslant2$ and $n\geqslant3$, let $\mu_{n-1}$ be generated from \eqref{eq:mun1} by the transition kernel \eqref{eq:Kds'wsz} or \eqref{eq:K2s'wsz}. Then
\begin{eqnarray}\label{eq:fndw}
f^{(d)}_n(w) = \int \cH_d(w;s,z)\mu_{n-1}(\dif s,\dif z).
\end{eqnarray}
Here $H_d$ is the explicit conditional density \eqref{eq:H2wsz}--\eqref{eq:Hdwsz}.
\end{thrm}

\begin{proof}
Conditional on the first $n-1$ states, $\Delta_n=\Innerm{\psi_n}{(C_{n-1}\out{\psi_1}{\psi_{n-1}})}{\psi_n}$. Its conditional law depends only on $S_{n-1}=\abs{C_{n-1}}^2,Z_{n-1}=\Delta_{n-1}$, and has density $H_d$. Averaging gives \eqref{eq:fndw}. The Markov transition gives the recursion for $\mu_{n-1}$.
\end{proof}
This is a deterministic integral construction of the full probability density for every finite $n$.

\subsubsection{Three states: recovery and a higher-dimensional simplification}

Since $\mu_2$ is explicitly known,
\begin{eqnarray}\label{eq:f3dw}
f_3^{(d)}(w)=(d-1)\int_0^1(1-u)^{d-2}\cH_d(w;u,u)\dif u.
\end{eqnarray}
Set $w=x+\mathrm{i}y, r=|w|$, and $t=1+2x$. For $s=z=u$,
\begin{eqnarray*}
D=s-\abs{z}^2=u(1-u), \quad B=u(t-u),\quad\cD=u[u(u-t)^2-4r^2].
\end{eqnarray*}
Define $\cI_w=\Set{u\in(0,1):w\in E^\circ(u,u)}$, where $E^\circ(u,u)$ means the interior of $E(u,u)$. For qubits ($d=2$),
\begin{eqnarray}\label{eq:f32w}
f_3^{(2)}(w)=\frac2\pi\int_{\cI_w}\frac{\dif u}{\sqrt{u(u(u-t)^2-4r^2)}}.
\end{eqnarray}
This is the quartic integral evaluated by the complete elliptic integral in Theorem~\ref{th:b3density}. For qutrits ($d=3$), inserting \eqref{eq:H3wsz} into \eqref{eq:f3dw} yields the \emph{single-integral logarithmic formula}
\begin{eqnarray}
f_3^{(3)}(w)=\frac4\pi\int_{\cI_w}\frac{\dif u}{u}\log\frac{u(t-u)+\sqrt{u(u(u-t)^2-4r^2)}}{u(t-u)-\sqrt{u(u(u-t)^2-4r^2)}}.
\end{eqnarray}
More generally, \eqref{eq:Hdwsz} makes \eqref{eq:f3dw} a one-dimensional integral of elementary functions for every integer $d\geqslant3$, rather than the two-dimensional representation given in the document.

\subsubsection{Four states: the next explicit density}

Let $A\sim\Beta(1,d-1)$. Starting from $(S_2,Z_2)=(A,A)$, the next step gives $S_3=AU,
Z_3=AU+\sqrt{A(1-A)UV}e^{\mathrm{i}\Theta}$. Therefore
$$
f_4^{(d)}(w)=\bbE\Br{\cH_d\Pa{w;AU,AU+\sqrt{A(1-A)UV}e^{\mathrm{i}\Theta}}}.
$$
For qubits, $V=1-U$, and this becomes
\begin{eqnarray}
f_4^{(2)}(w)=\int_0^1 \dif a\int_0^1\dif u\int_0^{2\pi}\frac{d\theta}{2\pi}\cH_2\Pa{w;au,au+\sqrt{a(1-a)u(1-u)}\,e^{\mathrm{i}\theta}}.
\end{eqnarray}
For $d\geqslant3$,
\begin{eqnarray}
f_4^{(d)}(w)
&=&(d-1)^2(d-2)
\int_0^1(1-a)^{d-2}\dif a\int_0^1du\notag\\
&&\times\int_0^{1-u}(1-u-v)^{d-3}\dif v\int_0^{2\pi}\frac{d\theta}{2\pi}\cH_d(w;au,au+\sqrt{a(1-a)uv}e^{i\theta}).
\end{eqnarray}
The arbitrary-$n$ construction continues in exactly the same way through \eqref{eq:muk+1}-\eqref{eq:fndw}. The support obeys the exact recursive description
\begin{eqnarray}
\cB_n = \overline{\bigcup_{(s,z)\in\op{supp}\mu_{n-1}}E(s,z)}.
\end{eqnarray}

%================================================================%
\section{Concluding remarks}\label{sect:conclusion}
%================================================================%

The laws derived here rest on two facts: the diagonal of an induced random state is Dirichlet distributed,
and a Haar-random vector reads a $2\times2$ or rank-one operator as a projected sphere. The first turns
expectation values into Curry--Schoenberg splines with repeated knots and moment vectors into Dirichlet laws on
a simplex; it also exposes the asymmetry that separates the mode of the qutrit overlap from its mean. The second
gives the qubit uncertainty laws, including the exponents $K-\frac12$ and $K$ that count the dimension of the set of
optimal states, and it reduces Bargmann invariants of any length to a three-parameter Markov chain whose long-cycle
behaviour is governed by sums of independent logarithms of Beta variables.

We close with three open problems.
\begin{itemize}
\item \textbf{Uniqueness of the mode.} Numerically $p_3$ is unimodal with mode $0.32328\ldots$. A proof requires
the sign of $p_3'$ on all of $(0,1)$, for which the representation \eqref{eq:p3rep} on the three ranges
$(0,\frac13)$, $(\frac13,\frac12)$, $(\frac12,1)$ is a starting point. The same question for $d\geqslant4$ and for other
values of $K$ is open.
\item \textbf{Bargmann densities in higher dimension.} For $d\geqslant3$, the density of $\Delta_3$ is given by the one-dimensional integral \eqref{eq:f3dw} while for general $n$ it is given by the two-dimensional representation \eqref{eq:fndw}. Its boundary behaviour, which by Proposition~\ref{prop:local} is a positive constant for $d=2$, and its
singularity at the origin are not known for $d\geqslant3$.
\item \textbf{Noncommuting qutrit observables.} For a pair of spin-one components the moment body
$\cC_{\bsA,\bsB}$ has lower dimension than four, and the law of $(\Delta_\rho(\bsA),\Delta_\rho(\bsB))$ must be computed with respect to
the natural measure on its affine hull. An exact law for such a pair would extend Section~\ref{sect:uncertainty}
beyond qubits.
\end{itemize}

%\subsection*{Acknowledgments}

%This research is supported by Zhejiang Provincial Natural Science Foundation of China under Grant No.
%LZ23A010005 and by NSFC under Grant No.11971140.

\subsection*{Data availability statement}

No experimental data were created or analyzed in this study. The code that evaluates the closed forms, checks the
theorems numerically and generates the figures is available from the author upon request.

\subsection*{Declaration on the use of generative AI}

The author designed the study, carried out the mathematical analysis, and prepared the initial manuscript.
Generative-AI tools were used for language refinement and assistance with computational checks. All mathematical
statements and computations were independently reviewed and verified by the author, who takes full responsibility
for the content of the manuscript.

%------------------------------------------------------------%
\newpage
\appendix
\appendixpage
\addappheadtotoc
%------------------------------------------------------------%

\section{Divided difference and its properties}\label{app:divdiff}

The divided difference is a functional, a core object in approximation theory, interpolation, and numerical analysis. Let me lay out the standard definition and its key properties.
\begin{definition}
For a function $f$ and distinct real points $t_1,\ldots,t_N$, the \emph{divided difference} is defined recursively: $[t_1]f=f(t_1)$, and
\begin{eqnarray}
[t_1,\ldots, t_n]f = \frac{[t_2,\ldots,t_n]f - [t_1,\ldots,t_{n-1}]f}{t_n-t_1},\quad n\geqslant2.
\end{eqnarray}
Equivalently, it is the leading coefficient of the interpolating polynomial of degree $\leqslant N-1$ through the points $(t_i,f(t_i))$.
\end{definition}
When arguments coincide, the definition is extended by continuity using derivatives. The rule is:
$$
[\underbrace{t,\ldots,t}_{m+1}]f = \frac{f^{(m)}(t)}{m!}
$$
More generally, if a point $t$ is repeated $m+1$ times, the divided difference incorporates the $m$-th derivative at $t$. The recursion still holds in the limit, e.g.,
$$
[t,t]f=\lim_{s\to t}\frac{f(t) - f(s)}{t-s} =f'(t),\qquad [t,t,t]f=\frac{f''(t)}{2!},
$$
and so on. Mixed cases work the same way: any block of repeated arguments is handled by the corresponding derivative, and the recursion propagates.

If all arguments lie in an interval and $f$ is sufficiently smooth, then for distinct points:
\begin{eqnarray}
[t_1,\ldots,t_N]f = \int_{\bbP_{N-1}} f^{(N-1)}\Pa{t_1 + \sum^N_{i=2}(t_i-t_1)u_{i-1}}\prod^N_{k=2}\dif u_k,
\end{eqnarray}
where the integral (Hermite--Genocchi integral representation) is over the standard $(N-1)$-simplex $\bbP_{N-1}$. This ties directly back to the simplex picture from before. When arguments repeat, the formula extends by continuity (and the domain of integration degenerates accordingly). The Hermite--Genocchi integral representation shows that divided differences are essentially averages of derivatives over a simplex. 

The key properties of the divided difference are listed below:
\begin{itemize}
\item \textbf{Symmetry:} $[t_1,\ldots,t_N]f$ is invariant under permutation of the $t_i$.
\item \textbf{Linearity:} It is linear in $f$.
\item \textbf{Annihilation of polynomials:} If $f$ is a polynomial of degree $<N-1$, then $[t_1,\ldots,t_N]f=0$. For degree exactly $N-1$, it equals the leading coefficient.
\item \textbf{Mean value theorem:} For distinct $t_i\in[a,b]$ and $f\in C^{N-1}$, there exists $\xi\in(a,b)$ with 
$$
[t_1,\ldots,t_N]f = \frac{f^{(N-1)}(\xi)}{(N-1)!}.
$$
\item \textbf{Newton form:} The interpolating polynomial can be written
$$
P(\bst) = \sum^N_{n=1}[t_1,\ldots,t_n]f \prod^{n-1}_{i=1}(t-t_i).
$$
\end{itemize}

\section{Barycentric coordinates on the moment simplex}\label{app:bary}

Let $a_1<\cdots<a_d$ and define
$$
r_k=\sum^d_{j=1}a_j^kx_j, \qquad k=1,\ldots,d-1,
$$
where $r_0=\sum^d_{j=1}x_j=1$. In matrix form,
\begin{eqnarray*}
\begin{pmatrix}
r_0\\
r_1\\
\vdots\\
r_{d-1}
\end{pmatrix}
=
\begin{pmatrix}
1&1&\cdots&1\\
a_1&a_2&\cdots&a_d\\
\vdots&\vdots&&\vdots\\
a_1^{d-1}&a_2^{d-1}&\cdots&a_d^{d-1}
\end{pmatrix}
\begin{pmatrix}
x_1\\
x_2\\
\vdots\\
x_d
\end{pmatrix}.
\end{eqnarray*}
The determinant of the Vandermonde matrix is
$V_d(\bsa)=\prod_{i<j}(a_j-a_i)$. For
$\hat k=\set{1,\ldots,d}\setminus\set{k}$, define
\begin{eqnarray*}
\phi_k(\bsr) =\phi_k(r_1,\ldots,r_{d-1})= \sum_{j=0}^{d-1} (-1)^{d-1-j} e_{d-1-j}(a_{\hat k})r_j.
\end{eqnarray*}
The polynomial $p_k(t) = \prod_{\ell\neq k}(t-a_\ell)$ has the
expansion
\begin{eqnarray*}
p_k(t) = \sum_{j=0}^{d-1} (-1)^{d-1-j} e_{d-1-j}(a_{\hat k})t^j.
\end{eqnarray*}
Therefore,
\begin{eqnarray*}
\phi_k(\bsr) = \sum_{j=0}^{d-1} (-1)^{d-1-j} e_{d-1-j}(a_{\hat k})
\sum_{\ell=1}^d a^j_\ell x_\ell= \sum^d_{\ell=1} x_\ell p_k(a_\ell).
\end{eqnarray*}
Since $p_k(a_\ell)=0$ for $\ell\neq k$, and $p_k(a_k) =
\prod_{\ell\neq k}(a_k-a_{\ell})$, we obtain
\begin{eqnarray*}
x_k = \frac{\phi_k(\bsr)}{\prod_{\ell\neq k}(a_k-a_\ell)}.
\end{eqnarray*}
which is \eqref{eq:bary}. The product of the absolute values of the denominators is
\begin{eqnarray*}
\prod^d_{k=1} \Abs{\prod_{\ell\ne k}(a_k-a_\ell)} =
\abs{V_d(\bsa)}^2.
\end{eqnarray*}
For a symmetric Dirichlet distribution with parameter $K$, this
implies
\begin{eqnarray*}
\prod^d_{k=1}x^{K-1}_k =
\frac{\prod_{k=1}^d\abs{\phi_k(\bsr)}^{K-1}}{\abs{V_d(\bsa)}^{2(K-1)}}
\end{eqnarray*}
on the moment simplex, with the signs already encoded by the
nonnegative barycentric coordinates. The affine Jacobian from
$(x_1,\ldots,x_{d-1})$ to $(r_1,\ldots,r_{d-1})$ has magnitude
$\abs{V_d(\bsa)}$. Hence the fully expanded denominator in the joint
moment density is $\abs{V_d(\bsa)}^{2(K-1)}$.

\section{The modified Bessel function of the second kind}\label{app:2ndBessel}

The \emph{modified Bessel function of the second kind} is a standard special function, usually denoted $K_\nu(x)$, where $\nu$ is its order and $x>0$ is the argument. It is one of the two standard solutions to the \emph{modified Bessel differential equation}:
$$
x^2\frac{\dif^2 y}{\dif x^2} + x\frac{\dif y}{\dif x} - (x^2+\nu^2)y=0.
$$
The other independent solution is the modified Bessel function of the first kind, $I_\nu(x)$. Because $I_\nu$ and $I_{-\nu}$ are not always independent, $K_\nu$ is introduced as a second solution.

\begin{definition}
A common definition is
$$
K_\nu(x)=\frac{\pi}{2}\frac{I_{-\nu}(x)-I_\nu(x)}{\sin(\nu \pi)},
$$
for non-integer $\nu$. For integer order $n$, this is interpreted as a limit:
$$
K_n(x)=\lim_{\nu\to n} K_\nu(x).
$$
\end{definition}
There are also integral representations, for example
\begin{eqnarray}\label{appeq:Knu}
K_\nu(x)=\int_0^\infty e^{-x\cosh t}\cosh(\nu t)\dif t,
\qquad x>0.
\end{eqnarray}
This form makes it clear that $K_\nu(x)$ decays exponentially for large $x$.

\paragraph{Key properties}
\begin{itemize}
\item Parity/order symmetry: $ K_{-\nu}(x)=K_\nu(x)$.
\item Asymptotic behavior for large $x$:  $K_\nu(x)\sim \sqrt{\frac{\pi}{2x}}e^{-x}  (1+\frac{4\nu^2-1}{8x}+\cdots)$.
\item Behavior near $x=0$:   For $\nu>0$,
$K_\nu(x)\sim \frac{1}{2}\Gamma(\nu)\left(\frac{x}{2}\right)^{-\nu}$. For $\nu=0$,
$K_0(x)\sim -\ln(\tfrac{x}{2})-\gamma$, where $\gamma$ is the Euler--Mascheroni constant.
\item Recurrence relations:
\begin{eqnarray*}
K_{\nu-1}(x)-K_{\nu+1}(x)&=&-\frac{2\nu}{x}K_\nu(x),\\
K_{\nu-1}(x)+K_{\nu+1}(x)&=&-2K_\nu'(x).
\end{eqnarray*}
\item Derivative: $\frac{\dif}{\dif x}K_\nu(x)=-\frac{1}{2}(K_{\nu-1}(x)+K_{\nu+1}(x))$.
\end{itemize}

%------------------------------------------------------------%

\end{document}

%% file: qudit_preamble.tex
\usepackage{young}
\usepackage[latin1]{inputenc}
\usepackage{appendix}
\usepackage{epstopdf}
\usepackage{xcolor}
\usepackage{subfigure}
\usepackage[T1]{fontenc}
\usepackage[sc]{mathpazo}
\usepackage{amsmath}
\usepackage{amssymb}
\usepackage{graphicx,color}
\usepackage{framed}
\usepackage{multirow}
\usepackage{enumerate}
\usepackage{amsthm}% theorems and proofs
\usepackage{amsfonts,mathrsfs}
\usepackage{geometry} % allows easy specifying of the page layout
\usepackage{eepic}
\usepackage{ifthen}
\usepackage[vcentermath]{youngtab}
\usepackage[unicode=true,pdfusetitle, bookmarks=true,bookmarksnumbered=false,bookmarksopen=false, breaklinks=false,pdfborder={0 0 0},backref=false,colorlinks=false] {hyperref}
\hypersetup{
colorlinks,linkcolor=myurlcolor,citecolor=myurlcolor,urlcolor=myurlcolor}
\definecolor{myurlcolor}{rgb}{0,0,0.7}% Color highlighted for numbers in citations and theorems

\usepackage{color}

\newcommand{\blue}{\textcolor{blue}}

\newcommand{\proj}[1]{| #1\rangle\!\langle #1 |}

\usepackage{hyperref}
\hypersetup{pdfpagemode=UseNone}

\newcommand{\tinyspace}{\mspace{1mu}}

\newcommand{\op}[1]{\operatorname{#1}}

\newcommand{\abs}[1]{\left\lvert\tinyspace #1 \tinyspace\right\rvert}

\newcommand{\norm}[1]{\left\lVert\tinyspace #1 \tinyspace\right\rVert}

\renewcommand{\det}{\operatorname{det}}
\renewcommand{\t}{{\scriptscriptstyle\mathsf{T}}}

\newcommand{\spn}{\op{span}}
\newcommand{\sign}{\op{sign}}

\def\haar{\mathrm{Haar}}

\def \dif {\mathrm{d}}
\def \diag {\mathrm{diag}}

\def \conv {\mathrm{conv}}

\def\complex{\mathbb{C}}

\def\I{\mathbb{1}}

\def\bdlambda{\boldsymbol{\lambda}}

\def\bdsigma{\boldsymbol{\sigma}}

\newenvironment{mylist}[1]{\begin{list}{}{
    \setlength{\leftmargin}{#1}
    \setlength{\rightmargin}{0mm}
    \setlength{\labelsep}{2mm}
    \setlength{\labelwidth}{8mm}
    \setlength{\itemsep}{0mm}}}
    {\end{list}}

\def\ot{\otimes}

\newcommand{\inner}[2]{\langle #1 , #2\rangle}

\newcommand{\out}[2]{| #1\rangle\langle #2 |}
\newcommand{\Inner}[2]{\left\langle #1 , #2\right\rangle}
\newcommand{\Innerm}[3]{\left\langle #1 \left| #2 \right| #3 \right\rangle}

\newcommand{\Herm}{\mathrm{Herm}}

\newcommand{\pa}[1]{(#1)}
\newcommand{\Pa}[1]{\left(#1\right)}

\newcommand{\Br}[1]{\left[#1\right]}
\newcommand{\set}[1]{\{#1\}}
\newcommand{\Set}[1]{\left\{#1\right\}}

\newcommand{\ket}[1]{|#1\rangle}

\DeclareMathOperator{\trace}{Tr}
\newcommand{\ptr}[2]{\trace_{#1}\pa{#2}}
\newcommand{\Ptr}[2]{\trace_{#1}\Pa{#2}}

\newcommand{\Tr}[1]{\Ptr{}{#1}}

\newcommand{\Abs}[1]{\left|\tinyspace#1\tinyspace\right|}

\def\cB{\mathcal{B}}\def\cC{\mathcal{C}}\def\cD{\mathcal{D}}\def\cE{\mathcal{E}}
\def\cF{\mathcal{F}}\def\cH{\mathcal{H}}\def\cI{\mathcal{I}}
\def\cL{\mathcal{L}}\def\cN{\mathcal{N}}

\def\cU{\mathcal{U}}

\def\bsA{\boldsymbol{A}}\def\bsB{\boldsymbol{B}}\def\bsE{\boldsymbol{E}}
\def\bsF{\boldsymbol{F}}\def\bsG{\boldsymbol{G}}\def\bsH{\boldsymbol{H}}\def\bsJ{\boldsymbol{J}}
\def\bsM{\boldsymbol{M}}\def\bsO{\boldsymbol{O}}
\def\bsP{\boldsymbol{P}}
\def\bsU{\boldsymbol{U}}\def\bsV{\boldsymbol{V}}\def\bsW{\boldsymbol{W}}\def\bsX{\boldsymbol{X}}\def\bsY{\boldsymbol{Y}}
\def\bsZ{\boldsymbol{Z}}
\def\bsa{\boldsymbol{a}}\def\bsb{\boldsymbol{b}}\def\bse{\boldsymbol{e}}
\def\bsg{\boldsymbol{g}}

\def\bsr{\boldsymbol{r}}\def\bst{\boldsymbol{t}}
\def\bsu{\boldsymbol{u}}\def\bsv{\boldsymbol{v}}\def\bsw{\boldsymbol{w}}\def\bsx{\boldsymbol{x}}
\def\bsz{\boldsymbol{z}}

\def\rB{\mathrm{B}}\def\rD{\mathrm{D}}\def\rE{\mathrm{E}}

\def\rU{\mathrm{U}}

\def\sE{\mathscr{E}}

\def\sL{\mathscr{L}}

\def\L{\textsf{L}}

\def\bbC{\mathbb{C}}\def\bbE{\mathbb{E}}
                  
\def\bbN{\mathbb{N}}
\def\bbP{\mathbb{P}}\def\bbR{\mathbb{R}}

\newtheorem{thrm}{Theorem}[section]
\newtheorem{lem}[thrm]{Lemma}
\newtheorem{prop}[thrm]{Proposition}
\newtheorem{cor}[thrm]{Corollary}
\theoremstyle{definition}
\newtheorem{definition}[thrm]{Definition}

\newtheorem{exam}[thrm]{Example}

\numberwithin{equation}{section}

\newcounter{questionnumber}